\documentclass{article}
\usepackage{array} 
\usepackage[english]{babel}
\usepackage[utf8]{inputenc}
\usepackage{authblk}
\usepackage{multicol}
\usepackage{setspace}
\usepackage[margin=1.25in]{geometry}
\usepackage{pgfplots}
\usepgfplotslibrary{fillbetween} 
\pgfplotsset{compat=newest}
\usepackage{graphicx}
\graphicspath{ {./figures/} }
\usepackage{subcaption}
\usepackage{amsmath}
\usepackage{amsthm}
\usepackage{amssymb}
\usepackage{braket}
\usepackage{dsfont}
\usepackage{physics}
\usepackage{xcolor}
\usepackage{mathrsfs} 
\usepackage{mathtools} 
\usepackage{tikz}
\usetikzlibrary{calc}
\usetikzlibrary{arrows}
\usepackage{tcolorbox}
\usepackage{enumitem}
\newlist{steps}{enumerate}{1}
\setlist[steps, 1]{label = Step \arabic*:}
\usepackage[colorinlistoftodos,prependcaption]{todonotes}
\usepackage{booktabs}
\usepackage{thmtools,thm-restate}
\usepackage{accents}

\usepackage{pdfpages} 

\usepackage{hyperref} 
\usepackage{cleveref}
\usepackage{xurl} 
\usepackage{doi} 
\hypersetup{
colorlinks = true,
citecolor= blue,
urlcolor= blue,
linkcolor = blue
}

\usepackage{authblk}
\usepackage{algorithm}
\usepackage[noend]{algpseudocode}

\usepackage{cancel}

\usepackage[dvipsnames]{xcolor}
\usepackage[colorinlistoftodos]{todonotes}

\newtheorem{theorem}{Theorem}[section]

\newtheorem{thm}[theorem]{Theorem}

\newtheorem{rem}[theorem]{Remark}
\newtheorem{cor}[theorem]{Corollary}
\newtheorem{lem}[theorem]{Lemma}
\newtheorem{prop}[theorem]{Proposition}

\theoremstyle{definition} 
\newtheorem{definition}[theorem]{Definition}

\numberwithin{equation}{section}

\newcounter{algsubstate}


\newcommand{\HH}{\mathcal{H}}
\newcommand{\Id}{\mathrm{Id}}
\newcommand{\id}{\mathds{1}}
\newcommand{\CC}{\mathbb{C}}

\newcommand{\NN}{\mathbb{N}}

\newcommand{\RR}{\mathbb{R}}

\newcommand{\BB}{\mathcal{B}}
\newcommand{\MM}{\mathcal{M}}
\newcommand{\UU}{\mathcal{U}}
\newcommand{\EE}{\mathcal{E}}
\newcommand{\Co}{\mathcal{C}}

\newcommand{\Ss}{\mathcal{S}}
\newcommand{\PP}{\mathcal{P}}
\newcommand{\OO}{\mathcal{O}}
\newcommand{\VV}{\mathcal{V}}
\newcommand{\RE}{\mathcal{R}}
\newcommand{\Ff}{\mathcal{F}}
\newcommand{\cC}{\mathcal{C}}

\newcommand{\Meff}{\mathbf{M}^{\mathrm{eff}}}
\newcommand{\Eeff}{\mathbf{E}^{\mathrm{eff}}}
\newcommand{\overlineMeff}{\overline{\mathbf{M}}^{\mathrm{eff}}}
\newcommand{\overlineEeff}{\overline{\mathbf{E}}^{\mathrm{eff}}}

\newcommand{\Sep}{\mathrm{SEP}}

\newcommand{\probP}{\operatorname{Pr}}
\newcommand{\poly}{\mathsf{poly}}
\newcommand{\CPTP}{\textup{CPTP}}

\newcommand{\negl}{\mathsf{negl}}

\newcommand{\cone}{\mathsf{cone}}
\newcommand{\conv}{\mathsf{conv}}

\newcommand{\extr}{\mathsf{extr}}
\newcommand{\Span}{\mathsf{span}}
\newcommand{\Pos}{\mathrm{Pos}}
\newcommand{\States}{\Ss}

\newcommand{\term}[1]{\textup{\textit{#1}}} 
\newcommand{\Renyi}{R\'{e}nyi}

\newcommand{\EffCone}{\Co^{\mathbf{E}_\mathrm{eff}}_n}

\newcommand{\DualEffCone}{\Co^{\mathcal{S_{\mathrm{eff}}}}_n}

\newcommand{\CompMaxDiv}{\overset{\smash{\text{\tiny\hspace{0.3em}\raisebox{-0.3ex}{c}}}}{D}_{\mathrm{max}}}

\newcommand{\CompDiv}{ \accentset{\text{\tiny\hspace{0.1em}\raisebox{0.05ex}{c}}}{D}}

\newcommand{\CompbbDiv}{\accentset{\text{\tiny\hspace{-0.05em}\raisebox{0.05ex}{c}}}{\mathbb{D}}}

\newcommand{\CompQ}{\accentset{\text{\tiny\hspace{0.05em}\raisebox{0.05ex}{c}}}{Q}}

\newcommand{\CompTrDis}{\accentset{\text{\tiny\hspace{0em}\raisebox{0.05ex}{c}}}{\Delta}}

\newcommand{\CompFid}{\accentset{\text{\tiny\hspace{0.05em}\raisebox{0.05ex}{c}}}{F}}

\newcommand{\CompE}{\accentset{\text{\tiny\hspace{0.15em}\raisebox{0.05ex}{c}}}{E}}

\newcommand{\CompMinErrPr}{\accentset{\text{\tiny\hspace{0.05em}\raisebox{0.05ex}{c}}}{\PP}}

\newcommand{\SandRenyi}{D}

\newcommand{\eps}{\varepsilon} 

\title{Efficient Quantum Measurements:\\
Computational Max- and Measured Rényi Divergences and Applications}

\author[1]{Álvaro Yángüez\thanks{alvaro.yanguez@lip6.fr}}
\author[2]{Thomas A. Hahn}
\author[3]{Jan Kochanowski}

\affil[1]{
Sorbonne Université, CNRS, LIP6, 4 Place Jussieu, 75005 Paris, France}
\affil[2]{The Center for Quantum Science and Technology, Department of Physics of Complex Systems, Weizmann Institute of Science, Rehovot, Israel}
\affil[3]{
Inria, Télécom Paris-LTCI, Institut Polytechnique de Paris, 91120 Palaiseau, France}

\date{}

\begin{document}

\maketitle


\begin{abstract}
Quantum information processing is limited, in practice, to efficiently implementable operations. This motivates the study of quantum divergences that preserve their operational meaning while faithfully capturing these computational constraints. Using geometric, computational, and information theoretic tools, we define two new types of computational divergences, which we term \term{computational max-divergence} and \term{computational measured Rényi divergences}. 
Both are constrained by a family of efficient binary measurements, and thus useful for state discrimination tasks in the computational setting. 
We prove that, in the infinite-order limit, the computational measured Rényi divergence coincides with the computational max-divergence, mirroring the corresponding relation in the unconstrained information-theoretic setting. For the many-copy regime, we introduce regularized versions and establish a one-sided computational Stein bound on achievable hypothesis-testing exponents under efficient measurements, giving the regularized computational measured relative entropy an operational meaning.
We further define resource measures induced by our computational divergences and prove an asymptotic continuity bound for the computational measured relative entropy of resource. Focusing on entanglement, we relate our results to previously proposed computational entanglement measures and provide explicit separations from the information-theoretic setting. Together, these results provide a principled, cohesive approach towards state discrimination tasks and resource quantification under computational constraints.

\end{abstract}

\newpage
\section{Introduction}
The success of quantum information theory lies, in part, in its ability to reduce the study of a wide range of processes, e.g., hypothesis testing~\cite{ON00,ACM+06}, resource theories~\cite{BG15,CG19}, quantum thermodynamics~\cite{BHH+15, ECP10, CRF21}, and quantum cryptography~\cite{renner2005security,Devetak_2005,Wolf2021QKD,VidickWehner2023}, to robust primitive measures such as quantum divergences and entropies. Yet, as currently defined, these measures are agnostic to certain practical limitations. Quantum information processing is complexity-constrained in practice: as systems scale, transformations and measurements that are information-theoretically allowed sometimes become infeasible under realistic time and gate constraints. A simple counting argument shows that almost all states — and likewise most measurements — require exponential-size circuits to be implemented~\cite{Kni95,Aar16,JW23}. Under such efficiency constraints, a stark gap emerges between information-theoretic and feasible operations, obscuring the direct operational meaning of standard divergences and entropies in real-world settings.

This gap has led to a rich study of quantum states and transformations under computational constraints~\cite{ABF+23,ABV23,GE24,LREJ25,HBE24,BMB+24,GY25,MKN+25}, showing that highly nonclassical states can be \emph{indistinguishable} from simpler ones by any polynomial-time measurement. Similarly, many useful information-theoretic tools can often no longer be applied as they require inefficient operations.\footnote{This is assuming certain computational hardness assumptions~\cite{Kre21,KQST23}.} This includes, e.g., 
Uhlmann's theorem~\cite{Uhl76,bostanci2023unitary,metger2023stateqip,bostanci2025local}, optimal data compression rates~\cite{bostanci2023unitary,4568378,haitner23}, and optimal quantum error correction rates~\cite{iyer2013hardness,PhysRevA.99.032344,Dennis_2002,Bernstein2009,McEliece1978APK}. 

The latter two examples highlight the need for studying computationally efficient quantum information protocols. Closely related to quantum error correction, the works~\cite{ABV23, LREJ25} consider efficient quantum entanglement distillation.\footnote{The work~\cite{ABV23,LREJ25} also considers the closely related concept of computational entanglement cost.} In particular,~\cite{LREJ25} provides an efficient distillation protocol that can distill entanglement at a non-zero rate from a specific class of states. Similarly,~\cite{YKH+22,MKN+25} considers efficient data compression (amongst other tasks),  and manages to relate them to a new complexity relative entropy, which we discuss further below.

Efforts have been made to build a solid foundation for a \term{computational quantum information theory}, which can accommodate for these practical constraints.
Of particular interest are computational entropies, which capture how random data is perceived to be for a computationally bounded observer.  In~\cite{CCL+17,avidan2025quantum,avidan2025fully}, the authors define new computational quantum entropies which incorporate computational constraints while retaining the operational meaning of the underlying information-theoretic quantities, with a particular focus on cryptography.

Particularly relevant to our setting is the complexity entropy introduced in~\cite{YKH+22}, and its underlying computational divergence, termed the complexity relative entropy~\cite{MKN+25}. This latter quantity is the computational analogue to the hypothesis-testing relative entropy. 

With this connection in mind, the authors of that work further show how one can relate these quantities to the optimal rates for efficient data compression, and decoupling.\footnote{For the latter, they consider a restricted kind of decoupling which traces out part of the system and relies on a conjectured chain rule.} Moreover, it can be used to express the optimal error coefficient for quantum hypothesis testing under computational constraints.

\begin{figure}[t!]
    \centering
    \includegraphics[width=1 \textwidth]{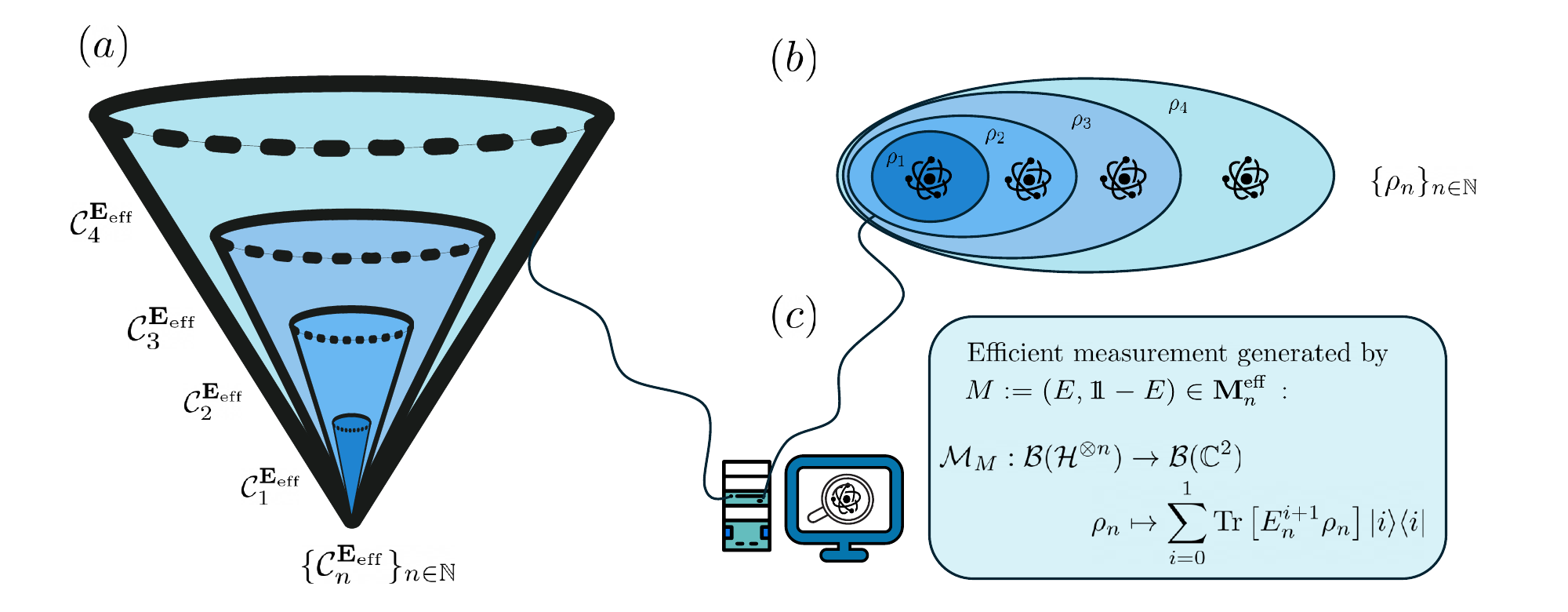}
    \caption{(a) Represents the family of cones that is induced by the set of efficiently generated quantum effect operators $\{\mathbf{E}^{\mathrm{eff}}_n\}_{n \in \NN}$. (b) A family of states $\{\rho_{n}\}_{n\in\NN}$ with growing system size. (c) Measurement maps $\mathcal{M}^{\mathrm{eff}}_n$ acting on $\rho_{n}$ (e.g., $p=\Tr[E\,\rho_{n}]$ for $E\in \mathbf{E}^{\mathrm{eff}}_n$); distances and divergences are defined by optimizing over these maps.}
    \label{fig:cones}
\end{figure}

Nevertheless, even with all of this progress, computational quantum information theory still lags far behind its information-theoretic counterpart.
In this work, we help lessen this gap by developing a principled, cohesive mathematical approach towards defining operational divergences under computational constraints.
Specifically, we introduce two new types of computational divergences: the \emph{computational max-divergence} and \emph{computational measured Rényi divergences}. We impose a practical restriction --- access only to efficient binary measurements, the basic primitive for state discrimination --- and incorporate it directly into the standard max-divergence and measured Rényi definitions.

For measured {\Renyi} divergences, this is relatively straightforward. Informally speaking, given two quantum states, $\rho$ and $\sigma$, and any {\Renyi} divergence\footnote{Since the measurement outcomes are classical, all quantum extensions of the classical {\Renyi} divergence collapse to the same expression. 
} $\SandRenyi_{\alpha}$, we aim to find the efficient (binary) measurement $\MM$ that maximizes the quantity $\SandRenyi_{\alpha} \left(\MM(\rho),\MM(\sigma) \right)$. If one views {\Renyi} divergences as a kind of distinguishing measure~\cite{T16,HP91,ON00}, this quantity essentially describes the optimal distinguishing power achievable via access only to efficient binary measurements.

To incorporate these constraints into the max-divergence, we adopt the geometric approach from~\cite{RKW11,GC24}. Concretely, we replace the standard Löwner order ``$\leq$" on positive semidefinite operators with a computationally motivated partial order ``$\leq_{\Co^{\mathcal{S_{\mathrm{eff}}}}_n}$". This partial order is induced by a cone generated via a set of efficiently implementable \term{effect operators}\footnote{We refer to an \term{effect operator} as an element $E \in \BB(\HH)$ that satisfies $0\leq E \leq \id$.} (each corresponding to a two-outcome POVM $(E, \id-E)$); see \Cref{fig:cones}. In other words, it turns out that $ \rho \leq_{\Co^{\mathcal{S_{\mathrm{eff}}}}_n} \sigma$ if $ \MM(\rho) \leq \MM(\sigma)$ for every efficient measurement. We then use this computational partial order to define a computational max-divergence, analogous to the information-theoretic case. Naturally, $\rho \leq \sigma$ implies $ \rho \leq_{\Co^{\mathcal{S_{\mathrm{eff}}}}_n} \sigma$. However, the converse is not generally true. In such cases where the converse does not hold, one should view the computational partial order as stating that even though $\sigma-\rho$ is not positive semidefinite, it appears to be when restricted to efficient binary measurements.

In addition to defining these computational divergences, our main contributions are as follows. We first show that the computational max-divergence is equal to the $\alpha=\infty$ computational measured {\Renyi} divergence in \Cref{thm:measured=conic}. This establishes a consistency between the conic and measured divergence formulations, and mirrors the corresponding identity in the information-theoretic (unconstrained) setting which states there exists a binary measurement that leaves the max-divergence invariant~\cite[Eq.~III.65]{MH23}. It also demonstrates the usefulness of conic theory to computational quantum information.
We further link these quantities to computational analogs of the \emph{trace distance} and \emph{fidelity}, which are related by a computational Fuchs–van de Graaf inequality, \Cref{lem:comp.fuchs.van.Graaf}, and derive a computational Pinsker inequality, \Cref{lem:pinsker}. These bounds allow us to show an explicit example in which our computational measures differ from their information-theoretic counterparts: by using the pseudoentanglement construction of~\cite{GE24}, we instantiate families of states whose computational fidelity and measured Rényi divergences for $\alpha \in (0,1/2)$ are negligible while their information-theoretic counterparts remain $\Theta(1)$, establishing a concrete separation.

We also demonstrate the reach of our approach through two applications. First, in the poly-copy efficient-measurement regime, we prove a computational Stein’s lemma which upper-bounds the asymmetric hypothesis-testing error exponent that is achievable via efficient binary measurements by the regularized computationally measured relative entropy (i.e.~the $\alpha = 1$ regularized computational measured {\Renyi} divergence), see \Cref{thm:qpt-stein}. This gives the latter an operational interpretation and relates this regularized quantity to the complexity relative entropy within the context of hypothesis testing. Second, we define computationally-constrained resource measures induced by our divergences and establish in \Cref{th:comptfannes} an asymptotic continuity bound for the computational measured relative entropy. This inequality allows one to essentially bound the change in the computational measured relative entropy that occurs when replacing a given state $\rho$ with $\rho^\prime$ in terms of their corresponding computational trace distance. This bound implies \term{computational faithfulness} of our quantities: if two families of states are computationally indistinguishable, then the gap in their computational measured relative entropy of resource is negligible. Lastly, in the entanglement setting, we (i) exhibit an explicit separation—there exists a family of states for which the \term{computational measured relative entropy of entanglement} is negligible in $n$ while its information–theoretic counterpart is maximal—and (ii) show that, under efficient LOCC operations, this computational measure approximately lies between the computational entanglement cost and the computational distillable entanglement of~\cite{ABV23}, thereby mirroring the information–theoretic hierarchy in the computational setting.

\paragraph{Relation to prior and concurrent work.}

Our work sits at the intersection between measured Rényi divergences which consider restricted sets of measurements and complexity–aware information theory. With regard to the first point,~\cite{RSB24,MH23} derive many properties for measured {\Renyi} divergences. In particular,~\cite{RSB24}
introduces \term{locally–measured Rényi divergences}, and optimizes classical Rényi divergences over locality–constrained POVMs. They additionally derive variational formulae that are directly related to the conic formulations which appear when defining max-divergences in a conic framework, and calculate
Stein/strong–converse exponents for symmetric data–hiding families. We essentially replace the locality constraints on the measurements with \emph{efficiency} constraints. So while the measurements that are considered differ, many operational implications remain the same. Similarly,~\cite{BHLP20} prove a quantum Stein’s lemma for several classes of \emph{restricted} measurements via an adversarial reduction; our “computational Stein's lemma” is a converse bound in the poly–copy, efficient–measurement regime and incorporates complexity constraints rather than locality/separability constraints. 

For the computational max-divergence, our cone-geometric approach is inspired by~\cite{RKW11, GC24,RSB24}. Here, they incorporate locality (or other  entanglement-related constraints) restrictions on the measurements or states directly into the underlying conic structure. Similar to our approach for the computational measured {\Renyi} divergences, we apply this framework to the computational setting.

As mentioned above, our computational divergences are close in spirit to the complexity entropy and
complexity relative entropy introduced in~\cite{YKH+22,MKN+25}. Both our work and theirs consider efficient binary measurements and both can be related to efficient hypothesis testing, which is one of the applications we consider in this work. The main difference of course is that while they wanted to define a computational hypothesis-testing relative entropy,  our goal is to define computational versions of the max-divergence and measured {Rényi} divergences, whose information-theoretic (unrestricted) counterparts have proven to be instrumental and found a plethora of applications in information theory~\cite{T16,RSB24}. We expect this to extend to the computational setting as well.

Since we study a slightly different setting, it is not immediately clear if and how our computational max-divergence can be related to the quantum computational entropies defined in~\cite{CCL+17,avidan2025quantum,avidan2025fully}. In particular, a direct comparison to the quantum unpredictability entropy~\cite{avidan2025quantum} and computational min-entropy~\cite{avidan2025fully} is hard to make. They require multi-output measurements that are constrained to act solely on quantum registers to which an adversary has access to, whereas we consider binary measurements that can act on the entire state. Nevertheless, we believe it should be possible to slightly extend or adapt our definitions to incorporate these measurement constraints. \\
\textbf{Note:} During the preparation of this manuscript, we became aware of the concurrent work~\cite{MRRLJE25} by Meyer et al. While there are some overlaps and they also consider computationally restricted hypothesis testing, many of their results complement ours and vice versa.

\vspace{0.65cm}

\noindent\textbf{Structure of the paper.}
In \Cref{Sec:preliminaries}, we review preliminary material on quantum information theory and conic geometry.
In \Cref{Sec: EffMeas}, we construct our sets of efficient measurements (\Cref{sec:comp.POVM}) and the associated cones of efficient operators used to define the computational max-divergence (\Cref{sec:proper.cone}). 
We also provide definitions for the complexity-constrained trace norm and trace distance,  and show how these typical notions from quantum cryptography can be related to our formalism (\Cref{sec:comp.trac.dist}).
In \Cref{Sec: CompQuantDiv} we introduce the computational quantum divergences: first the computational max-divergence (\Cref{Sec: CompMaxDiv}), then the computational measured Rényi divergences (\Cref{Sec: CompMeasRenDiv}), showing that they coincide with the former in the limit $\alpha \to \infty$.
This yields a notion of computational fidelity (\Cref{sec:fidelity}); combined with computational Fuchs--van de Graaf–type bounds, it leads to an explicit cryptographic separation between the computational divergences and their information-theoretic counterparts (\Cref{sec:example}).
In \Cref{Sec:applications} we present applications.
We show that the Stein exponent is upper bounded by the regularized computational measured relative entropy, giving it an operational interpretation (\Cref{sec:comp.hyp.test}).
We then develop computational quantum resource measures (\Cref{sec:comp.res.th}) and establish an asymptotic continuity bound for the computational measured relative entropy of a resource.
Finally, focusing specifically on entanglement (\Cref{sec:entanglement}), we introduce the computational measured relative entropy of entanglement and relate it to previously defined computational entanglement measures.

\section{Preliminaries}
\label{Sec:preliminaries}
\subsection{Notation}
In this work we assume all Hilbert spaces to be finite dimensional. Hilbert spaces will be denoted by $\HH$. Given a Hilbert space $\HH$, we denote the set of positive semidefinite operators acting on $\HH$ by $\Pos(\HH)$. An operator $\rho \in \Pos(\HH)$ is called a \emph{quantum state} if we have $\Tr \left[\rho \right] = 1$. The set of quantum states on $\HH$ is denoted by $\Ss(\HH)$. The set of bounded operators on $\HH$ is given by $\BB(\HH)$ and the corresponding (sub)set of Hermitian operators is denoted by $\BB^{\dagger}(\HH)$. The Hilbert-Schmidt inner product on $\BB(\HH)$ is given by $\langle a,b\rangle\coloneqq\Tr[a^*b]$, where $^*$ denotes the Hilbert-Schmidt adjoint. We will also denote the adjoint of a quantum channel $\Phi:\mathcal{B}(\HH_1)\to\mathcal{B}(\HH_2)$ by $\Phi^*:\mathcal{B}(\HH_2)\to\mathcal{B}(\HH_1)$.
For two operators $\rho, \sigma \in \Pos(\HH)$ we write $\rho \ll \sigma$ if $\ker{\sigma} \subseteq \ker{\rho}$, where $\ker{\tau} \coloneqq  \{\ket{v} \,:\, \tau \ket{v} = 0\}$. Further, we say the state $\rho$ is orthogonal to the state $\sigma$, denoted by $\rho \perp \sigma$, if $\Tr \left[\rho \sigma \right] = 0$.  

 In Landau notation, given two functions $f(n)$ and $g(n)$, we write $f(n)=o(g(n))$ if \newline$\lim_{n\rightarrow \infty} f(n)/g(n) = 0$. In the same way, $f(n)=\omega(g(n))$ if $\lim_{n\rightarrow \infty} f(n)/g(n) = \infty$. $f(n)=O(g(n))$ if there exists a constant $ C>0$ such that $\lim_{n\rightarrow \infty} f(n)/g(n) \leq C$. Similarly, $f(n)=\Omega(g(n))$ if there exists a constant $ C>0$ such that $\lim_{n\rightarrow \infty} f(n)/g(n) \geq C$. Lastly, $f(n)=\Theta(g(n))$ if both $f(n)=O(g(n))$ and $f(n)=\Omega(g(n))$. A function $f(n)$ is $\negl(n)$, i.e., negligible, if, for every fixed $c$, $f(n) = o(1/n^c)$.
 
 With $\log$ we denote the logarithm to base 2 and with $\ln$ the logarithm to base $e$.

\subsection{Information-Theoretic Concepts} \label{Sec: InfoTheorConc}

\begin{definition} [Sandwiched {\Renyi} Divergence~\cite{MLDSFT13, WWY14}]
\label{def:sandwiched.renyi}
Given any two positive semidefinite operators $\rho , \sigma\in \text{Pos}(\HH)$ with $\operatorname{Tr}\left[\rho\right] > 0$, and ${\alpha\in (0,1) \cup (1,\infty)}$, the \term{sandwiched {\Renyi} divergence} between $\rho$, $\sigma$ is given by: 
\begin{align}
			\SandRenyi_{\alpha}(\rho\|\sigma) &\coloneqq \begin{cases}
				\frac{1}{\alpha-1}\log\frac{ \norm{\sigma^{\frac{1-\alpha}{2\alpha}}\rho \sigma^{\frac{1-\alpha}{2\alpha}}}_\alpha^{\alpha}}{\Tr\left[\rho\right]} & \text{ if } \left( \alpha < 1 \land \rho \not\perp \sigma \right) \vee \rho \ll \sigma \\
				+\infty & \text{else} 
			\end{cases} \; .
		\end{align}
	\end{definition}	
In this work, the divergence of interest is the sandwiched {\Renyi} divergence. For $\alpha\geq \frac{1}{2}$, it satisfies the data-processing inequality (DPI) and, for $0\leq \alpha\leq\beta$, it is monotonic in $\alpha$, i.e. ${\SandRenyi_{\alpha}(\rho\|\sigma)\leq \SandRenyi_\beta(\rho\|\sigma)}$. Moreover, we use $Q_{\alpha}$ to denote
\begin{align}
\label{eq:Qrenyi}
    Q_{\alpha} \coloneqq \norm{\sigma^{\frac{1-\alpha}{2\alpha}}\rho \sigma^{\frac{1-\alpha}{2\alpha}}}_\alpha^{\alpha} \; .
\end{align}

\begin{definition} [Relative Entropy \& Max-divergence~\cite{D09}]
    Given any two positive semidefinite operators $\rho , \sigma\in \text{Pos}(\HH)$ with $\operatorname{Tr}\left[\rho\right] > 0$ and $\rho \ll \sigma$, the (Umegaki) \term{relative entropy} and \term{max-divergence} between $\rho$, $\sigma$ are, respectively, given by:
		\begin{align}
			D(\rho\|\sigma) &\coloneqq 
				\frac{\Tr[\rho(\log\rho-\log\sigma)]}{\Tr\left[\rho\right]} \, , \\
                D_\text{max}(\rho\|\sigma) &\coloneqq 
				\log\inf\{\lambda\in\RR|\rho\leq\lambda\sigma\} \; .
		\end{align}
\end{definition}
\begin{definition} [Fidelity]
    Given  two quantum states $\rho , \sigma\in \Ss(\HH)$, the  \term{fidelity} between $\rho$, $\sigma$ is given by:
		\begin{align}
			F(\rho,\sigma) &\coloneqq \norm{\sqrt{\rho }\sqrt{\sigma }}_{1}^{2}
				 \; .
		\end{align}
\end{definition}

We are particularly interested in the sandwiched {\Renyi} divergences for $\alpha \in \{\frac{1}{2},1,\infty \}$. 
These three cases are of particular use and one retrieves that
\begin{align*}
    \lim_{\alpha \to 1} \SandRenyi_\alpha(\rho\|\sigma) &= D(\rho\|\sigma) \, ,\\
    \lim_{\alpha\nearrow\infty} \SandRenyi_\alpha(\rho\|\sigma) &= D_\text{max}(\rho\|\sigma)\, , \\
    \SandRenyi_{\frac{1}{2}}(\rho\|\sigma) &= -\log F(\rho,\sigma) \; .
\end{align*} 
A closely related divergence is the measured {\Renyi} divergence. To properly define them, we require POVM operators which are defined as follows.
\begin{definition} [POVM] A \term{positive operator-valued measure} (POVM) is given by a tuple of positive semidefinite operators $M =( E_{1},\dots,E_{m} )$ (for some $m \in \mathbb{N}$) such that ${\sum_{k=1}^m E_{k} = \id}$. Moreover, if ${E_{i} E_{j} = \delta_{ij} E_{i}}$ for all $i,j \in \{1,\dots, m \}$, we call this a \term{projection-valued measure} (PVM).
The operators $E_i$ of a POVM are called \textit{effect operators} or \textit{POVM elements}.
\end{definition}
Equally one can associate to any $m$-outcome POVM $M=(E_1,...,E_m)$ a channel that maps quantum states to classical distributions over $\{0,\dots,m-1\}$, via
\begin{align}
    \MM_M\, : \; \rho \mapsto \sum_{i=0}^{m-1}\Tr[\rho E_{i+1}]\, \ketbra{i}{i} \,.
\end{align} 
We will sometimes drop the subscript $M$, when the POVM to which this measurement map corresponds is clear from context. 

\begin{definition}[Measured Divergence]
    Let $\mathbf{M}$ be a set of POVMs. Given any two positive semidefinite operators $\rho , \sigma\in \text{Pos}(\HH)$ with $\operatorname{Tr}\left[\rho\right] > 0$, for any divergence $\mathbb{D}$, the corresponding \term{measured divergence} between $\rho$, $\sigma$ is given by:
\begin{align}
    \mathbb{D}^{\mathbf{M}}(\rho\|\sigma) = \sup_{M \in \mathbf{M} } \mathbb{D} (\MM_M(\rho)\|\MM_M(\sigma)) \, .\label{Eq: MeasuredRD}
\end{align}
\end{definition}
We note that the RHS of Eq.~\eqref{Eq: MeasuredRD} is evaluated on classical states. If $\mathbb{D}=\SandRenyi_\alpha$, then we call $\SandRenyi_\alpha^{\mathbf{M}}(\rho\|\sigma)$ a \term{measured {\Renyi} divergence}, and the RHS reduces to the classical {\Renyi} divergence. As such, it is not strictly necessary to consider the sandwiched {\Renyi} divergence; the Petz-{\Renyi} divergence, see~\cite{T16} for a formal definition, would, e.g., also yield the same definition. Nevertheless, it seems somewhat natural to consider the sandwiched {\Renyi} divergence for the following reason. For any $\alpha \geq \frac{1}{2}$ and $\rho,\sigma \in \text{Pos}(\HH)$, if $\mathbf{M}_{n}= \mathbf{M}_{n,\textup{ALL}}$, where $\mathbf{M}_{n,\textup{ALL}}$ denotes the set of all POVMs that can act on $\rho^{\otimes n}$ (or $\sigma^{\otimes n}$), then (see, e.g.,~\cite{Mosonyi_2014})
\begin{align*}
    \SandRenyi_\alpha(\rho\|\sigma) = \lim_{n \to \infty} \frac{1}{n} \SandRenyi^{\mathbf{M}_n}_\alpha(\rho^{\otimes n}\|\sigma^{\otimes n}) \; .
\end{align*}
This naturally leads to the question of how one generally defines the regularization of measured divergences.
\begin{definition}[Regularized Measured Divergence~\cite{RSB24,MH23}]
    Let $\widetilde{\mathbf{M}} = \left(\mathbf{M}_{1},\mathbf{M}_{2}, \dots \right)$ be a family of POVM sets.  Given any two positive semidefinite operators $\rho , \sigma\in \text{Pos}(\HH)$ with $\operatorname{Tr}\left[\rho\right] > 0$,  for any divergence $\mathbb{D}$, the corresponding \term{regularized measured divergence} between $\rho$, $\sigma$ is given by:
\begin{align}
\mathbb{D}^{\widetilde{\mathbf{M}}, \infty}(\rho\|\sigma) \coloneqq \limsup_{n \rightarrow \infty } \frac{1}{n}  \mathbb{D}^{\mathbf{M}_{n}}(\rho^{\otimes n}\|\sigma^{\otimes n}) \; . \label{Eq: RegMeasuredRD}
\end{align}
\end{definition}
If the corresponding measured divergence  satisfies super-additivity, i.e.,
\begin{align*}
    \mathbb{D}^{\mathbf{M}_{k+l}}(\rho^{\otimes k+l}\|\sigma^{\otimes k+l}) \geq \mathbb{D}^{\mathbf{M}_{k}}(\rho^{\otimes k}\|\sigma^{\otimes k}) + \mathbb{D}^{\mathbf{M}_{l}}(\rho^{\otimes l}\|\sigma^{\otimes l})
\end{align*}
for any $k,l \in \NN$, then it follows from, e.g.~\cite[Lemma 4]{RSB24}, that
\begin{align*}
       \mathbb{D}^{\widetilde{\mathbf{M}}, \infty}(\rho\|\sigma) = \lim_{n \to \infty} \frac{1}{n}  \mathbb{D}^{\mathbf{M}_{n}}(\rho^{\otimes n}\|\sigma^{\otimes n}) = \sup_{n \to \infty} \frac{1}{n}  \mathbb{D}^{\mathbf{M}_{n}}(\rho^{\otimes n}\|\sigma^{\otimes n}) \; .
\end{align*}
In this work we are primarily interested in $2$-output POVMs or measurements. These are sometimes called \term{test-measured divergences}, see e.g.~\cite{MH23}. Moreover, while we provide general definitions for quantum divergences, we almost exclusively consider divergences between two quantum states $\rho,\sigma \in \Ss(\HH)$.

\subsection{Theoretic Background on Cones} \label{Sec: TheorBackCone}
In this section we summarize some basic notions and results related to cones~\cite{Bus73a,Bus73b,Eve95,RKW11,GC24}.

\begin{definition} [Convex Cone]
    Let $\mathcal{V}$ be a finite-dimensional real vector space with inner product $\langle\cdot,\cdot\rangle$. 
A \term{convex cone} $\mathcal{C}\subset \mathcal{V}$ is a subset s.t.
\begin{align*}
    a,b\in \mathcal{C} \implies \lambda a+\mu b\in\mathcal{C} \quad \forall \lambda,\mu\geq 0\, .
\end{align*}
Moreover, for any $K\subset \mathcal{V}$, we denote by $\cone(K)\coloneqq\conv(\bigcup_{\lambda\geq 0}\lambda K)$ the smallest (convex) cone that contains $K$.
\end{definition}
 A cone is called \term{closed} if it is closed with respect to the norm topology induced by the inner product of $\mathcal{V}$, and, as a consequence, includes its boundary. Following the notation of~\cite{RKW11} we further call a cone \textit{pointed} if $\mathcal{C}\cap(-\mathcal{C})=\{0\}$, where the elements of $-\mathcal{C}$ are given by
\begin{align*}
    a \in -\mathcal{C} :\hspace{-1mm}\iff -a \in \mathcal{C} \; .
\end{align*}
Similarly, a cone is \term{solid} if $\Span(\mathcal{C})=\mathcal{V}$. In finite-dimensional $\mathcal{V}$ this is equivalent to $\mathcal{C}$ having a non-empty interior $\text{int}(\cC)\neq\emptyset$.
\begin{definition} [Proper Cone]
A convex, closed, pointed, and solid cone is called a \textit{proper cone}.
\end{definition}
\begin{definition} [Cone Base]
Given some convex cone $\cC$, a convex set $B$ is called a \textit{base of the cone} $\cC$ if $\cC=\cup_{\lambda\geq 0}\lambda B$ and if every $c\in\cC\setminus\{0\}$ has a unique representation $c=\lambda b$ for some $b\in B, \lambda>0$. If $\cC$ is in addition closed and with non-empty interior, then $B$ is the \term{base of a proper cone}.
\end{definition}
One can also assign to each cone a dual cone.
\begin{definition} [Dual Cone]
Given a cone $\mathcal{C}$, its \term{dual cone} $\cC^*$ is given by
\begin{align*}
    \cC^*\coloneqq\{x\in\VV^*| \langle x,c\rangle\geq 0, \ \forall c\in\cC\} \; .
\end{align*} 
\end{definition}
One natural consequence of this definition is that if $\cC_1\subset\cC_2$, then $\cC^*_2\subset\cC_1^*$. We additionally note that the dual cone is also closed and convex. Moreover, if $\cC$ is a proper cone, then so is $\cC^*$. Lastly, in finite-dimensional vector spaces, closed convex cones are their own bi-dual cones, meaning that the corresponding dual cone is also its pre-dual cone, i.e. $(\cC^*)^*=\cC$.
Given a real vector space $\mathcal{V}$ and a proper cone $\mathcal{C} \subset \mathcal{V}$, one can define a natural partial order on $\mathcal{V}$.
\begin{definition}[Partial Order]
Let $\mathcal{V}$ be a finite-dimensional real vector space and $\mathcal{C} \subset \mathcal{V}$ be a proper cone. We say that $a\geq_{\cC}b$ if and only if $a-b\in\cC$.
\end{definition}

In the following, we will take $\VV$ to be a vector space of Hermitian matrices $\BB^{\dagger}(\HH)$, over some Hilbert space $\HH$. This vector space is isomorphic to $\RR^{d^2}$, where $d=\dim(\HH)$.
Over the space of Hermitian matrices, one of the most common examples of (proper) cones in quantum information theory is the cone of \term{positive semidefinite} operators in $\BB^\dagger(\HH)$, i.e.
\begin{align}
    \Pos(\mathcal{H})\coloneqq\{a\in\BB^\dagger(\HH) \,| \, a\geq 0\}\, .
\end{align}
Notably, this cone is self-dual in $\VV$ and we denote its induced (positive operator) partial order simply by $\geq$. This particular partial order is relevant for the definition of the max-divergence~\cite{D09,RKW11}, which satisfies
\begin{align*}
    D_\text{max}(a\|b) &=  \log\inf\{\lambda\in\RR| a\leq\lambda b\} \\
    &= \log\sup_{v\in\Pos(\mathcal{H})}\frac{\langle v,a\rangle}{\langle v,b\rangle} \; .
\end{align*}
\begin{rem} Note that in the above optimization one should interpret the $v$ as being an effect operator living in the (pre-)dual of cone of $\Pos(\HH)$, in which $a,b$ live, rather than that cone itself. However, since $\Pos(\HH)$ is self-dual with respect to the Hilbert-Schmidt inner product these cones are the same. This differentiation will become important for arbitrary cones in \Cref{Def: CMaxDiv}. 
\end{rem}
If one considers other proper cones inside of $\BB^{\dagger}(\HH)$, then one can use 
this relation between the max-divergence and the conic structure of $\Pos(\HH)$ to derive a more general notion of max-divergences.
\begin{definition} [$\cC$-Max-divergence~\cite{RKW11, GC24}] \label{Def: CMaxDiv}
    Let $\cC$ be a proper cone. For any $a,b\in\cC\setminus\{0\}$, the \term{$\cC$-max-divergence} between $a,b$ is given by
    \begin{align}
        D^{\cC}_\text{max}(a\|b)\coloneqq  \log\inf\{\lambda\in\RR| a\leq_{\cC}\lambda b\} \; .
    \end{align}
    If no $\lambda\in\RR$ exists such that $\lambda b-a\in\cC$, i.e. the infimum is over an empty set, then we set $ D^{\cC}_\text{max}(a\|b)=\infty$.
\end{definition}
We note that sometimes the notation $\sup(a/b)$ is used instead of  $\exp(D_\text{max}(a\|b))$, see e.g.~\cite{RKW11}. Moreover, there exists a dual formulation of the $\cC$-max-divergence.
\begin{lem} [Dual Expression~\cite{RKW11,GC24}] \label{lem: dualexp}
    Let $\cC$ be a proper cone. Then 
    \begin{align}
    D^{\cC}_\textup{max}(a\|b)\coloneqq \log\inf\{\lambda\in\RR| a\leq_{\cC}\lambda b\}=\log\sup_{v\in\cC^*}\frac{\langle v,a\rangle}{\langle v,b\rangle}  \; .
    \end{align}
\end{lem}
With these quantities, one can also additionally define a
\term{Hilbert's projective metric}~\cite{Bus73a,Eve95} on a proper cone $\cC$ as $ d^{\cC}_\mathrm{H}(a,b)\coloneqq D^{\cC}_\textup{max}(a\|b)+D^{\cC}_\text{max}(b\|a) $.
For the applications we are considering in this work, it will be natural to consider the following setup.
Let $\mathbb{M}\subset \VV^*$ be a closed, convex, and pointed set with a non-empty interior in the dual space, which contains an element $e\in\mathbb{M}$ such that 
\begin{align}
     E \in \mathbb{M} \implies  e-E\in\mathbb{M} \quad  \forall E \in \mathbb{M} \; . \label{eq: IncRevMeas}
\end{align}
In particular, this ensures that $\mathbb{M}$ generates a proper cone $\mathcal{C}_\mathbb{M}^*$.
\begin{definition} [Distinguishability Norm]\label{def:distinguishabilityNorm}
    Let $\mathbb{M} \subset \VV^*$ be a closed, convex, and pointed set with a non-empty interior.  Moreover let $\mathbb{M}$ contain an element $e\in\mathbb{M}$ for which Eq.~\eqref{eq: IncRevMeas} holds. For any element $v \in \VV$, we denote the \term{distinguishability norm} of $v$ by
\begin{align}
\|v\|_{(\mathbb{M})}\coloneqq \sup_{E\in\mathbb{M}}\langle 2E-e,v\rangle  \; .
\end{align}
\end{definition}
In particular, it is known that the distinguishability norm 
defines a norm on  $\VV$~\cite{RKW11}.
\begin{rem}
In quantum information theory, this norm is simply the trace norm. Here, one chooses $\mathbb{M}$ to be the set of all effect operators, i.e., $E \in \mathbb{M}$ if and only if $0 \leq E \leq \id$. 
The operational interpretation of the trace norm is given by the task of distinguishing two equiprobable quantum states $\rho$ and $\sigma$. In this case, one chooses $v = \rho -\sigma$ and $e = \id$, retrieving
\begin{align*}
    \|\rho -\sigma\|_{(\mathbb{M})} &= 2  \sup_{E\in\mathbb{M}}\Tr \left[ E \left(\rho -\sigma\right)\right] - \Tr \left[ \rho -\sigma\right] \\
    &= 2  \sup_{E\in\mathbb{M}}\Tr \left[ E \left(\rho -\sigma\right)\right] \\
    &=  \|\rho -\sigma\|_{1} \; .
\end{align*}
\end{rem}

If one additionally chooses $\mathbb{M}$ to be a subset of all effect operators, then $\mathbb{M}$ can be related to a set of binary (POVM) measurements. The corresponding \term{minimal error probability} for discriminating between two equiprobable states $\rho$ and $\sigma$ is given by
\begin{align}
    \PP^{\mathbb{M}}_{err} \left(\rho,\sigma\right)= \frac{1}{2} - \sup_{E  \in \mathbb{M}} \frac{1}{2}\left| \Tr[(\rho-\sigma)E] \right| = \frac{1}{2} - \frac{1}{4} \| \rho - \sigma \|_{\mathbb{M}} \, .
\end{align}

\section{Cone of Efficient Binary Measurements} \label{Sec: EffMeas}
 
We first study in \Cref{sec:comp.POVM} the consequences of complexity constraints for quantum measurements, defining sets of efficient two-output measurements. We then show in \Cref{sec:proper.cone} that one can naturally associate to these sets a proper cone inside of $\BB^{\dagger}(\HH)$; this is our main contribution for this section. As we then discuss in Section~\ref{Sec: CompQuantDiv}, this allows us to define a computational max-divergence via Definition~\ref{Def: CMaxDiv}. Moreover, we also discuss relations between 
our sets of efficient measurements and  corresponding notions of a \emph{computational trace norm} and \emph{computational trace distance} in \Cref{sec:comp.trac.dist}, and relate this connection to cryptographic constructions.

\subsection{Complexity-Constrained POVMs}
\label{sec:comp.POVM}
When we consider (generic) two-output measurements, we refer to POVMs $M = (E_1,E_2)$. It is known that any two-output POVM on $\HH$ is implemented by appending ancillas to the quantum system, applying a quantum circuit on the entire system, applying a projective measurement on a \textit{single} qubit in the computational basis, and tracing out the remaining subsystems. The formal way to express one specific implementation of $M$ is typically via Naimark's dilation theorem, see e.g.~\cite[Theorem 2.5.1]{H11}, which we include below for completeness and holds for arbitrary $m$-outcome POVMs.
\begin{thm}[Naimark’s Dilation Theorem in Finite Dimensions] \label{Thm: Naimark}
Any POVM $M=(E_i)_{i=1}^m$ on $\HH$ can be implemented via an isometry and a projective measurement on ancilla systems. In particular, there exists an isometry $V:\mathcal{B}^\dagger(\HH) \to \mathcal{B}^\dagger(\HH\otimes \HH_{K})$, where $\HH_{K}$ denotes a Hilbert space of dimension $m$, such that
\begin{align}
    E_k =  V^* \Bigl(\id_\HH \otimes |k\rangle\langle k|_{\HH_K} \Bigr) V \quad \forall k\in [m]. 
\end{align}
\end{thm}
\begin{rem}
   One can alternatively view $V$ as a unitary, $U$, on the space $\mathcal{B}^\dagger(\HH\otimes \HH_{K})$ which acts on states of the form $\rho \otimes \ketbra{0}{0}_{\HH_K} $.
   A simple consequence of Theorem~\ref{Thm: Naimark} is that for any state $\rho$,
   \begin{align}
       \Tr \left[ E_k \rho\right] = \Tr \left[ \Bigl(\id_\HH \otimes |k\rangle\langle k|_{\HH_K} \Bigr) \Bigl(U \left(\rho \otimes \ketbra{0}{0}_{\HH_K} \right) U^* \Bigr)\right] \; .
   \end{align}
   Since only the ancillary system is being measured, one can trace out $\HH$, thus retrieving the aforementioned interpretation of Naimark's dilation.
\end{rem}
\begin{rem}[Gate complexity of the POVM]\label{rem:POVMcomplexity}
    When we refer to the complexity of a POVM, we are referring to the complexity of the (unitary) quantum circuit $U$ that was necessary 
to implement it. We adopt the approach from, e.g.,~\cite{avidan2025quantum,avidan2025fully}, where, in other words, one considers the process of adding ancillas, tracing out subsystems, and measuring in the computational basis, to be `free' operations. To characterize the complexity of a quantum circuit, $U$, one first chooses an approximately universal gate set $\mathcal{G}$ consisting of $|\mathcal{G}|$ unique gates. 
The complexity of $U$ is equal to the minimal number of gates that are necessary to generate $U$.
\end{rem}

For binary measurements, there exists a bijection between POVMs $M = (E,\id-E)$ and effect operators $E$. We take advantage of this bijection for the remainder of this section and will in abuse of notation sometimes refer to both effect operators and POVMs as measurements, depending on context. In particular, whenever we consider the complexity of an effect operator $E$, we are referring to the complexity of the associated POVM $M$.

\begin{definition}[Polynomially Generated Sets of Effect Operators]\label{def:PGSEO}
Given an approximate universal quantum gate set $\mathcal{G}$ of fixed size  $|\mathcal{G}|<\infty$, a \term{polynomially generated set of effect operators} is a family of sets of effect operators 
\begin{align}
    \{\mathbf{E}^{\mathrm{eff}}_n\}_{n\in\mathbb{N}}, \quad \mathbf{E}^{\mathrm{eff}}_n\subseteq \{ E \in  \Pos(\HH^{\otimes n}): E \leq \id \}
\end{align}
such that there exists a polynomial $p$ with the property that any element of $E \in \mathbf{E}^{\mathrm{eff}}_n$ can be implemented using at most $p(n)$ elementary gates in the sense of \Cref{rem:POVMcomplexity}. We will further impose that if $E \in \mathbf{E}^{\mathrm{eff}}_n$, then so is $\id - E \in \mathbf{E}^{\mathrm{eff}}_n$. Lastly, given $E_1\in \mathbf{E}^{\mathrm{eff}}_{n_1}$ and $E_2\in \mathbf{E}^{\mathrm{eff}}_{n_2}$, then $E_1\otimes E_2 \in  \mathbf{E}^{\mathrm{eff}}_{n_1+n_2}$, if the gate complexities of $E_i$, say $c_i$, are such that $c_1+c_2< p(n_1+n_2)$.

If each $\mathbf{E}_n^\mathrm{eff}$ is informationally complete, i.e. it holds that $\Span(\mathbf{E}_n^\mathrm{eff})=\mathcal{B}(\HH^{\otimes n})$, then we further call this family an \term{informationally complete polynomially generated set of effect operators}.
\end{definition}
One should think of polynomially generated effect operators informally as `efficiently' implementable measurement operators.
To better understand this definition, a couple of observations and remarks are in order.
Observe that we (almost) trivially have that $0,\id \in\mathbf{E}_n^\mathrm{eff}$, since these measurements corresponds to trivially outputting $0,1$, respectively, regardless of input. A way to implement this is to generate an ancilla qubit in the state $\ket{0}$ (for $1$ one additionally applies an $X$ gate) and subsequently measuring it in the computational basis, which we assume to be free operations.
We will usually assume that the polynomial $p$ in the definition is strictly super-additive,  satisfying $p(n+m)\geq p(n)+p(m)+1$, unless otherwise mentioned.

\begin{rem}[Efficient effect operators, classical post-processing, and information completeness\label{rem:postprocessing}]
Implicit in this definition is that we assume that polynomial classical post-processing with NOT and OR operations is also efficient. To motivate that $E\in\mathbf{E}_n^\mathrm{eff} \implies \id-E \in\mathbf{E}_n^\mathrm{eff}$, we highlight that the POVMs, which are the physical implementation of the effect operators, $M_{1} = (E,\id - E)$ and $M_{2} = (\id - E,E)$ differ only by a single classical NOT post-processing step. As such, if $M_{1}$ is `efficient', then $M_{2}$ should also be considered efficient, which we enforce with this condition.

Likewise, the stability under tensor product condition (assuming the polynomial $p$ grows fast enough, which is the case if it is super-additive) is motivated by our implicit assumption that polynomially many classical OR post-processing operations are also assumed to be efficient. For example, for two efficient effect operators $E_1,E_2$ the corresponding binary measurements are $M_1 = (E_1,\id - E_1)$ and $M_2 = (E_2,\id - E_2)$, and $M_3 = (E_1 \otimes E_2,\id - E_1 \otimes E_2)$ corresponds to the tensor product measurement  $E_1\otimes E_2$. It is implementable by implementing $M_1,M_2$ in parallel using at most $p(n_1)+p(n_2)$ gates and mapping the output `$ \, 00$' to $0$ and the other three outcomes to $1$. This corresponds to a logical OR gate so is only polynomial classical post processing.
 
We point out that one way to ensure that the set of effect operators is informationally complete is if for every $n$,
\begin{align*}
    U\ketbra{k}{k}U^{*} \in \mathbf{E}_n^\mathrm{eff} \quad \forall k\in [\textup{dim}(\HH)^ n] \ \forall U \in \UU \;,
\end{align*}
where $\UU$ is the set of \emph{Weyl operators} on $\HH^{\otimes n}$~\cite[Section 4.1.2]{W18} and $|k\rangle\langle k|$ is the computational basis. For the case that $\textup{dim}(\HH) =2$, this can be reduced to the Pauli group. Therefore, it corresponds to first applying (for all $i \in \NN$) a unitary $U_{i} \in \{X,Y,Z\}$ on the $i$-th qubit. Then one has to apply a measurement from some set, where the corresponding effect operators are of the form 
\begin{align*}
    \ketbra{1^{i_1}\cdots 1^{i_n}}{1^{i_1}\cdots 1^{i_n}} \; . \label{eq: infocompletequbiteffop}
\end{align*}
One way to generate the corresponding measurement would be to first apply the gate $X^{i_1}\otimes \cdots \otimes X^{i_n}$. Subsequently, one applies a computational measurement on all $n$ qubits. Since we are only considering binary measurements, one has to then map the outcome $0 \cdots 0$ to $0$ and all other outcomes to $1$. This `$2^n$-output measurement $+$ post-processing step' is generally efficient, and can be implemented by an efficient quantum circuit and a computational measurement on a single qubit. So in effect, the informationally complete condition means that the polynomial $p(n)$ cannot be too small.
\end{rem}

Unless specified otherwise, we will for the following always assume that the polynomially generated set of effect operators is informationally complete.  Using the bijection between effect operators $E$ and binary POVMs $(E, \id-E)$, we consequently get the following definition.
\begin{definition}[Polynomially Generated Sets of Binary POVMs]\label{def:poly.binary.POVMs}
    Let $\{\mathbf{E}^{\mathrm{eff}}_n\}_{n\in\mathbb{N}}$ be an (informationally complete) polynomially generated set of effect operators. The family
    \begin{align}
   \{\Meff_n\}_{n\in\mathbb{N}}, \quad \text{where} \quad  \Meff_n \coloneqq \left\{(E, \id-E) |E\in \mathbf{E}^{\mathrm{eff}}_n\right\} \; ,
    \end{align} is the associated \textit{(informationally complete) polynomially generated set of binary POVMs}.
\end{definition}
We recall that directly from the properties of \Cref{def:PGSEO} it follows that if ${M=(E, \id-E)\in\Meff_n}$, then so is ${M^\prime}=(\id-E,E)$ and if $M_i=(E_i,\id-E_i)\in\Meff_{n_i}$ for $i=1,2$, then $(E_1\otimes E_2,\id-E_1\otimes E_2)\in \Meff_{n_1+n_2}$, if the polynomial implicit in the definition of the family $\{\Eeff_n\}_{n\in\NN}$ grows sufficiently fast.

\begin{rem}
 Our definition of efficiently generated measurements encapsulates non-uniform quantum poly time generated POVMs, which have found applications in cryptography due to their ability to model an adversary with quantum advice~\cite{AD13}. Moreover, it closely mirrors, e.g., the circuit complexity definition used in~\cite{avidan2025quantum,avidan2025fully} to define computational entropies that are useful for cryptography. As such, any divergences we define in 
Section~\ref{Sec: CompQuantDiv} via $\{\Meff_n\}_{n\in\mathbb{N}}$ may similarly be relevant for cryptographic applications.
\end{rem}

\subsection{The Proper Cone of Efficient Measurements}
\label{sec:proper.cone}
Given any family $\{\mathbf{E}^{\mathrm{eff}}_n\}_{n\in\mathbb{N}}$ as defined above, for each value $n\in\mathbb{N}$ one can naturally define the smallest convex set that contains all corresponding effect operators. The same also holds true for the corresponding family $\{\Meff_n\}_{n\in\mathbb{N}}$ of binary POVMs.

\begin{definition}[Polynomially Generated Convex Sets of Effect Operators and Binary POVMs]
Let $\{\mathbf{E}^{\mathrm{eff}}_n\}_{n\in \mathbb{N}}$ denote a (informationally complete) polynomially generated set of effect operators. The corresponding \textit{(informationally complete) polynomially generated convex set of effect operators} is given by the family 
\begin{align}
\left\{\overline{\mathbf{E}}^{\mathrm{eff}}_n \right\}_{n\in\NN}, \quad \text{where} \quad \overline{\mathbf{E}}^{\mathrm{eff}}_n \coloneqq \mathrm{conv}(\mathbf{E}^{\mathrm{eff}}_n)
= \left\{\sum_{i=1}^{k} \lambda_i E^i_n \,\Bigg|\, k \in \mathbb{N},\, \lambda_i \geq 0,\, \sum_{i=1}^k \lambda_i = 1,\, E^i_n \in \mathbf{E}^{\mathrm{eff}}_n \,  \right\} \; .
\end{align}
Similarly, for the corresponding polynomially generated set of binary POVMs $\{\Meff_n\}_{n\in \mathbb{N}}$, we define the \textit{(informationally complete) polynomially generated convex set of binary POVMs} as
\begin{align}
   \{\overlineMeff_{n}\}_{n\in\mathbb{N}}, \quad \text{where} \quad  \overlineMeff_{n} \coloneqq \left\{(E, \id-E) \Big|E\in \overline{\mathbf{E}}^{\mathrm{eff}}_n\right\} \; .
    \end{align}
\end{definition}
\begin{rem}
Note that the convex sets $\overline{\mathbf{E}}^{\mathrm{eff}}_n$ and $\overlineMeff_{n}$ are closed $\forall n\in\mathbb{N}$. This is due to the fact that both $\mathbf{E}^{\mathrm{eff}}_n$ and $\Meff_n$ are finite sets in a finite-dimensional vector space. As such, their convex hulls are simplexes and therefore closed.
\end{rem}

One key property of these sets is that every measurement in $\overlineMeff_{n}$
can be efficiently approximated. This ensures that one can also view $\overlineMeff_{n}$ as a set of efficiently approximable measurements. 
To prove this property, we require the following proposition. The proof is partly based on the techniques used in~\cite{Ba15}.

\begin{prop}\label{prop:approxcat}
For every effect operator $E\in\overline{\mathbf{E}}^\text{eff}_n$, there exists a subset of $k=\frac{5 \left(n \ln{\textup{dim}(\HH)}+1\right)}{\varepsilon^2}$ effect operators $\{E_{j_1},\dots,E_{j_k}\}\subset \mathbf{E}^{\mathrm{eff}}_n$  such that
\begin{align}
      \Bigl\|E - \frac{1}{k}\sum_{i=1}^k\,E_{j_i}\Bigr\|_{\infty} \leq \varepsilon \;.
\end{align}
\end{prop}
\begin{proof}

Let $m=|\mathbf{E}^{\mathrm{eff}}_n|$ denote the cardinality of $|\mathbf{E}^{\mathrm{eff}}_n|$ and let $E_1,\dots,E_m$ denote the elements of $\mathbf{E}^{\mathrm{eff}}_n$. By construction, there exist weights $\alpha_1,\dots,\alpha_m\geq0$ such that 
\begin{align*}
     \sum_{i=1}^{m} \alpha_i E_i = E \quad \text{and} \quad 
     \sum_{i=1}^{m} \alpha_i = 1 \; .
\end{align*}
Let $X$ be a self-adjoint random matrix valued variable which with probability $\alpha_i$ satisfies 
\begin{align*}
    X = E - E_i \; 
\end{align*}
Notably, this random variable satisfies
\begin{align*}
    \mathbb{E}[X] =0 \quad \text{and} \quad 
    \big\|\mathbb{E}[X^{\,2}]\big\|_\infty \leq 1 \; . 
\end{align*}
Consequently, for $X_1,X_2,\dots,X_k$, $k$ IID copies of $X$, the variance satisfies
\begin{align*}
\sigma^{2} \coloneqq \Big\|\sum_{i=1}^{k}\mathbb{E}[X_i^{\,2}]\Big\|_\infty\le k \; .
\end{align*}

Applying the matrix Bernstein inequality~\cite[Cor.~5.2]{MJC+14} for any $0\leq\varepsilon\leq1$, it holds that
\begin{align*}
\Pr\!\left[
    \lambda_{\max}\left(\frac{1}{k}\sum_{j=1}^{k}X_{j}\right)\geq \varepsilon
  \right] &=
      \Pr\!\left[
    \lambda_{\max}\left( \sum_{j=1}^{k}X_{j}\right)\geq k\eps
  \right]
  \\
  &\leq \left(\text{dim}(\HH)^n \right)
  \exp\!\Bigl(
    -\frac{k^2\varepsilon^{2}}{3k+2 k \varepsilon}
  \Bigr) \\
  &\leq \left(\text{dim}(\HH)^n \right)
  \exp\!\Bigl(
    -\frac{k^2\varepsilon^{2}}{5k}
  \Bigr) \\
  &= \left(\text{dim}(\HH)^n \right)
  \exp\!\Bigl(
    -\frac{k\varepsilon^{2}}{5}
  \Bigr) \\
   &= 
  \exp\!\Bigl(
   n \ln{\text{dim}(\HH)} -\frac{k\varepsilon^{2}}{5}
  \Bigr) \; ,
\end{align*}
where $\lambda_{\max}\left(\frac{1}{k}\sum_{j=1}^{k}X_{j}\right)$ denotes the largest eigenvalue of $\frac{1}{k}\sum_{j=1}^{k}X_{j}$.
Setting $k = \frac{5 \left(n \ln{\text{dim}(\HH)}+1\right)}{\varepsilon^2} $ hence ensures that 
\begin{align*}
    \Pr\!\left[
    \lambda_{\max}\left(\frac{1}{k}\sum_{j=1}^{k}X_{j}\right)\geq \varepsilon
  \right] \leq e^{-1}< \frac{1}{2} \; .
\end{align*}
By the same argument for $-X_i$ it holds that, for the same value $k$, we have
\begin{align*}
    \Pr\!\left[
    \lambda_{\min}\left(\frac{1}{k}\sum_{j=1}^{k}X_{j}\right)\leq -\varepsilon
  \right] \leq e^{-1}< \frac{1}{2} \; ,
\end{align*}
where $\lambda_{\min}\left(\frac{1}{k}\sum_{j=1}^{k}X_{j}\right)$ denotes the smallest eigenvalue of $\frac{1}{k}\sum_{j=1}^{k}X_{j}$.
Using a union bound argument, it follows from these two inequalities that
\begin{align*}
    \Pr\!\left[
    \lambda_{\min}\left(\frac{1}{k}\sum_{j=1}^{k}X_{j}\right)\geq -\varepsilon \, \land \, \lambda_{\max}\left(\frac{1}{k}\sum_{j=1}^{k}X_{j}\right)\leq \varepsilon
  \right] \geq 1- \frac{2}{e} > 0 \; .
\end{align*}
As such, there must exist values for $X_{1},\dots, X_{k}$ such that the maximal eigenvalue of $|\frac{1}{k}\sum_{j=1}^{k}X_{j}|$ is less than or equal to $\varepsilon$. For these specific values of $X_i$, we know that there exists $\{E_{j_i}\}_{i=1}^k\subset \mathbf{E}^{\mathrm{eff}}_n$ such that
\begin{align*}
    X_i = E - E_{j_i} \; .
\end{align*}
By construction, it must hold for these effect operators that 
\begin{align*}
      \Bigl\|E - \frac{1}{k}\sum_{i=1}^k\,E_{j_i}\Bigr\|_{\infty} \leq \varepsilon \; ,
\end{align*}
thus concluding the proof.
\end{proof}
We now discuss two consequences that follow from Proposition~\ref{prop:approxcat}. First, even if one chooses $\varepsilon=\frac{1}{\poly (n)}$, i.e. the allowed error to scale inverse polynomially in $n$, one can approximate any $E\in\overline{\mathbf{E}}^\text{eff}_n$ using an (in general non-unique and non-degenerate) polynomial subset of effect operators $\{E_{j_i}\}_{i=1}^k\subset \mathbf{E}^{\mathrm{eff}}_n$.
Secondly, our bound on the size of this set holds true for every effect operator $E\in\overline{\mathbf{E}}^\text{eff}_n$. This allows one to efficiently approximate every binary POVM $M\in \overlineMeff_{n}$ 
to inverse‑polynomial accuracy. 
Note that this in particular applies to $\alpha E$ for some $E\in \Eeff$ and $0\leq\alpha\leq 1$, as $0 \in\mathbf{E}_n^\mathrm{eff}$.
This consequence is formally stated in Corollary~\ref{cor: effimpconmix}.
\begin{cor} \label{cor: effimpconmix}
\label{col:multiplexed} 
Every  ${M\in \overlineMeff_{n}}$ can be implemented in poly-time up to inverse‑polynomial accuracy. 
\end{cor}
\begin{proof}
It suffices to show that any $E\in\overline{\mathbf{E}}_n^\text{eff}$ is polynomially preparable up to inverse polynomial error. For any $\widetilde{E} = \frac{1}{k} \sum_{i=1}^k E_{j_i}$, with $E_{j_i} \in \mathbf{E}^{\mathrm{eff}}_n$ and $k = \poly (n)$, we can define a circuit $D_n$:
\begin{enumerate}
    \item Given a control register, prepare a control state $|\lambda'\rangle_C = \sum_{i=1}^k \frac{1}{k} |i\rangle_C$, which requires $\OO(\log k)$ gates. 
    \item Apply the multiplexed unitary $\sum_{i=1}^k |i\rangle\langle i|_C \otimes U_i$, where each $U_i$ corresponds to the unitary operator that implements $E_{j_i}$. This circuit can be implemented by $\OO\left(\text{poly}(n) \cdot k \cdot \log(k)\right)$ gates and $\OO(\log(k))$ ancillas~\cite{SBM06,CBV+24}. 
    \item Trace out the control registers, along with all other registers that will not be measured. 
    \item Measure the remaining qubit register.
\end{enumerate}
Since the unitary operators corresponding to the effects are generated in poly-time by assumption and $k = \poly(n)$ by \Cref{prop:approxcat}, any $E\in \overline{\mathbf{E}}^{\mathrm{eff}}_n$ can be approximated up to inverse polynomial precision in polynomial time.
\end{proof}
While Proposition~\ref{prop:approxcat} 
provides an upper bound on the Schatten $\infty$-norm, also known as the operator norm, this can straightforwardly be translated into an upper bound on the difference between the corresponding measurement channels in the diamond norm, see e.g.~\cite[Definition~$6.18$]{khatri2020principles}. As discussed in Section~\ref{Sec: InfoTheorConc}, one should view any binary POVM $M=(E_1,E_2)$ as a channel
\begin{align}
    \MM_M : \BB(\HH^{\otimes n}) &\to \BB(\CC^2), \qquad
    \rho \mapsto \sum_{i=0}^{1} \Tr \left[ E_{i+1} \rho\right] \ketbra{i}{i} \; ,
\end{align}
where $\CC^2$ is a two-dimensional Hilbert space.
\begin{cor}\label{cor:diamondnormbound}
For every binary POVM $M\in\overline{\mathbf{M}}^\text{eff}_n$, there exists a subset of $k=\frac{5 \left(n \ln{\text{dim}(\HH)}+1\right)}{\varepsilon^2}$ effect operators $\{E_{j_i}\}_{i=1}^k\subset \mathbf{E}^{\mathrm{eff}}_n$ such that the measurement maps $\mathcal{M}$ corresponding to $M$ and $\mathcal{M}^\prime$ to and $M^\prime=$ $(\frac{1} {k}\sum_{i=1}^k\,E_{j_i},\id-\frac{1} {k}\sum_{i=1}^k\,E_{j_i})$ respectively,
satisfy
\begin{align}
   \big\|\MM - \MM^\prime\big\|_{\diamond} \leq2\varepsilon \; .
\end{align}
\end{cor}
The proof of this is standard and can be found in \Cref{app:diamondnormbound}. With these tools, we are now in a position to define the cone related to our sets of efficient effect operators.
\begin{definition}[Cones of Polynomially Generated Effect Operators]
    Given an (informationally complete) polynomially generated convex set of effect operators $\{\overline{\mathbf{E}}^{\mathrm{eff}}_n\}_{ n \in \NN}$, the 
\term{family of cones of (informationally complete) polynomially generated effect operators} is given by 
    \begin{align}
    \left\{\EffCone\right\}_{n\in \mathbb{N}}, \quad \text{ where } \quad  \EffCone = \bigcup_{\lambda \geq 0} \lambda\overline{\mathbf{E}}^{\mathrm{eff}}_n = {\textup{cone}\left(\mathbf{E}_n^\mathrm{eff}\right)} \; .
\end{align}    

\end{definition}

\begin{rem}
We note that this cone is the computational analog to the cones used in, e.g.,~\cite{RSB24,RKW11,GC24}. Apart from constraining the measurements to be efficiently generated, the main difference is that we consider binary measurements. Their work considers general $m$-output POVMs. In principle, we could define computational analogs to their sets of POVMs provided in~\cite[Eq.~(4.1)]{RSB24} or ~\cite[Eq.~(76)]{RKW11}. In their setting, ${E_1,\dots, E_m \in \mathbf{E}^{\mathrm{eff}}_n}$ if there exists an efficiently implementable measurement given by $M=(E_1,\dots,E_m)$. Elements in $\overline{\mathbf{E}}^{\mathrm{eff}}_n$ would still correspond to convex mixtures of elements in $\mathbf{E}^{\mathrm{eff}}_n$. However, as there is no more $1$-to-$1$ correspondence between the set of efficient effect operators $\mathbf{E}^{\mathrm{eff}}_n$ and the set of efficient measurements $\mathbf{M}^{\mathrm{eff}}_n$, \Cref{prop:approxcat} would no longer suffice to ensure that every measurement in $\overline{\mathbf{M}}^{\mathrm{eff}}_n$ is efficiently preparable. This issue becomes particularly exacerbated when $m$ is exponential in $n$, which is, e.g., the case when measuring $n$ qubits in the computational basis. Finding techniques to generalize our theory to the case for $m$-output measurements while adequately capturing the computational constraints thus remains an interesting open question.
\end{rem}

\begin{lem}
    The family of cones of informationally complete polynomially generated effect operators  $\{\EffCone\}_{ n \in \NN}$ is a family of proper cones.
\end{lem}

\begin{proof}
The cones are convex and closed because their generating sets $\overline{\mathbf{E}}^{\mathrm{eff}}_n$ are convex and closed. To see that the cones must be pointed, first note that every element $E \in \EffCone$ is positive semidefinite. As such, if $F\in \EffCone$ and $-F\in\EffCone$, it follows that $F=0$. Thus, every cone $\EffCone$ is pointed. Lastly, solidness, i.e. having a non-empty interior, follows from the information completeness condition imposed on $\overline{\mathbf{E}}^{\mathrm{eff}}_n$. To see this one can show, using a standard counter argument, e.g.,~\cite[Section 2.1.3]{BV04}, that closed convex sets with empty interior cannot span the whole space.
\end{proof}

To each $\EffCone$, one can generate its dual cone. This corresponding family of dual cones is defined below.
\begin{definition}[Family of Dual Cones]\label{def:dual_cone}
The family of \textit{dual cones} of  $\{\EffCone\}_{ n \in \NN}$ is denoted by
\begin{align}
    \{\DualEffCone\}_{n \in \NN} \coloneqq  \{\Co^{\mathbf{E}_\mathrm{eff} \,*}_n\}_{n \in \NN} = \left\{ Y \in \BB(\HH^{\otimes n}) \,\Big|\, \Tr(Y^* X) \geq 0 \ \forall X \in \EffCone \right\}_{n \in \NN}.
\end{align}
\end{definition}

\begin{rem} 
Since we work in finite dimensional vector spaces the dual cone is also the pre-dual cone, so its elements should be thought of as all `states' which are positive w.r.t. the partial order $\leq_{C_n^{S_\text{eff}}}$ induced by measurements in $\Eeff_n$. It is easy to see that the family of dual cones $\{\Co^{\mathcal{S_{\mathrm{eff}}}}_n\}_{n \in \NN}$ contains the family of cones of all positive semidefinite operators, i.e., $\Pos(\HH^{\otimes n}) \subseteq \Co^{\mathcal{S_{\mathrm{eff}}}}_n$, so in particular contains all quantum states.
\end{rem}
One interesting question that arises is whether all valid (binary) POVMs with effect operators in $\EffCone$ can be efficiently approximated. Since we can efficiently approximate convex mixtures of efficiently implementable effect operators, including $0,\id$, it follows that for any $0\leq\alpha\leq1$ one can approximate $\alpha E$ for any $E\in\mathbf{E}^\textup{eff}_n$, as mentioned above. 
The question now is, whether this also holds for $\alpha E$ when $\alpha>1$, if $\alpha E$ is part of some POVM, i.e. $E\in\overlineEeff_n$ and $\alpha E\leq\id$. 
We will now argue that for certain cases one should not expect this to be true. 

As is discussed in, e.g.,~\cite{Gold08}, one cannot efficiently implement all functions of the form
\begin{align*}
    f: \{ 0,1\}^n \to \{0,1\} \; .
\end{align*}
As such, one binary measurement that would generally be inefficient is, given a (uniformly) random function $f$, to first measure $n$ qubits in the computational basis and then map the output $x \in \{ 0,1\}^n$ to $f(x) \in \{0,1\}$.
However, as we discussed in \Cref{rem:postprocessing}, there exists a setting, i.e. a polynomial and approximately universal gate-set, in which the effect operators 
\begin{align*}
    \ketbra{x}{x}  \coloneqq \ketbra{1^{x_1}\dots 1^{x_n}}{1^{x_1}\dots 1^{x_n}} \in \mathbf{E}_n^\text{eff}
\end{align*}
are efficiently implementable. As such, the effect operator doing just that is
\begin{align*}
    E_f = \sum_{x : f(x)=0} \ketbra{x}{x} \in \mathcal{C}^{\textbf{E}_\text{eff}}_n
\end{align*}
and lies in the cone $\mathcal{C}^{\textbf{E}_\text{eff}}_n$, but would correspond to a binary measurement which is not efficient. 
Let us consider such a computationally inefficient function $f$ with $ |\{x: f(x)=0\}|=\exp\left(O(n)\right) $.
We highlight that the corresponding effect operators $E_f \notin \overline{\mathbf{E}}^{\mathrm{eff}}_n$ as they are not convex mixtures of efficient effect operators. However, there exists an $E^\prime \in \overline{\mathbf{E}}^{\mathrm{eff}}_n$ such that $E_f=\alpha E^\prime$ and ${\alpha = |\{x: f(x)=0\}|}$, by construction of $\overline{\mathbf{E}}^{\mathrm{eff}}_n$. In this case, since the average size of $|\{x: f(x)=0\}|$ scales exponentially in $n$, the eigenvalues of $E^\prime$ scale inverse exponentially and it is due to this scaling that $E^\prime$ becomes efficiently implementable.

\begin{rem}[Inefficient measurements vs. $\left\{\EffCone\right\}_{n\in \mathbb{N}}$]
Our efficient cone of effect operators is non-uniform, in the sense that its elements are prepared by families of polynomial-size circuits without a uniform generation requirement. Many natural effects and measurements nevertheless lie outside this cone due to computational hardness. Next we give explicit non-membership examples. 

Since $\Eeff_n$ is finite for each $n$, the cone $\EffCone$ is generated by finitely many rays. A rank-one projector $P_\psi\coloneqq\ket{\psi}\!\bra{\psi}$ can belong to $\EffCone$ only if one of these generators is proportional to $P_\psi$. Hence, for all but finitely many $\psi$, $P_\psi\notin \EffCone$. In addition, by a standard counting argument, with overwhelming probability a Haar–random rank-one projector $P_\psi$ requires exponential circuit size to implement~\cite{NC00,Gold08}. Moreover,~\cite{JW23} construct explicit families of states with exponential circuit lower bounds; thus the corresponding rank-one projectors are, eventually, not in $\Eeff_n$ (even up to inverse-polynomial approximation).

For QMA-complete local Hamiltonians, if one could (given the instance description) efficiently compile a polynomial-size circuit that implements a projector onto the ground space to constant precision, this would yield a BQP algorithm for a QMA-complete promise problem, implying $\mathrm{QMA}\subseteq\mathrm{BQP}$, which is considered unlikely~\cite{KKR06}. Finally, under standard pseudorandomness assumptions (against non-uniform quantum distinguishers), any measurement that separates pseudorandom states or unitaries from Haar-random distributions with non-negligible bias cannot be efficient—otherwise one could break the initial assumptions~\cite{JLS18,MPMY24,ABF+23}.
\end{rem}

\subsection{Computational Trace Norm and Trace Distance} \label{sec:comp.trac.dist}

Having constructed our bases and cones of efficiently implementable measurements one straightforwardly via \Cref{def:distinguishabilityNorm} gets corresponding distinguishability norms, which also appear in the computational quantum cryptographic setting, see e.g.~\cite{JLS18,ABF+23,BCQ23}.

Without computational constraints, where every measurement is implementable, the best possible distinguishability advantage between two quantum states $\rho$ and $\sigma$ is related to the trace distance $\| \rho - \sigma\|_1$. However, if the measurements are instead constrained to be implementable with poly-size quantum operations, the information-theoretic trace distance loses its operational meaning. In such cases, one instead considers the \term{computational distance}. This distance can essentially be recovered by a suitable \term{distinguishability norm}, see \Cref{def:distinguishabilityNorm}, which in this context we call a \term{computational trace norm}.

Let us recall that the convex hull of efficient effect operators $\{\overline{\mathbf{E}}^{\mathrm{eff}}_n\}_{n \in \NN}$ is a closed, convex, and informationally complete set, which additionally contains the trivial measurement $\id \in \overline{\mathbf{E}}^{\mathrm{eff}}_n$  and satisfies the condition
\begin{align}
    E \in \overline{\mathbf{E}}^{\mathrm{eff}}_n \implies  \id-E\in\overline{\mathbf{E}}^{\mathrm{eff}}_n  \quad \forall E \in \overline{\mathbf{E}}^{\mathrm{eff}}_n
\end{align}
for all $n \in \NN$.
For such sets, it follows from~\cite[Section~V]{RKW11} that 
\begin{align}
\|v\|_{\overlineMeff_n}\coloneqq \sup_{E\in\overline{\mathbf{E}}^{\mathrm{eff}}_n}\langle 2E-\id,v\rangle   \label{eq: normconcomtracdist}
\end{align}
defines a norm on the vector space of Hermitian operators $v \in \BB^{\dagger}(\HH^{\otimes n})$ for any $n \in \NN$. 
This naturally allows one to define a \term{computational trace norm}, where $\rho,\sigma\in \States(\HH^{\otimes n})$, $v= \rho-\sigma$, and $\Tr \left[v\right]=0$. 
\begin{definition}[Computational Trace Norm]
Given any two quantum states ${\rho,\sigma\in \States(\HH^{\otimes n})}$ and a family of informationally complete polynomially generated convex set of binary POVMs $\{\overline{\mathbf{M}}_n^\text{eff}\}_{n\in\mathbb{N}}$,  the \term{computational trace norm} between these states is given by
\begin{align}
\|\rho - \sigma\|_{\overlineMeff_n} \coloneqq 2\sup_{E\in\overline{\mathbf{E}}^{\mathrm{eff}}_n}\Tr \left[ E \left(\rho -\sigma\right)\right] \; . \label{eq: defcomptracenorm}
\end{align}
\end{definition}
For general $\rho,\sigma\in \Pos(\HH^{\otimes n})$, one gets  via Eq.~\eqref{eq: normconcomtracdist}, the more general
\begin{align}
\|\rho - \sigma\|_{\overlineMeff_n} =2\sup_{E\in\overline{\mathbf{E}}^{\mathrm{eff}}_n}\Tr \left[ E \left(\rho -\sigma\right)\right]-\Tr \left[  \rho -\sigma\right] \; .
\end{align}
\begin{rem}
  It was shown in~\cite{MWW09} (and further discussed in~\cite{RKW11}), that norms defined via the theory of proper cones (such as the one in Eq.~\eqref{eq: normconcomtracdist}) naturally satisfy the triangle inequality, positive definiteness and are absolute homogeneous, i.e.
  \begin{itemize}
  \item Triangle Inequality: for all $\rho, \sigma,\tau \in \States (\HH^{\otimes n})$
  \begin{align}
      \|\rho - \sigma\|_{\overlineMeff_n} \leq  \|\rho - \tau\|_{\overlineMeff_n} +  \|\tau -\sigma \|_{\overlineMeff_n} \; .
  \end{align}
  \item  Absolute Homogeneity: for all $\rho, \sigma \in \States (\HH^{\otimes n})$ and $\lambda \in \RR$
  \begin{align}
      \|\lambda\left(\rho - \sigma\right)\|_{\overlineMeff_n} = \left|\lambda\right|  \|\rho - \sigma \|_{\overlineMeff_n} \; .
  \end{align}
  \item Positive Definiteness: for all $\rho, \sigma \in \States (\HH^{\otimes n})$
  \begin{align}
      \|\rho - \sigma\|_{\overlineMeff_n} =0 \implies \rho = \sigma \; . 
  \end{align}
\end{itemize}
\end{rem}
One straightforward property of the computational trace norm is that the optimization appearing on the RHS of Eq.~\eqref{eq: defcomptracenorm} can be restricted to the set of efficient measurements $\mathbf{E}^{\mathrm{eff}}_n$. In other words, the optimal distinguishing measurement lies in $\Meff_{n}$.
\begin{lem} \label{lem: comptrnormEeffred}
    Given any two quantum states $\rho,\sigma\in \States(\HH^{\otimes n})$, the computational trace norm can alternatively be expressed as:
    \begin{align}
         \|\rho - \sigma\|_{\overlineMeff_n} =  \begin{cases}
2\max_{E\in\mathbf{E}^{\mathrm{eff}}_n}\Tr \left[ E \left(\rho -\sigma\right)\right]\; , \\
    \max_{M = (E_1,E_2) \in \Meff_n} \sum_{i=1}^2 \left|\Tr[E_i( \rho - \sigma)] \right| \; ,\\
     \max_{M \in \Meff_n} \|\MM_M(\rho) - \MM_M(\sigma)\|_{1}
\end{cases}
    \end{align}
\end{lem}
\begin{proof}
This follows from well-known arguments; for completeness we have included the proof in \Cref{app:comptrnormEeffred}.
\end{proof}

In quantum information literature, the trace distance is defined via the trace norm; see e.g.~\cite[Section~$6.1$]{khatri2020principles}. In the computational setting, one takes the same approach.
\begin{definition}[Computational Trace Distance]
Given any two quantum states ${\rho,\sigma\in \States(\HH^{\otimes n})}$ and a family of informationally complete polynomially generated convex set of binary POVMs $\{\overline{\mathbf{M}}_n^\text{eff}\}_{n\in\mathbb{N}}$,  the \term{computational trace distance} between these states is given by
\begin{align}
\CompTrDis \left(\rho,\sigma\right)
\coloneqq\frac{1}{2}\|\rho - \sigma\|_{\overlineMeff_n} \; . \label{eq: defcomptracedist}
\end{align}
\end{definition}
As discussed in \Cref{Sec:preliminaries} the minimum error probability for discriminating between two quantum states via efficient measurements can be related to the computational trace distance. In particular, the minimum error probability using only computationally efficient binary measurements from $\Meff_{n}$ is given by
\begin{align}
    \CompMinErrPr_{err} (\rho , \sigma) &= \frac{1}{2}\left(1- \CompTrDis \left(\rho,\sigma\right)\right)
    = \frac{1}{2} - \frac{1}{4} \| \rho - \sigma\|_{\overlineMeff_n} \; .
\end{align}
\begin{rem}[Computational trace distance and cryptography]
\label{rem:crypto} To further relate this to computational quantum cryptography, let us consider the following. Two families of quantum states $\{\rho_n\}_{n \in \NN}, \{\sigma_n\}_{n \in \NN} \in \{\States(\CC^{2^n})\}_{n \in \NN} $ are said to be \term{computationally indistinguishable} if for every quantum polynomial-time distinguisher $\mathcal{D}$ (possibly non-uniform, with quantum advice\footnote{If one is given quantum advice, i.e., additional quantum side-information, then this should be incorporated into the states $\rho_{n}, \sigma_{n}$.} $\alpha_n$),
\begin{align}
    \left| \probP [\mathcal{D}^{\alpha_n} (\rho_{n}) = 1] - \probP [\mathcal{D}^{\alpha_n} (\sigma_{n}) = 1]\right| \leq \negl(n) \,.
\end{align}
This expression is known as the \term{computational distance} and, for every $n \in \NN$, can be related to the computational trace distance via
\begin{align}
    \CompTrDis \left(\rho_n,\sigma_n\right) = \sup_{\mathcal{D} \in \Meff_n} \left| \probP [\mathcal{D}^{\alpha_n} (\rho_{n}) = 1] - \probP [\mathcal{D}^{\alpha_n} (\sigma_{n}) = 1]\right| \, .
\end{align}
Under standard quantum cryptographic assumptions (e.g., the existence of pseudorandom quantum states and pseudorandom unitaries), there are families of states that are indistinguishable by any QPT adversary. Therefore, explicit constructions of these cryptographic primitives such as pseudorandom and pseudoentangled quantum states~\cite{JLS18,ABF+23,ABV23,LREJ25} are also computational indistinguishable ---under cryptographic assumptions--- within our framework.
\end{rem}
It was shown in~\cite[Lemma 2.5]{GE24} that the computational distance is approximately sub-additive. In the following lemma, we express their result in terms of 
$ \| \cdot \|_{\overlineMeff_{n}}$  
\begin{lem}[Sub-additivity of $ \| \cdot \|_{\overlineMeff_{n}}$~\cite{GE24}]
\label{lem:comp.dist.sub}
    Given any two quantum states $\rho,\sigma\in \States(\HH^{\otimes n})$, the computational trace norm is approximately sub-additive under tensor products, i.e.,
    \begin{align}
        \| \rho^{\otimes m} - \sigma^{\otimes m}\|_{\overlineMeff_{n\cdot m}} \leq m \max_{i \in \{1,\dots, m \}}\| \rho_{i} - \sigma_{i}\|_{\overlineMeff_n} \; ,
    \end{align}
    where $\rho_{i} = \rho^{\otimes i-1}\otimes\rho \otimes\sigma^{\otimes n-i}$, $\sigma_{i} = \rho^{\otimes i-1}\otimes\sigma \otimes\sigma^{\otimes n-i}$.
\end{lem}

\begin{proof}
 This follows directly from the triangle inequality property of the computational trace norm. For completeness we provide the proof in \Cref{app:comp.dist.sub}, using the same steps as~\cite[Lemma 2.5]{GE24}.
\end{proof}

The $\rho_{i},\sigma_{i} $ cannot be viewed as simply holding a single copy of $\rho,\sigma$. For both states one has the same additional registers which hold multiple copies of $\rho$ and $\sigma$. These states have a computational cost associated to them, so there is a difference between having access to them for free and needing to first generate them. The triangle inequality can of course also be used to prove
     \begin{align}
        \| \rho_{1}\otimes\rho_{2} - \sigma_{1}\otimes\sigma_{2}\|_{\overlineMeff_{n}} \leq \min   \begin{cases}
    \| \rho_{1}\otimes \left(\rho_{2} - \sigma_{2} \right)\|_{\overlineMeff_{n}} +  \| \left(\rho_{1}-\sigma_{1} \right)\otimes\sigma_{2} \|_{\overlineMeff_{n}} \\
     \| \sigma_{1}\otimes \left(\rho_{2} - \sigma_{2} \right)\|_{\overlineMeff_{n}} +  \| \left(\rho_{1}-\sigma_{1} \right)\otimes\rho_{2} \|_{\overlineMeff_{n}}
\end{cases}
        \; .
    \end{align}

\section{Computational Quantum Divergences} \label{Sec: CompQuantDiv}
We introduce two new computational quantum divergences in this section and give them operational meaning. In Section~\ref{Sec: CompMaxDiv}, we first present a fully quantum computational max-divergence defined via our cone constructions. Here, we utilize the fact that the information-theoretic max-divergence can be expressed as an optimization over elements in a cone. In Section~\ref{Sec: EffMeas}, we constructed a cone of efficient binary measurements. This will effectively be the computational analog to the cone of positive semidefinite operators used to define the information-theoretic max-divergence, which is retrieved in the asymptotic circuit-depth limit. 
In Section~\ref{Sec: CompMeasRenDiv}, we then derive a family of computational measured {\Renyi} divergences. Measured {\Renyi} divergences are becoming well-studied objects in QIT literature, see e.g.~\cite{fang2025variational,RSB24,MH23}. In~\cite{fang2025variational} they were, among others, critical to deriving a generalized quantum asymptotic equipartition theorem (QAEP). This generalized QAEP can be used to derive a variety of statements, including a generalized Quantum Stein's lemma. Our computational measured divergences can naturally be viewed as computational versions of these information-theoretic measured divergences.

Importantly, in \Cref{thm:measured=conic}, we recover the computational max-divergence as a special case ($\alpha \to \infty $) of the computational measured {\Renyi} divergences, which can be viewed as one of the main results of this section. This gives the conic max-divergence an operational interpretation and naturally mimics the information-theoretic result of~\cite[Eq.~III.65]{MH23}, where it is shown that
\begin{align}
    D^{\mathbf{M}}_{\max}(\rho\|\sigma) = D_{\max}(\rho\|\sigma)
\end{align}
if $\mathbf{M} $ is the set of all binary POVMs.\footnote{It is also closely related to~\cite[Appendix A]{Mosonyi_2014}, where they show a slightly less general result, in that $\mathbf{M} $ is now the set of all POVMs, not just binary measurements.}
Lastly, following the approach of~\cite{RKW11}, we consider in \Cref{sec:fidelity} a \emph{computational (binary) fidelity} and study its relation to the other quantities, including a straightforward equivalence to the computational measured Rényi divergence at $\alpha = 1/2$ and a computational analogue of the Fuchs-van der Graaf inequality. Lastly, we study an explicit separation between computational measured Rényi divergences when $\alpha \in (0,1/2]$ and their informational counterparts by using the pseudoentanglement construction proposed by~\cite{GE24}.

\subsection{Computational Max-Divergence} \label{Sec: CompMaxDiv}
In the rest of this paper, let $\{\Eeff_n\}_{n\in\NN}$ be an informationally complete polynomially generated set of effect operators, $\{\Meff_n\}_{n\in\NN}$ its corresponding informationally complete polynomially generated set of binary POVMs, and its associated family of proper cones $\{\mathcal{C}_n^{\Eeff}\}_{n\in\NN}$ .
For each proper cone of efficient effect operators $\EffCone$, one can define a corresponding max-divergence using \Cref{Def: CMaxDiv}. For any such cone, we call the corresponding divergence, the \term{computational max-divergence}. We recall that this family of cones implicitly depends on a polynomial $n\mapsto p(n)$, which we sometimes state explicitly when it becomes relevant for clarity, however, most of the time we omit it to simplify notation. 

\begin{definition}[Computational Max-Divergence]\label{def:ComConeMax}
Given a family of proper cones $\{\EffCone\}_{n\in\NN}$ and any two positive semidefinite operators $\rho,\sigma \in \text{Pos}(\HH^{\otimes n})$
with $\operatorname{Tr}\left[\rho\right] > 0$ and $\rho \ll \sigma$, we denote the corresponding \term{computational max-divergence} by:
\begin{align}\label{equ:def:ConeDmax}
    \CompMaxDiv(\rho \|\sigma) \coloneqq \log \inf\{\lambda \in \RR | \rho \leq_{\Co^{\mathcal{S_{\mathrm{eff}}}}_n} \lambda \sigma\}  =\log \sup_{v \in \EffCone} \frac{\Tr[ v\rho ]}{\Tr[ v \sigma ]} \; ,
\end{align}
where we recall that $\Co^{\mathcal{S_{\mathrm{eff}}}}_n$ is the dual cone to $\EffCone$ and thus also proper.
\end{definition}
We will occasionally also denote the computational max-divergence by $\CompMaxDiv^{p(n)}$ to highlight the underlying polynomial $p(n)$ that is used in the definition of the family $\{\EffCone\}_{n\in\NN}$. The equality in Eq.~\eqref{equ:def:ConeDmax} holds whenever $\EffCone$ is a proper cone, see~\cite[Eqs.~4-6]{RKW11} and~\cite[Proposition 10]{GC24}.

\begin{rem}
Since we construct the cone of efficient effects  $\mathcal{C}_n^{\textbf{E}_\text{eff}}$ rather implicitly, depending on some fixed polynomial and approximately universal gate set, existence of a tractable membership or linear-minimization oracle is unclear. Projection or Frank–Wolfe–type methods would require such an oracle to solve the membership problem. Hence we cannot provide a general algorithm for approximating the computational divergences that are based on these cones and introduced in the following section. Designing relaxations or hierarchies that admit verifiable oracles is an interesting open direction.
\end{rem}
\subsubsection{Super- and Sub-additivity Properties} \label{sec: ConMaxDivAdditivity}
Recall that the elements of $\mathbf{E}^{\mathrm{eff}}_n$ are given by \textit{all} POVM elements which can be generated with circuit size of (at most) $p(n)$ and possible polynomial classical post-processing with NOT and OR operations. So if we assume this polynomial to be strictly super-additive, i.e. $p(n+m)\geq p(n)+p(m)+1$, then the family $\{\mathbf{E}^{\mathrm{eff}}_n\}_{n\in\NN}$ and the family of its induced cones $\{\EffCone\}_{n\in\NN}$ are stable under tensor products, in the sense that if $E_i\in \Eeff_{n_i}$, then $E_1\otimes E_2\in \Eeff_{n_1+n_2}$ and the same holds for $\EffCone$.
With this in mind we are now ready to study some of the super-additivity and sub-additivity properties of the computational max-divergence, where we follow the proof techniques from~\cite{RLW24}.
\begin{lem}[Super-additivity of $\CompMaxDiv$]
\label{lem:superadd.max}
    The computational max-divergence $\CompMaxDiv(\rho \| \sigma)$ is  super-additive, i.e. for any $\rho = \rho_{1} \otimes \rho_{2}$, $\sigma = \sigma_{1} \otimes \sigma_{2} $ such that $\rho_{i},\sigma_{i}\in \Pos (\HH^{\otimes n_i})$,
    \begin{align}
         \CompMaxDiv^{p(n_1 + n_2)}(\rho \|\sigma) \geq \CompMaxDiv^{p(n_1)}(\rho_{1} \|\sigma_{1}) + \CompMaxDiv^{p(n_2)}(\rho_{2} \|\sigma_{2}) \; .
    \end{align}
\end{lem}
\begin{proof}
 By construction and super-additivity of $p$ we have for $E_i \in \mathbf{E}^{\mathrm{eff}}_{n_i}$ (for $i \in \{1,2\}$), that ${E_1 \otimes E_2 \in \mathbf{E}^{\mathrm{eff}}_{n_1+n_2}}$. Due to convexity, it then naturally follows that if  $E_i \in \overline{\mathbf{E}}^{\mathrm{eff}}_{n_i}$ (for $i \in \{1,2\}$), then $E_1 \otimes E_2 \in \overline{\mathbf{E}}^{\mathrm{eff}}_{n_1+n_2}$. We refer to \Cref{rem:postprocessing} for a more detailed discussion on this.

For any $v_i \in \Co^{\mathbf{E}^{\mathrm{eff}}}_{n_i} $, there exists a $\lambda_{i}$ and $E_i \in \overline{\mathbf{E}}^{\mathrm{eff}}_{n_i}$ such that $v_i = \lambda_{i} E_{i}$. Since $E_1 \otimes E_2 \in \overline{\mathbf{E}}^{\mathrm{eff}}_{n_1+n_2} $, it must hold that $v_{1} \otimes v_{2} \in \Co^{\mathbf{E}^{\mathrm{eff}}}_{n_1+n_2}$.
As such
    \begin{align*}
        \CompMaxDiv^{p(n_1 + n_2)} (\rho \|\sigma)  &= \log \sup_{v \in  \Co^{\mathbf{E}^{\mathrm{eff}}}_{n_1+n_2}} \frac{\Tr[v\left(\rho_{1} \otimes \rho_{2}\right)]}{\Tr[ v\left(\sigma_{1} \otimes \sigma_{2}\right)  ]}\\
        &\geq \log \sup_{v_{i} \in  \Co^{\mathbf{E}^{\mathrm{eff}}}_{n_i}} \frac{\Tr[\left(v_1 \otimes v_2\right)\left(\rho_{1} \otimes \rho_{2}\right)]}{\Tr[ \left(v_1 \otimes v_2 \right)\left(\sigma_{1} \otimes \sigma_{2}\right)  ]}\\
    &= \log \sup_{v_{1} \in  \Co^{\mathbf{E}^{\mathrm{eff}}}_{n_1}} \frac{\Tr[v_1 \rho_{1} ]}{\Tr[ v_1 \sigma_{1}  ]} +  \log \sup_{v_{2} \in  \Co^{\mathbf{E}^{\mathrm{eff}}}_{n_2}} \frac{\Tr[v_2\rho_{2}]}{\Tr[  v_2\sigma_{2}  ]} \\
    &= \CompMaxDiv^{p(n_1)}(\rho_{1} \|\sigma_{1}) + \CompMaxDiv^{p(n_2)}(\rho_{2} \|\sigma_{2}) \, .
    \end{align*}
\end{proof}
The same proof can be extended to operators of the form $\rho^{\otimes m},\sigma^{\otimes m} \in \Pos(\HH^{\otimes n\cdot m})$. Here, one simply finds that
    \begin{align}
         \CompMaxDiv^{p(n
         \cdot m)}(\rho^{\otimes m} \|\sigma^{\otimes m}) \geq m\CompMaxDiv^{p(n)}(\rho \|\sigma) \; .
    \end{align}

In other words, the super-additivity property follows from the observation that, informally, tensor products of efficient measurements are also an efficient measurement. This measurement can efficiently be reduced to a binary measurement, using logical OR gates if $m \in \poly(n)$.

Sub-additivity cannot be proven via the same techniques. Nevertheless, acting on quantum states, it is possible to prove that the computational max-divergence must also be approximately sub-additive.
\begin{lem}[Sub-additivity of $\CompMaxDiv$] \label{lem: subaddmaxdiv}
The computational max-divergence $\CompMaxDiv(\rho \| \sigma)$ is approximately  sub-additive for quantum states. For any $\rho = \rho_{1} \otimes \rho_{2}$, $\sigma = \sigma_{1} \otimes \sigma_{2} $ such that ${\rho_{i},\sigma_{i}\in \States (\HH^{\otimes n_i})}$ with cost $\kappa(\rho_{i})$, $\kappa(\sigma_{i})$,\footnote{Here, the cost is the minimal circuit size required to prepare them using gates from the prior fixed gate set $\mathcal{G}$.}, it holds that
    \begin{align}
         \CompMaxDiv^{p(n_1 + n_2)}(\rho \|\sigma) \leq \CompMaxDiv^{p(n_1+n_2)+\kappa}(\rho_{1} \|\sigma_{1}) + \CompMaxDiv^{p(n_1+n_2)+\kappa}(\rho_{2} \|\sigma_{2}) \; ,
    \end{align}
 where $\kappa \coloneqq \max\{\kappa(\rho_{1}), \kappa(\rho_{2}),\kappa(\sigma_{1}), \kappa(\sigma_{2})\}$.
\end{lem}
\begin{proof}
The proof technique is again inspired by~\cite{RLW24}. For any $\lambda_{1},\lambda_{2} \in \RR$, it holds that
   \begin{align*}
        (\lambda_{1}\cdot\lambda_{2}) \sigma_{1} \otimes \sigma_{2} - \rho_{1} \otimes \rho_{2} = \left( \lambda_{1} \sigma_{1} - \rho_{1}\right) \otimes \lambda_{2} \sigma_{2} + \rho_{1} \otimes \left(\lambda_{2} \sigma_{2} -\rho_{2} \right) \; .
    \end{align*}
Let us consider any $\lambda =\lambda_{1}\cdot\lambda_{2}\geq 0$ for which it holds that 
\begin{align} \label{eq: subaddcond}
    \Tr[(\lambda\sigma-\rho) E]\geq 0 \; , \qquad \forall E \in \mathbf{E}^{\mathrm{eff}}_{n_1+n_2} \; .
\end{align}
By construction it must similarly hold that 
\begin{align*}
    \Tr[(\lambda\sigma-\rho) v]\geq 0 \; , \qquad \forall v \in \Co^{\mathbf{E}^{\mathrm{eff}}}_{n_1+n_2} \; ,
\end{align*}
and, by \Cref{lem: dualexp},  one then finds that
\begin{align*}
    \CompMaxDiv^{p(n_1 + n_2)}(\rho \|\sigma) \leq \log  \lambda \; .
\end{align*}
Our goal is to show that $\lambda=\lambda_{1}\cdot\lambda_{2}$ satisfies Eq.~\eqref{eq: subaddcond}, if one sets
    \begin{align}
       \log \lambda_{i} = \CompMaxDiv^{p(n_1+n_2)+\kappa +1}(\rho_{i} \|\sigma_{i}) \; . \label{eq: lambdaidef}
    \end{align}
Note that, in particular, Eq.~\eqref{eq: lambdaidef} implies $\lambda_{i} \geq 0$ for $ i \in \{1,2\}$. Since it takes at most $\kappa$ extra gates to generate $\sigma_2$, for any $E \in \mathbf{E}^{\mathrm{eff}}_{n_1+n_2}$ acting on $\HH^{\otimes n_1} \otimes \HH^{\otimes n_2}$ there exists another effect operator $E^\prime$ on $\HH^{\otimes n_1}$ that first additionally prepares $\sigma_2$ and then implements $E$. By construction, it must hold that 
\begin{align*}
     \Tr[(\lambda_{1}\sigma_{1}-\rho_{1})\otimes \sigma_2 E] =  \Tr[(\lambda_{1}\sigma_{1}-\rho_{1}) E^\prime] \;. 
\end{align*}
Moreover, since $E^\prime$ uses at most $p(n_1+n_2)+\kappa$ gates, it follows from the definition of $\lambda_1$ that 
\begin{align*}
    \Tr[(\lambda_{1}\sigma_{1}-\rho_{1})\otimes \sigma_2 E] =  \Tr[(\lambda_{1}\sigma_{1}-\rho_{1}) E^\prime] \geq 0 \; .
\end{align*}
Similarly, it must hold for any $E \in \mathbf{E}^{\mathrm{eff}}_{n_1+n_2}$ that
\begin{align}
    \Tr[\rho_{1} \otimes(\lambda_{2}\sigma_{2}-\rho_{2}) E]  \geq 0 \; .
\end{align}
Combining this, we find for $\lambda =\lambda_{1}\cdot\lambda_{2}$ and any $E \in \mathbf{E}^{\mathrm{eff}}_{n_1+n_2}$that 
\begin{align*}
    \Tr[(\lambda\sigma-\rho) E] &= \lambda_{2}  \Tr[(\lambda_{1}\sigma_{1}-\rho_{1})\otimes\sigma_2 E] + \Tr[\rho_{1} \otimes (\lambda_2\sigma_2-\rho_2) E]
    \geq 0 \; ,
\end{align*}
thus satisfying Eq.~\eqref{eq: subaddcond}. As a result, it must hold that
\begin{align*}
     \CompMaxDiv^{p(n_1 + n_2)}(\rho \|\sigma) &\leq \log  \lambda \\
     &= \log  \lambda_1 + \log  \lambda_2 \\
     &= \CompMaxDiv^{p(n_1+n_2)+\kappa}(\rho_{1} \|\sigma_{1}) + \CompMaxDiv^{p(n_1+n_2)+\kappa}(\rho_{2} \|\sigma_{2}) \; .
\end{align*}
\end{proof}
\begin{rem}
It generally makes sense to consider quantum states $\rho$ and $\sigma$ that are efficiently preparable, i.e. where $\kappa$ scales polynomially in $n_1,n_2$. This way, the max-divergences that appear in \Cref{lem: subaddmaxdiv} are all related to cones generated by efficient effect operators. Moreover, the proof can straightforwardly be modified to hold for positive semidefinite operators $\rho, \sigma \in \Pos \left( \HH^{\otimes n_1} \otimes \HH^{\otimes n_2}\right)$. However, for any $\rho_{i} \in \Pos \left( \HH^{\otimes n_i}\right)$, $\kappa(\rho_{i})$ would then be the gate complexity required to generate the corresponding normalized quantum state $\frac{\rho_{i}}{\Tr \left[ \rho_{i} \right]}$ and analogously for $\sigma_{i}$.
\end{rem}
For states of the form $\rho^{\otimes m},\sigma^{\otimes m} \in \States( \HH^{\otimes n \cdot m})$, the same proof can be slightly modified, using~\cite{RLW24}
\begin{align}\label{equ:powersumequality}
    \lambda^m \sigma^{\otimes m} - \rho^{\otimes m} = \sum_{j=1}^{m}\lambda^{j-1} \rho^{\otimes m-j}\otimes(\lambda\sigma-\rho)\otimes\sigma^{\otimes j-1} \; .
\end{align}
One then retrieves the bound
\begin{align} \label{eq: IIDsubaddbounds}
         \CompMaxDiv^{ p(n\cdot m)}(\rho^{\otimes m} \|\sigma^{\otimes m}) \leq m\CompMaxDiv^{p(n\cdot m)+(m-1)\kappa}(\rho \|\sigma) \; ,
    \end{align}
where $\kappa \coloneqq \max\{\kappa(\rho), \kappa(\sigma)\}$. For fixed $\rho,\sigma$, one should expect the RHS of Eq.~\eqref{eq: IIDsubaddbounds} to converge to the information-theoretic max-divergence in the limit $m\to\infty$. This is due to the fact that while the states are fixed, the gate complexity of the allowed measurement is increasing as a function of $m$. This inequality is not unexpected, as it always holds that
\begin{align}
    \CompMaxDiv(\rho^{\otimes m} \|\sigma^{\otimes m}) \leq D_\text{max}(\rho^{\otimes m} \|\sigma^{\otimes m}) = m D_\text{max}(\rho \|\sigma) \; .
\end{align}

Lastly, let us relate the computational max-divergence with the computational trace distance. 

\begin{lem}[Computational Max-divergence and Computational Trace Distance]
\label{lem:comp.max.norm}
Given two quantum states $\rho,\sigma\in \States(\HH^{\otimes n})$, we have that
\begin{align}
\CompMaxDiv(\rho \| \sigma) \geq \log\left(1 + \CompTrDis \left(\rho,\sigma\right)\right) \; .
\end{align}
\end{lem}
\begin{proof}
The proof is straight forward and can be found in \Cref{app:comp.max.norm}.
\end{proof}

\subsection{Computational Measured {\Renyi} Divergences}  \label{Sec: CompMeasRenDiv}
This section is devoted to computational measured Rényi divergences. Here, we consider an informationally complete polynomially generated set of binary POVMs $\{\Meff_n\}_{n\in\mathbb{N}}$, who's underlying polynomial is strictly super-additive. For each set $\Meff_n$, we will define a measured {\Renyi} divergence. 

\begin{definition}[Computational Measured Rényi Divergence]
Given any two positive semidefinite operators $\rho,\sigma \in \Pos (\HH^{\otimes n})$ and an element $\Meff_n$ of an informationally complete polynomially generated set of binary POVMs $\{\Meff_n\}_{n\in\mathbb{N}}$, we denote the corresponding \term{computational measured Rényi divergence} for any $\alpha \in \left(0,1\right) \cup \left(1,\infty\right)$ by:
\begin{align}
\CompDiv^{\Meff_n}_{\alpha}(\rho\|\sigma) \coloneqq \sup_{M \in \Meff_n} \SandRenyi_{\alpha}(\MM_M(\rho)\|\MM_M(\sigma)) \; ,
\end{align}
where $\SandRenyi_\alpha(\rho\|\sigma)$ is the sandwiched Rényi divergence (\Cref{def:sandwiched.renyi}) and $\MM(\cdot)$ is the measurement channel corresponding to the POVM $M$.
\end{definition}
Moreover, we use $\CompQ^{\Meff_n}_{\alpha}(\rho\|\sigma)$ to denote
\begin{align}\label{eq:comp.Q.Renyi}
    \
     \CompQ^{\Meff_n}_{\alpha}(\rho\|\sigma) &\coloneqq \begin{cases}
				\inf_{M \in \Meff_n} Q_{\alpha}(\MM(\rho)\|\MM(\sigma)) & \text{ if }  \alpha \in (0,1) \\
				\sup_{M \in \Meff_n} Q_{\alpha}(\MM(\rho)\|\MM(\sigma)) & \text{ if }  \alpha \in (1,\infty)
			\end{cases} \; .
\end{align}
\begin{definition} [Computational Measured Relative Entropy \& Max-divergence]
\label{def:comp.rel.max.diver}
Given any two positive semidefinite operators $\Pos (\HH^{\otimes n})$ with $\operatorname{Tr}\left[\rho\right] > 0$ and $\rho \ll \sigma$ and an element $\Meff_n$ of an informationally complete polynomially generated set of binary POVMs $\{\Meff_n\}_{n\in\mathbb{N}}$, the  \term{computational measured relative entropy} and \term{max-divergence} between $\rho$, $\sigma$ are, respectively, given by:
		\begin{align}
\CompDiv^{\Meff_n}(\rho\|\sigma) &\coloneqq 
				\sup_{M \in \Meff_n}\frac{\Tr[\MM \left(\rho\right)(\log\MM \left(\rho\right)-\log\MM \left(\sigma\right))]}{\Tr\left[\MM \left(\rho\right)\right]} = \lim_{\alpha\nearrow 1}\CompDiv^{\Meff_n}_{\alpha}(\rho\|\sigma)\; ,  \\
                \CompDiv^{\Meff_n}_\text{max}(\rho\|\sigma) &\coloneqq 
				\sup_{M \in \Meff_n}\log\inf\{\lambda\in\RR|\MM \left(\rho\right)\leq\lambda \MM \left(\sigma\right)\}= \lim_{\alpha\nearrow\infty}\CompDiv^{\Meff_n}_{\alpha}(\rho\|\sigma) \; .
		\end{align}
\end{definition}
By taking the appropriate limits, one can retrieve the computational measured relative entropy and max-divergence from the computational measured {\Renyi} divergences, see e.g.~\cite[Lemma 2]{RSB24}.
We note that a priori one could have also defined the computational measured {\Renyi} divergences using 
$\overlineMeff_{n}$ instead of $\Meff_n$. We show in \Cref{lem: equivalanceMandConvM} below, however, that these definitions are equivalent.

\begin{lem} \label{lem: equivalanceMandConvM}
For any $\mathbb{D}\in\{\SandRenyi_\alpha,D,D_\text{max}\}$, where $\alpha  \in \left(0,1\right) \cup \left(1,\infty\right)$, it holds that
\begin{align}
\CompbbDiv^{\Meff_n}(\rho\|\sigma) = \CompbbDiv^{\overlineMeff_{n}}(\rho\|\sigma) \; .
\end{align}
\end{lem}
\begin{proof}
Since $\Meff_n \subset \overlineMeff_{n}$, it always holds that
\begin{align}
\CompbbDiv^{\Meff_n}(\rho\|\sigma) \leq \CompbbDiv^{\overlineMeff_{n}}(\rho\|\sigma) \; .
\end{align}
To prove the reverse inequality, we first note that
the relative entropy $D$ and $\SandRenyi_{\alpha}$ are jointly convex on classical states in the range ${\alpha \in \left(0,1\right)}$; the max-divergence $D_\text{max}$ and $\SandRenyi_\alpha$ are jointly quasi-convex for ${\alpha\in \left[1,\infty\right)}$~\cite{T16}. Denoting the extremal points of $\overlineMeff_{n}$ by $\extr\left(\overlineMeff_{n}\right)$, a consequence of the (quasi-) convexity is that
\begin{align*}
\CompbbDiv^{\overlineMeff_{n}}(\rho\|\sigma)&\coloneqq \sup_{M \in \overlineMeff_{n}}\mathbb{D}(\MM(\rho)\|\MM(\sigma)) 
=\sup_{M\in\extr\left(\overlineMeff_{n}\right)}\mathbb{D}(\MM(\rho)\|\MM(\sigma)) \; . 
\end{align*}
Since the extremal points of $\overlineMeff_{n}$ satisfy 
$\extr\left(\overlineMeff_{n}\right)\subset \Meff_n$ by construction of $\overlineMeff_{n}$, 
it holds for all $\mathbb{D}\in\{D,\SandRenyi_\alpha,D_{\text{max}}\}$ that 
\begin{align*}
\CompbbDiv^{\overlineMeff_{n}}(\rho\|\sigma)&=\sup_{M\in\extr\left(\overlineMeff_{n}\right)}\mathbb{D}(\MM(\rho)\|\MM(\sigma)) \leq \sup_{M\in{\Meff_n}}\mathbb{D}(\MM(\rho)\|\MM(\sigma)
= \CompbbDiv^{\mathbf{M}_n^{\mathrm{eff}}}(\rho\|\sigma)  \; .  
\end{align*} 
\end{proof}

\subsubsection{Properties of Measured {\Renyi} Divergences}
In~\cite{RSB24}, properties of generic measured {\Renyi} divergences are derived. For completeness, in this section,  we rephrase the results of~\cite[Lemmas 2 \& 5]{RSB24} in our setting of efficient measurements. 
\begin{lem} \label{lem: MeasRenDivProp}
    For any element $\Meff_n$ of an informationally complete polynomially generated set of binary POVMs $\{\Meff_n\}_{n\in\mathbb{N}}$ and $\alpha,\beta \in \left(0,1\right) \cup \left(1,\infty\right)$, the corresponding computational measured {\Renyi} entropies satisfy
    \begin{enumerate}
        \item (Non-negativity): For any two normalized quantum states $\rho,\sigma \in \States (\HH^{\otimes n})$,
        \begin{align}
     \CompDiv_{\alpha}^{\Meff_{n}}(\rho\|\sigma) \geq 0 \; .
        \end{align}
        \item (Monotonicity in $\alpha$): For any $\alpha \leq \beta$ and $\rho,\sigma \in \Pos(\HH^{\otimes n})$,
              \begin{align}
 \CompDiv_{\alpha}^{\Meff_{n}}(\rho\|\sigma) \leq  \CompDiv_{\beta}^{\Meff_{n}}(\rho\|\sigma) \; .
              \end{align}
        \item (Joint (Quasi-) Convexity):  For any two normalized quantum states $\rho,\sigma \in \States (\HH^{\otimes n})$, we have that $\CompQ^{\Meff_n}_{\alpha}(\rho\|\sigma)$ is jointly concave in $\rho,\sigma$ if $\alpha \in \left(0,1\right)$ and jointly convex if $\alpha \in \left(1,\infty\right)$. Moreover, $\CompDiv^{\Meff_n}_{\alpha} (\rho\|\sigma)$ is jointly convex if $\alpha \in \left(0,1\right)$ and jointly quasi-convex if $\alpha \in \left(1,\infty\right)$.
        \item (Data Processing Inequality):  For any $\rho,\sigma \in \Pos(\HH^{\otimes n})$, let $\EE \in CPTP \left(\HH^{\otimes n},\HH^{\otimes n}\right)$ be a completely positive trace-preserving map such that for any $E \in\Eeff_n$, $\EE^* \left( E\right) \in \overline{\mathbf{E}}^\text{eff}_{n}$. Then
              \begin{align}
 \CompDiv_{\alpha}^{\Meff_{n}}(\rho\|\sigma) \geq  \CompDiv_{\alpha}^{\Meff_{n}}(\EE(\rho)\|\EE(\sigma)) \; .
              \end{align}
    \end{enumerate}
\end{lem}
The first three properties are directly inherited from the information-theoretic sandwiched {\Renyi} divergence. The last property can be understood as stating that we only consider $\EE$ such that for every $M \in\Meff_n$ there exists another measurement 
$\bar{M} \in \overlineMeff_{n}$ which satisfies $\Bar{\MM}\left(\rho\right) = \MM \circ \EE \left(\rho\right)$ for every $\rho$. In this case, optimizing over $M \circ \EE$ is more restrictive than simply optimizing over $M$.
\begin{proof}
The proof for \Cref{lem: MeasRenDivProp} directly follows from~\cite{RSB24}. However, for completeness, we include it in~\Cref{app:MeasRenDivProp}.    
\end{proof}

Next let us prove a Pinsker inequality for the computational measured Rényi divergences.
\begin{lem}[Computational Pinsker inequality]
\label{lem:pinsker}
 Given two quantum states $\rho,\sigma\in \States(\HH^{\otimes n})$, we have that for $\alpha \in (0,\infty)$, 
\begin{align}
    \CompDiv^{\Meff_n}_{\alpha} (\rho \|\sigma) \geq \frac{\min\{1,\alpha\}}{2\ln 2}\| \rho - \sigma\|_{\overlineMeff_n}^2 \,.
\end{align}
\end{lem}
\begin{proof}
As discussed in~\cite{RSB24}, the bound for $\alpha \in (0, 1]$ follows directly from the classical result~\cite{Gi10}. Additionally, we convert the logarithm to base 2 and take the supremum over all efficient measurements. The bound for $\alpha \in (1, \infty)$ follows from the monotonicity in $\alpha$ property of the computational measured Rényi divergence (see \Cref{lem: MeasRenDivProp}).
\end{proof}
\subsubsection{Equivalence Between Conic and Measured Max-Divergence}
We are now ready to prove one of the main results of this section, namely that the initially geometrically defined computational max-divergence ($\CompDiv_\text{max}$ from \Cref{def:ComConeMax}) and via the measured {\Renyi} divergences constructed computational measured max-divergence ($\CompDiv^{\Meff_n}_\text{max}$ from \Cref{def:comp.rel.max.diver}) are equivalent in our setting. The fact that these two divergences are actually equivalent is a priori not clear and this result highlights the strength of applying conic techniques to the computational setting.

This also naturally complements the known information-theoretic result, which states that there exists an optimal measurement which leaves the max-divergence invariant~\cite{MH23,Mosonyi_2014}.
\begin{theorem}[Max-Divergence Equivalence]\label{thm:measured=conic}
Let the families $\{\Eeff_n\}_{n\in\NN}, \{\Meff_n\}_{n\in\NN},$ and $\{\EffCone\}_{n\in\NN}$ be as in \Cref{Sec: EffMeas} and informationally complete and define $\CompDiv_\text{max}$ and $\CompDiv^{\Meff_n}_\text{max}$ w.r.t.~these as in \Cref{Sec: CompMaxDiv} and \Cref{Sec: CompMeasRenDiv}, respectively. Then for any $n \in \NN$ and two positive semidefinite operators $\rho,\sigma\in \States( \HH^{\otimes n})$, it  holds that
\begin{align}\label{equ:MaxDivEqu}
\CompDiv_\textup{max}(\rho\|\sigma) =\CompDiv^{\Meff_n}_\textup{max}(\rho\|\sigma) \eqqcolon \lim_{\alpha\to\infty}\CompDiv^{\Meff_n}_\alpha(\rho\|\sigma)  \; .
\end{align}
\end{theorem}

\begin{proof}
Using the approach of~\cite{RSB24}, we first note that the computational max-divergence can alternatively be expressed as
\begin{align} \label{eq: dualexpequivproof}
    \CompMaxDiv(\rho \|\sigma) =\log \sup_{v \in \EffCone} \frac{\Tr[ v\rho ]}{\Tr[ v \sigma ]} \; .
\end{align}
By construction, for every $v$, there exists a $\lambda \geq  0$ and an element $E \in \overline{\mathbf{E}}^{\mathrm{eff}}_{n}$ such that $v=\lambda E$. Since the RHS of Eq.~\eqref{eq: dualexpequivproof} is invariant under the scaling of $\lambda$, it must also hold that
\begin{align*}
    \CompMaxDiv(\rho \|\sigma) &=\log \sup_{E \in \overline{\mathbf{E}}^{\mathrm{eff}}_{n}} \frac{\Tr[ E\rho ]}{\Tr[ E \sigma ]}
    = \sup_{M=(E^1,E^2)\in\overlineMeff_{n}}\log \max_{i\in\{1,2\}}\frac{\Tr[E^i\rho]}{\Tr[E^i\sigma]} \; .
\end{align*}
The last equality holds due to the condition that $E \in \overline{\mathbf{E}}^{\mathrm{eff}}_{n}$ implies $\id -E \in \overline{\mathbf{E}}^{\mathrm{eff}}_{n}$.
We will now derive the same expression for the computational measured max-divergence. 
Let $M=(E^1,E^2) \in \overlineMeff_n$ be a binary measurement with associated measurement map $\mathcal{M}$, then it is easy to check that the condition $\MM \left(\rho\right)\leq\lambda \MM \left(\sigma\right)$ is equivalent to 
\begin{align*}
    \lambda \geq \max_{i\in\{1,2\}}\frac{\Tr[E^i\rho]}{\Tr[E^i\sigma]} \; .
\end{align*}
Hence it follows that
\begin{align*}
   \inf\{\lambda\in\RR|\MM \left(\rho\right)\leq\lambda \MM \left(\sigma\right)\} = \max_{i\in\{1,2\}}\frac{\Tr[E^i\rho]}{\Tr[E^i\sigma]} \; .
\end{align*}
Combining these observations with \Cref{lem: equivalanceMandConvM}, we directly find that
\begin{align*}
\CompDiv^{\Meff_n}_\text{max}(\rho\|\sigma) &=\CompDiv^{\overlineMeff_n}_\text{max}(\rho\|\sigma) \\
&=\sup_{M=(E^1,E^2) \in \overlineMeff_n}\log\inf\{\lambda\in\RR|\MM \left(\rho\right)\leq\lambda \MM \left(\sigma\right)\}\\
&=    \sup_{M =(E^1,E^2)\in \overlineMeff_n}\log  \max_{i\in\{1,2\}}\frac{\Tr[E^i\rho]}{\Tr[E^i\sigma]} \\
&= \CompMaxDiv(\rho \|\sigma) \; ,
   \end{align*}
which concludes the proof.
\end{proof}
One interesting consequence of \Cref{thm:measured=conic} can be stated as follows. If two sets of measurements generate the same cone, then their corresponding measured max-divergences are also equal to each other. As we discussed in \Cref{Sec: EffMeas} the cone $\EffCone$ may contain POVM elements which are on their own not efficiently preparable, however, due to the above argument, they do not aid in distinguishability when using the max-divergence as a distance measure.

\begin{rem}
That these two divergences ---one defined via the geometric construction of cones of efficient effect operators, and the other via the optimization over measured $\alpha=\infty$-{\Renyi} \ divergences--- are equivalent is a priori not clear in general. Let us consider a slightly more general setting where one starts with a family of admissable (not necessarily binary) POVMs $\mathbf{M}$ and generates from this the cone $\mathcal{C}^E$ as the smallest cone that contains all effect operators (POVM elements) appearing in all $M\in \mathbf{M}$. This is, among others, the setting of~\cite{RSB24}. Then, if one defines these two max-divergences w.r.t.~these objects, one can show that the measured max-divergence is always upper bounded by the conical max-divergence; see~\cite[Proposition 13]{RSB24}. In other words, the ``$\leq$" inequality of \eqref{equ:MaxDivEqu} is always true. Intuitively, this is because every POVM element is also element of the cone $\mathcal{C}^E$. The other inequality and thus equality, however, may not always be true. This direction requires an extra condition. Up to rescaling, every element of the cone $\mathcal{C}^E$ needs to either be an effect operator that appears in some \term{binary} POVM in $\mathbf{M}$ or needs to be expressible as a convex mixture of such binary measurements. In our case this is not an issue, since we consider only binary POVMs. 
However, if one, e.g., considers $\mathbf{M}$ to be the set of LO measurements, which are necessarily generally non-binary, then the above proof of equality generally breaks down~\cite{RSB24}. This is because generating the required two-output measurements would require classical communication, which is not allowed.
\end{rem}

\subsection{Regularized Computational Measured {\Renyi} Divergences} 
Let us define the regularized versions of the computational max-divergence and the computational measured {\Renyi} divergences. We now pass from finite-block, per-instance quantities to \emph{asymptotic rates}. In computational settings, an efficient tester is constrained not only in circuit size but also in the number of samples; accordingly, we adopt a \emph{poly-sample} regime in which the number of copies satisfies $m(k)\in \poly(k)$. Since we are studying asymptotic behavior, we henceforth consider families of states $\{\rho_k\}_{k\in\NN}$ and $\{\sigma_k\}_{k\in\NN}$. Again, we fix families of efficient measurements $\{\Meff_n\}_{n\in\NN}$. While the single-instance measures are evaluated at each fixed $k$, the regularized quantities are meaningful primarily at the family level.

\begin{definition}[Regularized Computational Max-divergence]
    Given two family of states $\{\rho_k\}_{k\in\NN},$ $\{\sigma_k\}_{k\in\NN} \in \{\Ss(\HH^{\otimes k})\}_{k \in \NN}$, their \term{regularized computational max-divergence} when $m(k) \in \poly(k)$ is given by,
    \begin{align}
        \CompMaxDiv^{\infty}(\{\rho_k\}_k \|\{\sigma_k\}_k) \coloneqq \limsup_{k\rightarrow\infty}\frac{1}{m(k)}\CompMaxDiv(\rho_k^{\otimes m(k)} \|\sigma_k^{\otimes m(k)} )\; .
    \end{align}
\end{definition}
For the simple case in which $\rho_k = \rho^{\otimes k}$ ($\sigma_k = \sigma^{\otimes k}$) for all $k \in \NN$, this regularization (denoted by $\CompMaxDiv^{\infty}(\rho \| \sigma)$) admits a convenient equivalent form.
\begin{lem}
Given two quantum states $\rho,\sigma \in \Ss(\HH)$ such that, their regularized computational max-divergence can be expressed as,
      \begin{align}
        \CompMaxDiv^{\infty}(\rho \| \sigma) = \lim_{m \to \infty}\frac{1}{m}\CompMaxDiv(\rho_k^{\otimes m} \|\sigma_k^{\otimes m}) = \sup_{m>0}\frac{1}{m}\CompMaxDiv(\rho_k^{\otimes m} \|\sigma_k^{\otimes m})\; .
    \end{align}
\end{lem}
\begin{proof}
    The proof follows from the super-additivity property of the computational max-divergence (see \Cref{lem:superadd.max}) together with Fekete's lemma.
\end{proof}
\begin{definition}[Regularized Computational Measured {\Renyi} Divergences]
    Given two families of states $\{\rho_k\}_{k\in\NN},\{\sigma_k\}_{k\in\NN} \in \{\Ss(\HH^{\otimes k})\}_{k \in \NN}$, for any $\alpha \in (0,1) \cup (1,\infty)$ their \term{regularized computational measured {\Renyi} divergence} when $m(k) \in \poly(k)$ is given by,
    \begin{align}
        \CompDiv_{\alpha}^{\Meff, \infty}(\{\rho_k\}_k \|\{\sigma_k\}_k) \coloneqq   \limsup_{  k \rightarrow \infty} \frac{1}{m(k)}  \CompDiv_{\alpha}^{\Meff_{k\cdot m(k)}}(\rho_k^{\otimes m(k)} \| \sigma_k^{\otimes m(k)})   \; .
    \end{align}
\end{definition}

\begin{definition}[Regularized Computational Measured Relative Entropy]
    Given two families of quantum states $\{\rho_k\}_{k\in\NN},\{\sigma_k\}_{k\in\NN} \in \{\Ss(\HH^{\otimes k})\}_{k \in \NN}$, their \term{regularized computational measured relative entropy} is given by,
    \begin{align}
        \CompDiv_{}^{\Meff, \infty}(\{\rho_k\}_k \|\{\sigma_k\}_k) \coloneqq   \limsup_{  k \rightarrow \infty} \frac{1}{m(k)} \CompDiv^{\Meff_{k\cdot m(k)}}(\rho_k^{\otimes m(k)} \| \sigma_k^{\otimes m(k)})   \; .
    \end{align}
\end{definition}

\begin{rem}
In fact, for binary measurements super-additivity fails for general $\alpha \neq \infty$~\cite{MH23}. It remains an interesting open question whether the regularized computational measured relative entropy can be equivalently  expressed as a supremum, similar to the computational max-divergence.
\end{rem}
We conclude by relating the various regularized divergences to each other.
\begin{lem}  \label{lem: regularizeddivergenceineq}
For any two  families of quantum states $\{\rho_k\}_{k\in\NN},\{\sigma_k\}_{k\in\NN} \in \{\Ss(\HH^{\otimes k})\}_{k \in \NN}$, $\alpha\in (0,1)$, and ${\beta\in (1,\infty)}$, it holds that
\begin{align}
\CompDiv^{\Meff, \infty}_{\alpha}
(\{\rho_k\}_k \|\{\sigma_k\}_k) \leq \CompDiv^{\Meff,\infty}_{}
(\{\rho_k\}_k \|\{\sigma_k\}_k) \leq \CompDiv^{\Meff,\infty}_{\beta}
(\{\rho_k\}_k \|\{\sigma_k\}_k) \leq \CompDiv_{\max}^{\infty}(\{\rho_k\}_k \|\{\sigma_k\}_k)\; .
\end{align}
\end{lem}
\begin{proof}
    This follows directly from the monotonicity property provided in \Cref{lem: MeasRenDivProp}.
\end{proof}

\subsection{Computational Fidelity}\label{sec:fidelity}
In this section we define the \term{computational (binary) fidelity} using the approach from~\cite{RKW11}. This is the computational analog to the \term{fidelity}~\cite{T16,Uhl76,Fuc96} between two post-measurement single-bit outputs and measures the similarity between two quantum states when restricted to binary efficient measurements. Additionally, we show that the resulting fidelity is effectively the computational measured {\Renyi} divergence for $\alpha=\frac{1}{2}$.  Lastly, we relate it to the \emph{computational trace distance}  and the \emph{computational max-divergence}, following the techniques from~\cite{RKW11}.

\begin{definition}[Computational (Binary) Fidelity] 
 Given any two quantum states $\rho,\sigma\in \States(\HH^{\otimes n})$ the \term{computational (binary) fidelity} is given by
\begin{align}
    \CompFid(\rho, \sigma) \coloneqq \inf_{M = (E_1,E_2) \in \overlineMeff_n}  \left(\sum_{i = 1}^{2} \sqrt{\Tr[E_i\rho]\Tr[E_i\sigma]} \right)^2\; .
\end{align}    
\end{definition}
\begin{lem}[Computational Fidelity and Computational Measured Rényi Divergence]\label{lem:fid.comp.Renyi}
Given two quantum states $\rho,\sigma\in \States(\HH^{\otimes n})$, it follows that
 \begin{align}
    \CompDiv_{1/2}^{\Meff_n}(\rho \|\sigma) = -\log \CompFid(\rho, \sigma) \,.
\end{align}
\end{lem}
\begin{proof}
 The proof can be found in \Cref{app:fid.comp.Renyi}
\end{proof}
Since the computational fidelity is related to the computational measured {\Renyi} divergence for $\alpha=\frac{1}{2}$, due to \Cref{lem: equivalanceMandConvM} it automatically follows that 
\begin{align}
    \inf_{M = (E_1,E_2) \in \overlineMeff_n}  \left(\sum_{i = 1}^{2} \sqrt{\Tr[E_i\rho]\Tr[E_i\sigma]} \right)^2 = \inf_{M = (E_1,E_2) \in \Meff_n}  \left(\sum_{i = 1}^{2} \sqrt{\Tr[E_i\rho]\Tr[E_i\sigma]} \right)^2 \; ,
\end{align}
meaning that one could have equivalently defined the computational fidelity by solely optimizing over measurements $\Meff_{n}$.
The computational fidelity and computational trace distance additionally straightforwardly satisfy a Fuchs-van de Graaf type inequality.

\begin{lem}[Computational Fuchs-van de Graaf]
\label{lem:comp.fuchs.van.Graaf}
 Given two quantum states $\rho,\sigma \in \Ss(\HH^{\otimes n}) $, it holds that
 \begin{align}
     1-\sqrt{\CompFid(\rho, \sigma)} \leq \CompTrDis(\rho, \sigma) \leq \sqrt{ 1-\CompFid(\rho, \sigma)} 
 \end{align}
\end{lem}

\begin{proof}
    These inequalities follow directly from~\cite[Proposition 19]{RKW11} applied to our family of informationally complete polynomially generated set of binary POVMs $\Meff_n$. For completeness, we include a proof in \Cref{app:comp.fuchs.van.Graaf}.
\end{proof}
We conclude this section by relating the fidelity to the computational max-divergence. Applying the computational Fuchs-van de Graaf inequality to this bound then yields a result similar to \Cref{lem:comp.max.norm}.
\begin{lem}
\label{lem:fid.max.div}
Given two quantum states $\rho,\sigma\in \States(\HH^{\otimes n})$, we have that
\begin{align}
 \CompFid(\rho,\sigma) \geq 2^{-\CompDiv_{\textup{max}}(\rho \|\sigma)} \; .
\end{align}
\end{lem}

\begin{proof}
By monotonicity of the measured {\Renyi} $\alpha$ divergences, see \Cref{lem: MeasRenDivProp}, and \Cref{thm:measured=conic} it follows that
\begin{align*}
     \CompDiv_{1/2}^{\Meff_n}(\rho \|\sigma) \leq  \CompDiv_{\max}^{\Meff_n}(\rho \|\sigma) = \CompDiv_{\text{max}}(\rho \|\sigma)\;.
\end{align*}
Applying \Cref{lem:fid.comp.Renyi}, one then finds that
\begin{align*}
    \CompFid(\rho,\sigma) \geq 2^{-\CompDiv_{\text{max}}(\rho \|\sigma)}\; .
\end{align*}
\end{proof}
The previous bounds allow us to derive another relation between the computational max-divergence and the computational trace distance.
\begin{cor}
\label{cor:comp.dist.max.div}
Given two quantum states $\rho,\sigma\in \States(\HH^{\otimes n})$, we have that
\begin{align}
\CompMaxDiv(\rho \| \sigma) \geq -\log\left(1 - \left(\CompTrDis \left(\rho,\sigma\right)\right)^2\right) \; .
\end{align}
\end{cor}
\begin{proof}
   This follows from \Cref{lem:fid.max.div} together with the computational Fuchs-van de Graaf inequality provided in \Cref{lem:comp.fuchs.van.Graaf}.
\end{proof}
We provide a similar bound in \Cref{lem:comp.max.norm}. For $\CompTrDis \left(\rho,\sigma\right)\approx 0$, \Cref{lem:comp.max.norm} is tighter. Conversely, for $\CompTrDis \left(\rho,\sigma\right)\approx 1$, we find that \Cref{cor:comp.dist.max.div} provides a substantially tighter bound.

\subsection{An explicit cryptographic separation}\label{sec:example}
We now exhibit an explicit pair of states that witnesses a sharp separation between the \emph{computational} and \emph{informational} divergences. The same calculations can be extended to other cryptographic constructions which exhibit a gap between the computational trace distance and the informational counterpart. Nevertheless, due to its simplicity, we focus on the pseudoentanglement construction proposed by~\cite{GE24}. Let
\begin{align}
    \psi^{AB}_n &\coloneqq \frac{1}{4}\left(\ketbra{\Phi^+}{\Phi^+}^{AB}+ \ketbra{\Phi^-}{\Phi^-}^{AB}\right)\otimes \left(\rho^A_{0,n}+\rho^A_{1,n}\right)\, , \\
    \phi^{AB}_n &\coloneqq \frac{1}{2}\left(\ketbra{\Phi^+}{\Phi^+}^{AB}\otimes\rho^A_{0,n} + \ketbra{\Phi^-}{\Phi^-}^{AB}\otimes \rho^A_{1,n}\right)\,,
\end{align}
where $\ket{\Phi^+}^{AB} = \frac{1}{\sqrt{2}}( \ket{00}^{AB}+ \ket{11}^{AB})$, $\ket{\Phi^-}^{AB} = \frac{1}{\sqrt{2}}( \ket{00}^{AB}-\ket{11}^{AB})$   and  $\{\rho^A_{0,n},\rho^A_{1,n}\}_{n \in \NN}$ are EFI pairs.\footnote{An EFI pair~\cite{BCQ23} consists of two efficiently generated families of quantum states which are far in trace distance ($\Delta(\rho_{0, n}, \rho_{1, n}) \approx1 $), which are indistinguishable by computationally bounded adversaries given $m \in \poly(n)$ copies , i.e., $ \left| \probP\left[\mathcal{D}(\sigma_{n},\rho_{n,0}^{\otimes m}) = 1] - \probP [\mathcal{D}( \sigma_{n}, \rho_{n,1}^{\otimes m}) = 1\right]  \right| \leq \nu (n)$.}

Assuming EFI pairs exist,~\cite{GE24} proves \emph{poly-copy} computational indistinguishability of the two families, which in our notation implies
\begin{align*}
    \CompTrDis(\psi^{AB}_n, \phi^{AB}_n) \leq \negl(n) \,.
\end{align*}

Therefore, by the computational Fuchs–van de Graaf inequality (\Cref{lem:comp.fuchs.van.Graaf}) we obtain
\begin{align}
    1-\negl(n) \leq \CompFid(\psi^{AB}_n, \phi^{AB}_n) \leq 1\, .
\end{align}

This allows us to calculate the following upper bound for the computational measured Rényi divergence for $\alpha \in (0,1/2)$,
\begin{align}
     \CompDiv^{\Meff_n}_{\alpha} (\psi^{AB}_n, \phi^{AB}_n) \leq -\log( 1-\negl(n))\; ,
\end{align}
which follows from \Cref{lem:fid.comp.Renyi} together with the monotonicity property of the measured Rényi divergences (\Cref{lem: MeasRenDivProp}).
Nevertheless, in the information-theoretic setting one has, by the Pinsker inequality, e.g.,~\cite{Gi10}, that
\begin{align*}
     D_{\alpha} (\psi^{AB}_n, \phi^{AB}_n) \gtrsim \frac{\min\{1,\alpha\}}{8\ln 2} \,,
\end{align*}
since the information-theoretic trace distance satisfies $\Delta(\psi^{AB}_n, \phi^{AB}_n) \approx 1/2$ and the fidelity is bounded away from $1$:
\begin{align}
         \frac{1}{4} \leq F(\psi^{AB}_n, \phi^{AB}_n) \leq \frac{3}{4} \, .
\end{align}

Therefore, under standard hardness assumptions, the computational measures collapse \\ $(\CompTrDis(\psi^{AB}_n, \phi^{AB}_n)\leq\negl(n)$, $1-\negl(n)\leq\CompFid(\psi^{AB}_n, \phi^{AB}_n)\leq1$, and $\CompDiv_{\alpha}(\psi^{AB}_n, \phi^{AB}_n)^{\Meff_n}\leq\negl(n)$ for $\alpha\le \tfrac12$), whereas their informational counterparts are $\Theta(1)$.\footnote{Indeed, the gap can be polynomially amplified by taking poly-many copies of pseudoentangled states, as proven in~\cite[Lemma 5.1]{GE24}}

\section{Applications}
\label{Sec:applications}
In this section we develop applications for the computational divergences introduced above.
In \Cref{sec:comp.hyp.test} we revisit the quantum Stein’s lemma under computational constraints (\Cref{thm:qpt-stein}) and show that the optimal one-sided error exponent is upper bounded by the regularized computational measured relative entropy, thereby giving it an operational meaning (see \Cref{def:comp.rel.max.diver}). In \Cref{sec:comp.res.th} we define computational quantum resources based on our computational divergences and establish an asymptotic continuity bound for the so-called computational measured relative entropy of a resource (\Cref{th:comptfannes}). Finally, further focusing on entanglement (\Cref{sec:entanglement}), we introduce the \term{computational measured relative entropy of entanglement} and relate it to the previously proposed computational entanglement measures from~\cite{ABV23} in \Cref{lem: comp.ent.example}.

\subsection{Computational Hypothesis Testing}
\label{sec:comp.hyp.test}
Let us consider the problem of binary hypothesis testing, where one is tasked with discriminating between two input states. The main result of this section is a rigorous study of the optimal type II error decay rate in asymmetric hypothesis testing, commonly known as the Stein's exponent, when restricted to efficient binary measurements~\cite{HP91,ON00}.

In binary quantum hypothesis testing, the goal is to discriminate between two quantum states; one of the states, $\rho$, is called the \term{null hypothesis $H_0$} and the other state, $\sigma$, is known as the \term{alternate hypothesis $H_1$}. While in the information-theoretic setting one typically studies two fixed states $\rho$ and $\sigma$, in our complexity constrained framework the studied objects are families of states $\{\rho_k\}_{k \in \mathbb{N}}$ and $ \{\sigma_k\}_{k \in \mathbb{N}}$ of which we are given $m(k)\in\poly(k)$ polynomially many copies.

Thus, we study the distinguishing rate in the asymptotic limit, i.e., when $k \rightarrow \infty$ given $m(k)$ copies of $\rho_k,\sigma_k \in \States(\HH^{\otimes k})$ under the restriction that the observer is constrained to implement efficient binary measurements. We say that the observer determines the system to be in state $\rho_k$, instead of $\sigma_k$ if the performed measurement $M$ outputs $0$. Conversely, if the output is $1$, the observer assumes that they initially held $m(k)$ copies of $\sigma_k$.

For any $k, m(k) \in \NN$, let $\Meff_{k\cdot m(k)}$ denote the set of efficient binary measurements acting on $\HH^{\otimes k \cdot m(k)}$. For any feasible binary measurement ${M = ( E_{k\cdot m(k)},\id -E_{k\cdot m(k)}) \in \Meff_{k\cdot m(k)}}$, there are two types of errors that arise:
\begin{itemize}
    \item \textbf{Type-I error}: the probability of incorrectly rejecting the null hypothesis (i.e., mistakenly inferring $\sigma_k$ when the state is actually $\rho_k$). It is hence defined as
    \begin{align}
        \alpha_{m(k)}(M) \coloneqq \Tr [ \rho_k^{\otimes m(k)}(\id - E_{k\cdot m(k)} )] \; .
    \end{align}
    \item \textbf{Type-II error}: the probability of incorrectly accepting the null hypothesis (i.e., mistakenly inferring $\rho_k$ when the state is actually $\sigma_k$). It is given by
    \begin{align}
        \beta_{m(k)}(M) \coloneqq \Tr[\sigma_k^{\otimes m(k)} E_{k\cdot m(k)} ]  \; .
        \label{eq:type2err}
    \end{align}
\end{itemize}

We are primarily interested in the \emph{asymmetric} setting, where the aim is to minimize the Type-II error while the Type-I error is constrained to not exceed a prescribed threshold $\varepsilon \in (0,1)$~\cite{HP91,ON00}. Formally, one defines
\begin{align}\label{equ:def:optimaltypeIIerror}
    \beta^{\varepsilon}_{m(k)} (\Meff_{k \cdot m(k)}) \coloneqq \min_{M \in \Meff_{k \cdot m(k)} }\{ \beta_{m(k)}(M) \, | \, \alpha_{m(k)}(M) \leq \varepsilon\} \; ,
\end{align}
to be the optimal type-II error given a type-I error bound by $\varepsilon$. As mentioned above, our goal is to characterize the scaling of $\beta^{\varepsilon}_{m(k)} (\Meff_{k \cdot m(k)})$ in the asymptotic regime. In this setting, the optimal Type-II error probability is expected to decay exponentially with the number of copies $m(k) = \OO(\poly (n))$, i.e., $ \beta_{m(k)}^{\varepsilon}(\Meff_{k \cdot m(k) }) \sim 2^{-\xi^{\mathrm{eff}} m(k) + o(m(k))}$,
where the optimal (asymptotic) error exponent can be expressed as,
\begin{align}   
  \xi^{\mathrm{eff}} = \lim_{\varepsilon \rightarrow 0} \limsup_{ k \rightarrow \infty}  -\frac{1}{m(k)} \log \beta^{\varepsilon}_{m(k)} (\Meff_{ k\cdot m(k) })  \; .
\end{align}

If one adopts the approach from~\cite{YKH+22,MKN+25}, then this exponent can be expressed via the complexity relative entropy, which we define below for completeness.

\begin{definition}[Complexity Relative Entropy~\cite{MKN+25}]\label{def:ComplexRelEnt}
Let $\mathbf{E}^{\mathrm{eff}}_n$ be the set of efficient effect operators on $\HH^{\otimes n}$ that can be generated using at most $p(n)$ elementary gates. Then for any two positive semidefinite operators $\rho,\sigma \in \text{Pos}(\HH^{\otimes n})$ and any $\eta \geq 0$, the \term{complexity relative entropy} is given by:
\begin{align}\label{equ:def:ComplexRelEnt}
  D_{\text{H}}^{p(n),\eta} \left(\rho \|\sigma\right) \coloneqq - \log \inf_{\substack{E \in \mathbf{E}^{\mathrm{eff}}_n \\
  \Tr \left[ E \rho\right]\geq \eta}}  \frac{ \Tr \left[ E \sigma\right]}{ \Tr \left[ E \rho\right]}
  \; .
\end{align}
It then directly follows from~\cite[Proposition 1]{MKN+25} that the optimal type II error is expressed as
\begin{align*}
    \beta^{\varepsilon}_{m(k)} (\Meff_{ k\cdot m(k) }) = \left(1-\varepsilon\right) 2^{-D_{\text{H}}^{p(k\cdot m(k)),1-\varepsilon} (\rho_{k}^{\otimes m(k)} \|\sigma_{k}^{\otimes m(k)})} \; .
\end{align*}
\end{definition}
This yields an optimal error exponent of
\begin{align}   
  \xi^{\mathrm{eff}} &= \lim_{\varepsilon \rightarrow 0} \limsup_{ k \rightarrow \infty} -\frac{\log(1-\varepsilon)}{m(k)} +\frac{D_{\text{H}}^{p(k\cdot m(k)),1-\varepsilon} (\rho_{k}^{\otimes m(k)} \|\sigma_{k}^{\otimes m(k)})}{m(k)} \\
  &= \lim_{\varepsilon \rightarrow 0} \limsup_{ k \rightarrow \infty}  \frac{D_{\text{H}}^{p(k\cdot m(k)),1-\varepsilon} (\rho_{k}^{\otimes m(k)} \|\sigma_{k}^{\otimes m(k)})}{m(k)}\; .
\end{align}
Here, our main contribution is \Cref{thm:qpt-stein}, which shows that the optimal error exponent is 
bounded by the regularized computational measured relative entropy. In doing so, we also establish a connection between the computational measured {\Renyi} divergences of this work and the complexity relative entropy from~\cite{MKN+25}.
\begin{theorem}[Computational Stein's Lemma]\label{thm:qpt-stein}
Let $\{\rho_k\}_{k \in \NN}, \{\sigma_k\}_{k \in \NN} \in \{\States(\HH^{\otimes k})\}_{k \in \NN}$ be two families of quantum states and let $\{\Meff_n\}_{n\in\NN}$ be a family of informationally complete polynomially generated set of binary POVMs. For any $0 < \varepsilon < 1$, the optimal type II error exponent $\xi^{\mathrm{eff}}(\{\rho_k\}_k \|\{\sigma_k\}_k)$ for testing $H_0: \sigma_{k}^{\otimes m(k)}$ against $H_1: \rho_{k}^{\otimes m(k)}$, with $m(k) = \OO(\poly(k))$, using  efficient measurements satisfies:
\begin{align}
\xi^{\mathrm{eff}}(\{\rho_k\}_k \|\{\sigma_k\}_k) \coloneqq  \lim_{\varepsilon \rightarrow 0}\limsup_{k \to \infty} - \frac{1}{m(k)} \log \beta_{m(k)}^{\varepsilon}(\Meff_{k \cdot m(k) }) \leq \CompDiv_{}^{\Meff,\infty}(\{\rho_k\}_k \|\{\sigma_k\}_k) \; ,
\end{align}
where $\beta_{m(k)}^{\varepsilon}(\Meff_{k \cdot m(k)})$ is the minimal type II error given the type I error $\leq \varepsilon$, and $\CompDiv_{}^{\Meff, \infty}(\{\rho_k\}_k \|\{\sigma_k\}_k)$ is the regularized computational measured relative entropy.
\end{theorem}

\begin{proof}
The proof follows from the general result of~\cite{BHLP20}. We adapt it to the polynomially constrained scenario and include it in \Cref{app:Steinproof} for completeness.
\end{proof}
In the simple case, where $\rho,\sigma$ are fixed and one distinguishes $\rho^{\otimes m}$ from $\sigma^{\otimes m}$, this reduces to 
\begin{align}
    \xi^{\mathrm{eff}}(\rho \|\sigma) \coloneqq  \lim_{\varepsilon \rightarrow 0}\limsup_{m \to \infty} - \frac{1}{m} \log \beta_{m}^{\varepsilon}(\Meff_{m }) \leq \CompDiv_{}^{\Meff,\infty}(\rho \|\sigma) \; .
\end{align}
This bound is tight for the original quantum Stein's lemma~\cite{ON00}, which one 
recovers by removing the computational constraints on $\Meff_{k}$. Moreover, it immediately follows from \Cref{lem: regularizeddivergenceineq} that the regularized computational max-divergence is an upper bound on $\xi^{\mathrm{eff}}(\{\rho_k\}_k \|\{\sigma_k\}_k)$.

\subsection{Computational Resource Theory}
\label{sec:comp.res.th}
Properties of quantum systems such as entanglement or magic can be studied
systematically by \term{quantum resource theories}~\cite{CG19}.
A resource theory is a pair
$\RE=(\Ff,\OO)$ where, for any Hilbert spaces $\HH,\HH_1,\HH_2$,
\begin{itemize}
    \item $\Ff(\HH)\subseteq\Ss(\HH)$ is the set of
\emph{free states} on $\HH$, i.e., those that lack the resource and
    \item $\OO(\HH_1,\HH_2)\subseteq \CPTP(\HH_1,\HH_2)$ is the set of
\emph{free operations}.
\end{itemize} 
A free operation, $\Lambda\!\in\!\OO(\HH_{1},\HH_{2})$, must minimally satisfy $\Lambda(\sigma)\!\in\!\Ff(\HH_{2})$ \emph{for every} $\sigma\!\in\!\Ff(\HH_{1})$. When studying, e.g. entanglement as a resource, one chooses $\Ff=\mathrm{SEP}$ (separable states) and $\OO=\mathrm{LOCC}$ (local operations + classical communication).

When restricted to polynomially implementable operations, families of states  with high resource content can become computationally indistinguishable from
low-resource families. This phenomenon, known as pseudoresourced quantum states~\cite{HBE24,BMB+24, GY25} has a direct consequence: when constrained to polynomially implementable operations, information-theoretic resource measures  will generally not be operationally meaningful. For entanglement there exist, e.g., exponential separations between information-theoretic and computational measures~\cite{ABF+23,ABV23,LREJ25}. Nevertheless, to the best of our knowledge, no general measure for polynomially-constrained resources has been proposed in the literature.

There exist several measures for quantifying resources in quantum information theory, ranging from geometric to witness-based approaches. Typically, these measures assess the resource content of a state by evaluating its “distance” from the set of free states. In this section, we focus on entropic measures, defining the computational counterpart of the \term{measured relative entropy of resource} and the \term{max-relative entropy of resource}. Throughout this section, $\OO_n(\HH_1,\HH_2)$ will denote the free operations realized by poly-size circuits, i.e., the ``efficient" free operations from the Hilbert space $\HH_1$ to $\HH_2$. In the same way, $\mathcal{F}_n\equiv\Ff(\HH^{\otimes n})$ represents the set of free states, i.e. the ones without the resource. 

\begin{definition}[Set of Efficient Free Quantum Operations]\label{def:eff.free.op}
Given $\{\mathbf{E}_n^{\text{eff},(i)}\}_{n\in\NN},$ where \\ ${\mathbf{E}_n^{\text{eff},(i)}\subset\mathcal{B}(\HH_i^{\otimes n})}$ are families of informationally complete polynomially generated sets of effect operators on, respectively systems $i=1,2$, with possibly different polynomials. On the sets of free states $\Ff(\HH_1^{\otimes n})$ and $\Ff(\HH_2^{\otimes n})$, a CPTP map  $\Lambda_{n}\in\mathcal{O}(\mathcal{H}^{\otimes n}_1,\mathcal{H}^{\otimes n}_2)$  is said to be an \term{efficient free quantum operation} if the following two conditions hold.
\begin{enumerate}
    \item \textbf{State condition:} $\Lambda_n(\sigma) \in \Ff(\HH_2^{\otimes n})$ for all $\sigma \in\Ff(\HH_1^{\otimes n})$.
    \item \textbf{Test stability:} $  \Lambda_n^*({\mathbf{E}}^\text{eff}_{n}(\HH_2^{\otimes {n}}))\subseteq \overline{\mathbf{E}}^\text{eff}_{n}(\HH_1^{\otimes {n}})$.
\end{enumerate}
\end{definition}
Therefore, the set of free operations can be seen as quantum channels that map free states into free states, not generating the studied resource. The main difference with respect to the information-theoretic case lies in the added complexity constraint: these transformations cannot generate more complexity either. Let us now define computational resource distance measures through our proposed computational divergences.

\begin{definition}[Computational Measured Relative Entropy of Resource] 
Given a quantum state $\rho \in \Ss(\HH^{\otimes n})$, its \term{computational measured relative entropy of resource} is given by,
    \begin{align}
       \CompDiv_{\Ff_{n}} ( \rho): =  \min_{\sigma\in \Ff_n} \CompDiv^{\Meff_n} (\rho\|\sigma)\, .
    \end{align}
\end{definition}

\begin{definition}[Computational Max-Relative Entropy of Resource]\label{def:max.entropy.resource}
Given a quantum state $\rho \in \Ss(\HH^{\otimes n})$, its \term{computational max-relative entropy of resource} is given by,
    \begin{align}
    \CompDiv_{\mathrm{max},\Ff_{n}} ( \rho): =  \min_{\sigma \in \Ff_n} \CompMaxDiv (\rho \|\sigma) \, .
    \end{align}
\end{definition}

In the information-theoretic setting, the (unmeasured) relative entropy of resource and the smoothed max-relative entropy have standard operational roles: the former captures asymptotic hypothesis-testing/interconversion rates, while the latter characterizes one-shot formation/distillation (e.g., entanglement cost and distillable entanglement) under non-entangling or catalytic transformations~\cite{D09,BD10,LBT19,RFW19,RW20}. In the computational setting, a full operational characterization remains open. Nevertheless, by parallelism with the unconstrained case, we adopt their measured and restricted counterparts. Moreover, they satisfy the crucial properties of resource monotones under our allowed maps, as we now prove. 

\begin{lem}\label{lem:prop.res}
Assuming that the set of free states $\Ff_n=\Ff(\HH^{\otimes n})$ is closed, the computational resource measures $\CompbbDiv \in \{\CompDiv, \CompMaxDiv\}$ that we have defined satisfy the following properties,
\begin{itemize}
    \item \textbf{Faithfulness}: given a quantum state $\rho \in \Ss(\HH^{\otimes n})$, its resource measure $\CompDiv_{\Ff_n}(\rho)$ is zero iff $\rho\in\Ff(\HH^{\otimes n})$.

    \item \textbf{Monotonicity: } For any free operation $\Lambda_n \in \OO_n(\HH_1^{\otimes n},\HH_2^{\otimes n})$,
    \begin{align}
        \CompbbDiv_{\Ff_{}(\HH_2^{\otimes n})} (\Lambda_n (\rho)) \leq \CompbbDiv_{\Ff_{}(\HH_1^{\otimes n})} ( \rho) \, .
    \end{align}
\end{itemize}
\end{lem}

\begin{proof} 
On the one hand, if $\rho\in\Ff(\HH^{\otimes n})$, then it naturally holds that
\begin{align}
    \CompbbDiv_{\Ff_{n}} ( \rho) =0 \; .
\end{align}
The other direction which is necessary to prove faithfulness then directly follows for $\CompDiv_{\max}$ from \Cref{lem:comp.max.norm} and for the measured relative entropy measures from \Cref{lem:pinsker} together with the fact that the computational trace distance $\CompTrDis$ is a norm and the set $\Ff_n(\HH^{\otimes n})$ is closed.
For the monotonicity property, it follows from the state condition of \Cref{def:eff.free.op} that,
\begin{align*}
     \CompbbDiv_{\Ff_{}(\HH_2^{\otimes n})} (\Lambda_n (\rho)) = \hspace{-0.2cm}\min_{\tau \in \Ff_{}(\HH_2^{\otimes n})} \CompbbDiv^{\Meff_n} (\Lambda_n (\rho)\| \tau) \leq \hspace{-0.2cm}\min_{\sigma \in \Ff_{}(\HH_1^{\otimes n})} \CompbbDiv^{\Meff_n} (\Lambda_n (\rho)\| \Lambda_n(\sigma)) \leq \CompbbDiv^{\Meff_n} (\Lambda_n (\rho)\| \Lambda_n(\sigma)) \,,
\end{align*}
for $\CompbbDiv^{\Meff_n} \in \{\CompDiv^{\Meff_n}, \CompMaxDiv\}$ and any $\sigma \in \Ff(\HH_1^{\otimes n})$, since $\Lambda_n(\Ff(\HH_1^{\otimes n})) \subseteq \Ff(\HH_2^{\otimes n})$. Then, by the test stability property of \Cref{def:eff.free.op} together with Property  4 of~\Cref{lem: MeasRenDivProp} (see also~\cite[Lemma 5]{RSB24}) it follows that, for all $\sigma \in \Ff(\HH_1^{\otimes n})$,
\begin{align*}
    \CompbbDiv^{\Meff_n} (\Lambda_n (\rho)\| \Lambda_n(\sigma)) \leq \CompbbDiv^{\Meff_n} (\rho\| \sigma) \,.
\end{align*}
The proof then concludes by taking the minimum.
\end{proof}

For the case of the computational measured relative entropy, it follows from~\cite[Theorem 11]{SW23} that it a satisfies asymptotic continuity with respect to the trace norm. Concretely, for any convex set $\Ff_n$ and assuming $\frac{1}{2} \|\rho-\rho'\|_1 \leq \varepsilon$ it then holds that
\begin{align}
    \left|\CompDiv_{\Ff_{n}} ( \rho)-\CompDiv_{\Ff_{n}} ( \rho')\right|\leq g(\varepsilon) +  \varepsilon \sup_{M \in \Meff_n} \kappa_M\,
\end{align}
where $\kappa_M =\sup_{\mu,\nu}\left| D_{\Ff_n}^M(\mu) - D_{\Ff_n}^M(\nu)\right|$ with $D_{\Ff_n}^M(\rho)\coloneqq\inf_{\sigma \in \Ff_n}D(\MM_M(\rho)\|\MM_M(\sigma))$ is the maximal absolute difference of the measured relative entropy of resource between any two states and $g(\varepsilon)\coloneqq (1+\varepsilon)h\left( \frac{\varepsilon}{1+\varepsilon}\right)$ and $h(\cdot)$ is the binary entropy. 

Since we demonstrate an exponential separation between the information-theoretic trace distance and its computational counterpart, continuity bounds phrased in the computational metric can be exponentially sharper. This makes the above bound in practice less useful and motivates a tighter relation between our computational divergences and the corresponding restricted trace distance—so that the right-hand side of the inequality depends only on a computational notion of distinguishability. \Cref{th:comptfannes} provides exactly this: an asymptotic continuity bound for the computational measured relative entropy of resource in terms of the computational trace distance, furnishing a direct analytic bridge with clear operational meaning. To the best of our knowledge, this is the first continuity bound that relates a measurement-restricted divergence to its matching restricted trace distance, and it is likely of independent interest.

\begin{theorem}[Asymptotic Continuity Bound for the Computational Measured Relative Entropy of Resource]
\label{th:comptfannes}
Let $\Ff_n\subset\Ss(\HH^{\otimes n})$ be closed, convex, and bounded containing a full rank state $\sigma^\star$.
For $\rho,\rho' \in\Ss(\HH^{\otimes n})$ and  $\CompTrDis(\rho,\rho')= \frac{1}{2}\|\rho-\rho'\|_{\Meff_n}\le\varepsilon$ it holds that, 
\begin{align}
\big|\CompDiv_{\Ff_n}(\rho)-\CompDiv_{\Ff_n}(\rho')\big|
\leq  (1+\varepsilon)h\left(\frac{\varepsilon}{1+\varepsilon}\right) + \varepsilon \left(\kappa + \log\left(\frac{2}{\varepsilon}\right)\right)\, ,
\end{align}
where $h(\cdot)$ is the binary entropy and $\kappa\coloneqq-\log(\lambda_{\textup{min}}(\sigma^*)) >0$ is the negative logarithm of the smallest eigenvalue of the full rank state $\sigma^\star$.
\end{theorem}

\begin{proof}
In order to end up with the desired bound for $\CompDiv_{\mathcal{F}_n}$, we will first swap the infimum and supremum implicitly appearing in $\CompDiv_{\mathcal{F}_n}$. We do this by using~\cite[Lemma 13]{BHLP20}, see also~\cite[Lemma 10]{SW23}, which, however, a priori does not directly apply to our situation, since the set of measurement maps $\mathcal{M}$ corresponding to binary POVMs $\overlineMeff_n$ is not closed under flagged finite convex combinations. We say a set $\mathscr{M}$ of CPTP maps is closed under flagged finite convex combinations~\cite{SW23, BHLP20} if for any finite probability distribution $\{p_i\}_{i=1}^m$ and any $\mathcal{M}_i\in\mathscr{M}$ it follows that $\bigoplus_{i=1}^mp_i\mathcal{M}_i\in\mathscr{M}$. So as a first step we define 
\begin{align*}
    \mathscr{M}^\text{eff}_n\coloneqq\left\{\bigoplus_{i=1}^mp_i\mathcal{M}_i=\sum_{i=1}^mp_i\mathcal{M}_i\otimes|i\rangle\langle i|\,\Bigg|\,m\in\NN,\, p_i\geq 0,\, \sum_{i=1}^mp_i=1,\, \mathcal{M}_i\equiv \mathcal{M}_{M_i}, \, M_i\in\overlineMeff_n\right\} \; ,
\end{align*}
where here $\mathcal{M}_i\equiv\mathcal{M}_{M_i}$ is the measurement channel implementing the binary POVM $M_i\in\overlineMeff_n$.
This set of flagged measurement channels is now closed under flagged finite convex combinations by construction, and so it follows by~\cite[Lemma 10]{SW23} that
\begin{align*}
\inf_{\sigma\in\mathcal{F}_n}\sup_{\mathcal{M}\in\mathscr{M}^\text{eff}_n}D(\mathcal{M}(\rho)\|\mathcal{M}(\sigma)) = \sup_{\mathcal{M}\in\mathscr{M}^\text{eff}_n}\inf_{\sigma\in\mathcal{F}_n}D(\mathcal{M}(\rho)\|\mathcal{M}(\sigma))\;.
\end{align*}
Importantly by expanding our set of measurement channels we do not change our objective values. It holds that
\begin{align} \label{equ:proof2}
    \sup_{\mathcal{M} \in \mathscr{M}^\text{eff}_n} D (\MM(\rho)\|\MM(\sigma)) = \sup_{M \in \overlineMeff_n} D (\MM_M(\rho)\|\MM_M(\sigma))\,.
\end{align}
To prove this one easily checks that for a channel $\mathcal{M}=\bigoplus_{i=1}^mp_i\mathcal{M}_i$ it holds that
$ D(\MM(\rho)\|\MM(\sigma))  = \sum_{i=1}^mp_i D(\MM_i(\rho)\|\MM_i(\sigma))$ which is an average and hence the claimed Eq.~\eqref{equ:proof2} follows. Completely analogously one shows that
\begin{align}\label{eq:equivalence.distances}
\sup_{\mathcal{M}\in\mathscr{M}_n^\text{eff}}\|\mathcal{M}(\rho)-\mathcal{M}(\rho^\prime)\|_1=\sup_{M\in\overlineMeff_n}\|\mathcal{M}_M(\rho)-\mathcal{M}_M(\rho^\prime)\|_1=2\CompTrDis(\rho,\rho^\prime)\;.
\end{align}
Now putting these observations together we get
\begin{align}\label{eq:minimax}
   \big|\CompDiv_{\Ff_n}(\rho)-\CompDiv_{\Ff_n}(\rho')\big| &=  \Big| \min_{\sigma\in \Ff_n}\sup_{M \in \overlineMeff_n} D (\MM_M(\rho)\|\MM_M(\sigma))-\min_{\sigma\in \Ff_n}\sup_{M \in \overlineMeff_n} D (\MM_M(\rho')\|\MM_M(\sigma))\Big|\notag \\
   &=  \Big| \min_{\sigma\in \Ff_n}\sup_{\mathcal{M} \in \mathscr{M}^\text{eff}_n} D (\MM(\rho)\|\MM(\sigma))-\min_{\sigma\in \Ff_n}\sup_{\mathcal{M} \in \mathscr{M}^\text{eff}_n} D (\MM(\rho')\|\MM(\sigma))\Big| \notag \\
   &= \Big| \sup_{\mathcal{M} \in \mathscr{M}^\text{eff}_n} \min_{\sigma\in \Ff_n} D (\MM(\rho)\|\MM(\sigma))-\sup_{\mathcal{M} \in \mathscr{M}^\text{eff}_n}\min_{\sigma\in \Ff_n} D (\MM(\rho')\|\MM(\sigma))\Big|\notag \\
   &\leq \sup_{\mathcal{M} \in \mathscr{M}^\text{eff}_n} \Big|  \min_{\sigma\in \Ff_n} D (\MM(\rho)\|\MM(\sigma))-\min_{\sigma\in \Ff_n} D (\MM(\rho')\|\MM(\sigma))\Big|\notag \\
   & =\sup_{\mathcal{M} \in \mathscr{M}^\text{eff}_n} \Big| \min_{\sigma\in \Ff_n} \sum^m_{i=1}p_i D (\MM_{i}(\rho)\|\MM_{i}(\sigma))-\min_{\sigma\in \Ff_n} \sum_{i=1}^mp_iD (\MM_{i}(\rho')\|\MM_{i}(\sigma))\Big|\; , \qquad
\end{align}
where in the last line the supremum is over all $\mathcal{M}=\bigoplus_{i=1}^mp_i\mathcal{M}_i$, where $\mathcal{M}_i$ is a measurement channel corresponding to the binary POVM $M_i\in\overlineMeff_n$.
Before we now continue to bound the last term of Eq.~\eqref{eq:minimax} ---inspired by the proof of~\cite[Lemma 7]{Win16}--- we introduce some notation. Fix some $\mathcal{M}=\bigoplus_{i=1}^mp_i\mathcal{M}_i\in\mathscr{M}_n^\text{eff}$ and define 
    $\delta\coloneqq\frac{1}{2}\|\mathcal{M}(\rho)-\mathcal{M}(\rho^\prime)\|_1$, $ 
    \delta_i\coloneqq\frac{1}{2}\|\mathcal{M}_i(\rho)-\mathcal{M}_i(\rho^\prime)\|_1$,
 then it follows that
 \begin{align*}
  \delta=\sum_ip_i\delta_i\leq \CompTrDis(\rho,\rho^\prime)   
 \end{align*}
and that $\eta_i\coloneqq\frac{1}{\delta_i}\left(\mathcal{M}_i(\rho)-\mathcal{M}_i(\rho^\prime)\right)_+$ and $
    \eta_i^\prime \coloneqq \frac{1}{\delta_i}\left(\mathcal{M}_i(\rho^\prime)-\mathcal{M}_i(\rho)\right)_+$
 are both valid quantum states, where $(A)_+$ is the projection onto the positive part of $A$. We can now bound
\begin{align*}
    \mathcal{M}_i(\rho)&=\mathcal{M}_i(\rho^\prime)+(\mathcal{M}_i(\rho)-\mathcal{M}_i(\rho^\prime)) \\
    &\leq \mathcal{M}_i(\rho^\prime)+ \delta_i\eta_i \\ &= (1+\delta_i)\left(\frac{1}{1+\delta_i}\mathcal{M}_{i}(\rho^\prime)+\frac{\delta_i}{1+\delta_i}\eta_{i}\right) \\ &=:(1+\delta_i)\omega_i \; ,
\end{align*} to get that
\begin{align*}
    \omega_i = \frac{1}{1+\delta_i}\mathcal{M}_i(\rho^\prime)+\frac{\delta_i}{1+\delta_i}\eta_i = \frac{1}{1+\delta_i}\mathcal{M}_{i}(\rho)+\frac{\delta_i}{1+\delta_i}\eta^\prime_i \; .
\end{align*}

Since the relative entropy is jointly convex, we have that $\min_{\sigma\in\mathcal{F}_n}\sum_{i}p_i(1+\delta_i)D(\cdot\|\mathcal{M}_i(\sigma))$ is a convex functional and hence
\begin{align*}
    \min_{\sigma\in\mathcal{F}_n}\sum_ip_i(1+\delta_i)D(\omega_i\|\mathcal{M}_i(\sigma))\leq \min_{\sigma\in\mathcal{F}_n}\sum_ip_iD(\mathcal{M}_i(\rho')\|\mathcal{M}_i(\sigma))+\min_{\sigma\in\mathcal{F}_n}\sum_ip_i\delta_iD(\eta_i\|\mathcal{M}_i(\sigma)) \; ,
\end{align*} which rearranged gives
\begin{align}
\min_{\sigma\in\mathcal{F}_n}\sum_ip_iD(\mathcal{M}_i(\rho^\prime)\|\mathcal{M}_i(\sigma))\geq \min_{\sigma\in\mathcal{F}_n}\sum_ip_i(1+\delta_i)D(\omega_i\|\mathcal{M}_i(\sigma))-\min_{\sigma\in\mathcal{F}_n}\sum_ip_i\delta_iD(\eta_i\|\mathcal{M}_i(\sigma))\;.
\end{align}

On the other hand, for any $\sigma \in \Ff_n$ we have that,
\begin{align*}
    D(\omega_i\| \MM_i(\sigma)) &= -S(\omega_i)-\Tr[\omega_i\log(\MM_i(\sigma))] \\ 
    & \geq -h\left(\frac{\delta_i}{1+\delta_i}\right)- \frac{1}{1+\delta_i} S(\MM_i(\rho))- \frac{\delta_i}{1+\delta_i} S(\eta'_i)\\ 
    &\qquad- \frac{1}{1+\delta_i} \Tr[\MM_i(\rho)\log (\MM_i(\sigma)))] -  \frac{\delta_i}{1+\delta_i}\Tr[\eta'_i\log (\MM_i(\sigma))]\\ 
    &= -h\left(\frac{\delta_i}{1+\delta_i}\right) + \frac{1}{1+\delta_i} D(\MM_i(\rho)\|\MM_i(\sigma)) + \frac{\delta_i}{1+\delta_i} D(\eta'_i\|\MM_i(\sigma)) \, ,
\end{align*}
where the first inequality is from ~\cite[Eq.~(1)]{Win16} and follows from the strong sub-additivity of the von Neumann entropy. 
Multiplying with $p_i(1+\delta_i)$, then rearranging the terms, and subsequently taking the minimum yields
\begin{align}
\min_{\sigma\in\mathcal{F}_n}\sum_ip_iD(\MM_i(\rho)\|\MM_i(\sigma)) &\leq \min_{\sigma\in\mathcal{F}_n}\left(\sum_ip_i(1+\delta_i) D(\omega_i\| \MM_i(\sigma))
+\sum_ip_i(1+\delta_i)h\left(\frac{\delta_i}{1+\delta_i}\right) \right. \notag\\
& \qquad \left.-\sum_ip_i\delta_iD(\eta'_i\|\MM_i(\sigma)) \right) \notag\\
&\leq \min_{\sigma\in\mathcal{F}_n}\sum_ip_i(1+\delta_i) D(\omega_i\| \MM_i(\sigma))\notag
    \\
    &\qquad +\sum_ip_i(1+\delta_i)h\left(\frac{\delta_i}{1+\delta_i}\right)-\min_{\sigma\in\mathcal{F}_n}\sum_ip_i\delta_iD(\eta'_i\|\MM_i(\sigma)) \; .
\end{align}
Putting both inequalities together we get that, 
\begin{align}
  \min_{\sigma\in \Ff_n} \sum_ip_i D (\MM_{i}(\rho)\|\MM_{i}(\sigma))&-\min_{\sigma\in \Ff_n} \sum_ip_iD (\MM_{i}(\rho')\|\MM_{i}(\sigma)) \leq  \sum_ip_i(1+\delta_i) h\left(\frac{\delta_i}{1+\delta_i}\right)\notag  \\ &+\min_{\sigma\in\mathcal{F}_n}\sum_ip_i\delta_iD(\eta_i\|\mathcal{M}_i(\sigma))-\min_{\sigma\in\mathcal{F}_n}\sum_ip_i\delta_iD(\eta_i^\prime\|\mathcal{M}_i(\sigma)) \; .
\end{align}
Observing that swapping $\rho\leftrightarrow \rho^\prime$ implies $\eta\leftrightarrow\eta^\prime$ and thus effectively a minus term in the above inequality where $\rho,\rho^\prime,\eta,\eta^\prime$ yields the same inequality with absolute values. Then, we get

\begin{align}\label{equ:proof1}
    &\Big| \min_{\sigma\in \Ff_n} \sum_ip_i D (\MM_{i}(\rho)\|\MM_{i}(\sigma))-\min_{\sigma\in \Ff_n} \sum_ip_iD (\MM_{i}(\rho')\|\MM_{i}(\sigma))\Big|\notag  \\
     &\qquad \leq  \sum_ip_i \left((1+\delta_i) h\left(\frac{\delta_i}{1+\delta_i}\right) \right)  + \left|\min_{\sigma\in \Ff_n} \sum_i p_i\delta_i D(\eta_i \|\MM_i(\sigma))- \min_{\sigma\in \Ff_n} \sum_i p_i \delta_i D(\eta_i' \|\MM_i(\sigma))\right| \notag \\ 
    &\qquad= \sum_ip_i \left((1+\delta_i) h\left(\frac{\delta_i}{1+\delta_i}\right) \right)  + \Bigg| \min_{\sigma\in \Ff_n} \sum_i p_i \delta_i D\left(\frac{(\MM_i(\rho)-\MM_i(\rho'))_+}{\delta_i}\Big\|\MM_i(\sigma)\right)\notag \\
    &\hspace{6cm}-\min_{\sigma\in \Ff_n} \sum_i p_i \delta_i D\left(\frac{(\MM_i(\rho')-\MM_i(\rho))_+  }{\delta_i}\Big \|\MM_i(\sigma)\right)\Bigg| \notag\\ 
    &\qquad\leq  \sum_ip_i \left((1+\delta_i) h\left(\frac{\delta_i}{1+\delta_i}\right) \right)  +  \kappa_{\mathscr{M}} \, ,
\end{align}
where in the second equality we substituted by the definitions of $\eta_i, \eta'_i$  and, lastly, defined 
\begin{align*}
    \kappa_{\mathscr{M}} \coloneqq \sup_{\tau, \tau'\in \States(\HH^{\otimes n})}\Bigg[&\min_{\sigma\in \Ff_n} \sum_i p_i \delta_i D\left(\frac{(\MM_i(\tau)-\MM_i(\tau'))_+}{\delta_i}\big\|\MM_i(\sigma)\right) \\ &\qquad-\min_{\sigma\in \Ff_n} \sum_i p_i \delta_i D\left(\frac{(\MM_i(\tau')-\MM_i(\tau))_+  }{\delta_i}\big \|\MM_i(\sigma)\right)\Bigg]\;,
\end{align*}
which doesn't require absolute values since exchanging $\tau\leftrightarrow \tau^\prime$ gives a global minus sign.

We bound this last term in the following way,
\begin{align}
    \kappa_{\mathscr{M}} &= \sup_{\tau, \tau'\in \States(\HH^{\otimes n})}\Bigg[\min_{\sigma\in \Ff_n} \sum_i p_i \delta_i D\left(\frac{(\MM_i(\tau)-\MM_i(\tau'))_+}{\delta_i}\big\|\MM_i(\sigma)\right) \notag\\
    & \qquad-\min_{\sigma\in \Ff_n} \sum_i p_i \delta_i D\left(\frac{(\MM_i(\tau')-\MM_i(\tau))_+  }{\delta_i}\Bigg \|\MM_i(\sigma)\right)\Bigg]\notag\\
    & \leq \sup_{\tau, \tau'}\min_{\sigma \in \Ff_n}\sum_i p_i \delta_i D\left(\frac{(\MM_i(\tau)-\MM_i(\tau'))_+}{\delta_i}\Bigg\|\MM_i(\sigma)\right)\notag\\
    & \leq \sup_{\tau, \tau'}\sum_i p_i \delta_i D\left(\frac{(\MM_i(\tau)-\MM_i(\tau'))_+}{\delta_i}\Bigg\|\MM_i(\sigma^\star)\right)\notag \\
    &\leq -\sum_ip_i \delta_i\log\lambda_\text{min}(\MM_i(\sigma^\star))\notag
    \\
    &\leq \kappa \cdot \delta +\sum_ip_i \delta_i \log\left(\frac{2}{\delta_i}\right),
\end{align}
where in the second inequality we upper bounded by plugging in the full rank state $\sigma^\star$. The third inequality follows from the fact that $\tilde{\tau}_i=\frac{(\MM_i(\tau)-\MM_i(\tau'))_+}{\delta_i}$ is a valid quantum state and hence the inequality $D(\tilde{\tau}_i\|\MM_i(\sigma^\star))\leq-\log\lambda_{\text{min}}(\MM_i(\sigma^\star))$ holds.
The fourth inequality above stems from the following lower bound
\begin{align}
    \lambda_{\text{min}}(\MM_i(\sigma^\star)) &=\min\{\Tr[E_i\sigma^\star],\Tr[(\id-E_i)\sigma^\star]\} \notag\\ &\geq \lambda_\text{min}(\sigma^\star)\cdot\min\{\Tr[E_i],\Tr[\id-E_i]\} \notag\\
    &\geq \lambda_{\text{min}}(\sigma^\star)\cdot\frac{\delta_i}{2},
\end{align}
where $M_i=(E_i,\id-E_i)$ and
where the last inequality itself follows from the following observation. We have
\begin{align*}
  \delta_i &= |\Tr[E_i(\rho-\rho')]|\leq \|E_i\|_{\infty}\|\rho-\rho'\|_1 \leq 2\Tr[E_i], \\
  \delta_i &= |\Tr[(\id-E_i)(\rho-\rho')]|\leq \|\id-E_i\|_{\infty}\|\rho-\rho'\|_1 \leq 2\Tr[\id-E_i],
\end{align*} which imply that $\min\{\Tr[E_i],\Tr[\id-E_i]\} \geq \frac{\delta_i}{2}$.
Now setting $\kappa\coloneqq-\log\lambda_{\text{min}}(\sigma^\star)$, recalling $\delta = \sum_i p_i \delta_i$ and applying the logarithm yields the desired inequality.

Now putting everything together we have that,
\begin{align}
     \big|\CompDiv_{\Ff_n}(\rho)-\CompDiv_{\Ff_n}(\rho')\big| 
     &  \leq \sup_{\mathcal{M} \in \mathscr{M}^\text{eff}_n} \sum_i p_i \underbrace{\left[   (1+\delta_i)h\left(\frac{\delta_i}{1+\delta_i}\right)  + \delta_i\left(\kappa +\log\left(\frac{2}{\delta_i}\right) \right) \right]}_{=:f(\delta_i)} \notag\\ &\equiv \sup_{\mathcal{M} \in \mathscr{M}^\text{eff}_n} \sum_i p_if(\delta_i)\notag \\ &\leq \sup_{\mathcal{M} \in \mathscr{M}^\text{eff}_n}f\left(\sum_ip_i\delta_i\right) = f(\delta) \leq f(\varepsilon) \notag\\
     & = (1+\varepsilon)h\left(\frac{ \varepsilon}{1+ \varepsilon}\right)  + \varepsilon \left(\kappa +\log\left(\frac{2}{ \varepsilon}\right) \right) ,
\end{align}
where we have used concavity and monotonicity, since $f(\delta_i)\coloneqq (1+\delta_i)h\left(\frac{\delta_i}{1+\delta_i}\right)  + \delta_i\left(\kappa +\log\left(\frac{2}{\delta_i}\right) \right): [0,1]\to [0,3+\kappa]$ is a positive, continuous, concave, and monotonically increasing function.
\end{proof}

Moreover, as discussed in~\cite{Win16}, $\kappa$ is guaranteed to be finite since $\Ff_n$ contains at least one full-rank state. For example, in the case of the entanglement resource, $\Ff_n = \mathrm{SEP}_n\ni \frac{\id}{2^n} $ and so $\kappa = n$.
Therefore, for the \term{computational measured relative entropy of resource} another property is fulfilled: \term{computational faithfulness}. 

\begin{cor}
Let $\Ff_n\subset\Ss(\HH^{\otimes n})$ be closed, convex, and bounded containing a full-rank state.
Given two families of states  $\{\rho_n\}_{n\in\NN}$ and $\{\rho'_n\}_{n\in\NN}$ which are computationally indistinguishable, then, 
    \begin{align}
        \big|\CompDiv_{\Ff_n}(\rho_n)-\CompDiv_{\Ff_n}(\rho'_n)\big| \leq \negl(n) \,
    \end{align}
\end{cor}

\begin{proof}
    It directly follows from \Cref{th:comptfannes} assuming that $\CompTrDis(\rho_k,\rho'_n) \leq \negl(n)$.
\end{proof}

The operational implications of this equation are as follows: if a family of states $\{\rho_n\}_{n\in\NN}$ with high (information-theoretic) resource content is computationally indistinguishable from a family of free states $\{\sigma_n\}_{n\in\NN}$ with $\sigma_n\in\Ff_n$, then the computational measured relative entropy of resource of $\{\rho_n\}_{n\in\NN}$ is negligible, i.e., $\CompDiv_{\Ff_n}(\rho_n)=\negl(n)$. This phenomenon has been studied before in the context of pseudoentanglement~\cite{ABF+23,ABV23,LREJ25} and other pseudoresources~\cite{HBE24,BMB+24, GY25} from a more heuristic perspective: if a computational distinguisher is able to distill more resource from one of the families, then it can distinguish them. Here, we provide an analytical expression that quantifies this for every resource.

\subsection{Computational Entanglement Measures}\label{sec:entanglement}

Let us now focus on the resource of entanglement. Analogously to the information-theoretic scenario, we define the computational measured relative entropy of entanglement as
\begin{align}
    \CompE_R(\rho) \coloneqq \CompDiv_{\Sep_n}(\rho)
    = \min_{\sigma\in \Sep_n}\, \CompDiv^{\Meff_n}(\rho\Vert\sigma) \; ,
\end{align}
where we take $\mathcal{F}_n = \Sep_n$ to be the set of separable states, i.e., 
\begin{align}
    \Sep_n\coloneqq\conv\{\rho^n_A\otimes \rho^n_B \, : \, \rho^n_A \in \Ss(\HH^{\otimes n}_A), \rho^n_B \in \Ss(\HH^{\otimes n}_B)\}
\end{align}
Let us again show an explicit separation between the computational measures and its information-theoretic counterparts using the explicit cryptographic construction of \Cref{sec:example},
\begin{align*}
    \psi^{AB}_n &\coloneqq \frac{1}{4}\left(\ketbra{\Phi^+}{\Phi^+}^{AB}+ \ketbra{\Phi^-}{\Phi^-}^{AB}\right)\otimes \left(\rho^A_{0,n}+\rho^A_{1,n}\right)\, , \\
    \phi^{AB}_n &\coloneqq \frac{1}{2}\left(\ketbra{\Phi^+}{\Phi^+}^{AB}\otimes\rho^A_{0,n} + \ketbra{\Phi^-}{\Phi^-}^{AB}\otimes \rho^A_{1,n}\right)\,,
\end{align*}
where $\ket{\Phi^+}^{AB} = \frac{1}{\sqrt{2}}( \ket{00}^{AB}+ \ket{11}^{AB})$, $\ket{\Phi^-}^{AB} = \frac{1}{\sqrt{2}}( \ket{00}^{AB}-\ket{11}^{AB})$   and  $\{\rho^A_{0,n},\rho^A_{1,n}\}_{n \in \NN}$ are EFI pairs. One can easily see that $E_R(\phi^{AB}_n)\approx1$ by construction. On the other hand, as stated in \Cref{sec:example}, $\CompTrDis(\psi^{AB}_n, \phi^{AB}_n) \leq \negl(n)$. Then, by \Cref{th:comptfannes}, it follows that $\CompE_R(\phi^{AB}_n) \le \negl(n)$. Moreover, this explicit gap can again be (polynomially) amplified as proven in~\cite[Lemma 5.1]{GE24}.

To relate the \emph{computational measured relative entropy of entanglement} to the \emph{(one-shot) computational entanglement cost} and the \emph{(one-shot) computational distillable entanglement} defined by~\cite{ABV23}, we first fix the class of efficient free operations we work with. In entanglement theory, local operations and classical communication (LOCC) form the canonical class of free maps: they cannot increase entanglement, and they directly model a practical two-party communication setting. 

\begin{definition} [LOCC Map~\cite{ABV23}] A quantum channel is said to be an LOCC map 
\begin{equation*}
    \Gamma : \HH_A \otimes\HH_B \mapsto \HH_{\bar{A}} \otimes \HH_{\bar{B}}
\end{equation*}
if it can be implemented by a two-party interactive protocol where each party can implement arbitrary local quantum computations and the two parties can exchange arbitrary classical communication. 
\end{definition}
We now consider their complexity-constrained counterparts.
\begin{definition} [Circuit Description of an LOCC Map~\cite{ABV23}]
Given an LOCC map $\Gamma$, its circuit description is given by two families of circuits $\{\cC_{A,i}\}_{i \in \{1,..., r\}}$ and $\{\cC_{B,i}\}_{i \in \{1,..., r\}}$, each of them acting on $n_A +  t_A + c$ and $n_B +  t_B + c$ qubits respectively, such that the following procedure implements the map $\Gamma$ on an arbitrary input $\varphi_{AB} \in (\CC^2)^{\otimes n_A} \otimes (\CC^2)^{\otimes n_B}$ : 

\begin{enumerate}
    \item Registers $A$ and $B$ of $n_A $ and $n_B$ qubits are initialized in the state $\varphi_{AB}$. Ancilla registers $A'$ and $B'$ of $t_A$ and $t_B$ qubits are initialized in the $\ket{0}$ state. Communication register $C$ is also initialized in the $\ket{0}$ state.

    \item For $i = 1,..., r$, the circuit $\cC_{A,i}$ is applied to registers $A$, $A'$ and $C$. Then, register $C$ is measured in the computational basis. Next, the circuit $\cC_{B,i}$ is applied to registers $B$, $B'$, and $C$. Lastly, register $C$ is measured in the computational basis.

    \item The final output of the LOCC map in the subspace $\HH_{A'}$ corresponds to the state in the registers $A$ and $A'$. In the same way, the state in the registers $B$ and $B'$ is the final output on $\HH_{\bar{B}}$.
\end{enumerate}
A family of LOCC maps $\{\Gamma_{n}\}_{n\in \NN}$ is said to be efficient if there exists a polynomial $c$ such that for all $n$, $\Gamma_{n}$ has a circuit description whose total number of gates, including the ancilla creation and qubit measurements, is at most $c(n)$.
\end{definition}
Let us take the class of efficient LOCC as the set of free operations as in \Cref{def:eff.free.op}. It is easy to see that, by construction, it satisfies the state condition: LOCC operations map separable states to separable states. Moreover, by the efficiency condition, if an input state is efficiently preparable, then so is the output, since composing a polynomial-size state-preparation circuit with a polynomial-size LOCC circuit remains polynomial. 
\begin{rem}
 For the test stability property, we consider efficient LOCC channels $\Gamma_n \, : \, \Ss(\HH^{\otimes n}_1 )\rightarrow \Ss(\HH_{2,n})$, where $\dim\HH_{2,n}\in2^{\poly(n)}$ with gate complexity $c(n)$, and an output-side efficient effect operator $E \in \mathbf{E}^\text{eff}_n(\HH_{2, n})$ with complexity cost $p(n)$.\footnote{While this kind of mapping is not directly captured by \Cref{def:eff.free.op}, it can easily be extended to this form.} Hence, the composed effect operator $\Gamma^*(E)$ belongs to the set of efficient effects $\mathbf{E}^\text{eff}_n(\HH_1^{\otimes n})$ if the associated complexity cost $q(n)$ of the set $\{\mathbf{E}^\text{eff}_n(\HH_1^{\otimes n})\}_{n \in \NN}$ satisfies,
\begin{align*}
    q(n) \geq p(n) + c(n) \, ,
\end{align*}
for all $n \in \NN$. Since the set $\{\mathbf{E}^\text{eff}_n(\HH_1^{\otimes n})\}_{n \in \NN}$ can be generated in polynomial time, without loss of generality we assume from now on that both conditions of \Cref{def:eff.free.op} are satisfied by the set of efficient LOCC operations. As such, the computational measured relative entropy of entanglement satisfies the monotonicity property under efficient LOCC operations.
\end{rem}

With efficient LOCC fixed as our class of free operations, we are now ready to introduce the computational entanglement measures defined by~\cite{ABV23}. Operationally, they encode how many Bell states $\Phi\coloneqq\ket{\Phi_1}\bra{\Phi_1}$ one can distill from or, respectively, are required to produce a given state using only efficient LOCC operations, where the Bell state is defined as $ \ket{\Phi_1}\coloneqq\frac{1}{\sqrt{2}}(\ket{00}+\ket{11})$. Moreover, it satisfies that $\Phi^{\otimes n}=\ket{\Phi_n}\bra{\Phi_n}$, where $ \ket{\Phi_n}\coloneqq\frac{1}{\sqrt{2^n}}\sum_{i=0}^{2^n-1}|i\rangle|i\rangle$.

\begin{definition}[Computational One-Shot Entanglement Cost~\cite{ABV23}]\label{def:ent.cost}
Let $\varepsilon : \NN_+ \rightarrow [0,1]$ and $n\in \NN_+$. Fix polynomial functions $n_A, n_B: \NN_+ \rightarrow \NN_+$. Let $\{\rho_{AB}^{n}\}_{n} $ be a family of quantum states such that, for any $n\geq 1$, $\rho_{AB}^{n} \in \HH_A \otimes \HH_B$ is a bipartite state on $n_A(n) + n_B(n)$. The function $k : \NN \rightarrow \NN$ is an upper bound on the computational entanglement cost of the family $\{\rho_{AB}^{n}\}_{n} $, i.e. $\hat{E}_C^{\varepsilon} (\{\rho_{AB}^{n}\}_{n}) \leq k(n)$, if there exists an efficient LOCC map family $\{\Gamma_{n}\}_{n}$ such that, for each $n\geq 1$, $\Gamma_{n}$ takes an input $k(n)$ maximally entangled $d$-dimensional states, and
\begin{equation*}
     1 - \text{F} ( \rho^{n}_{AB} , \Gamma_{n}(\Phi^{\otimes k (n) })) \leq \varepsilon ({n}) \, .
\end{equation*} 
\end{definition}

\begin{definition} [Computational One-Shot Distillable Entanglement~\cite{ABV23}] \label{def:ent.dist}
Let $\varepsilon : \NN_+ \rightarrow [0,1]$ and $n\in \NN_+$. Fix polynomial functions $n_A, n_B: \NN_+ \rightarrow \NN_+$. Let $\{\rho_{AB}^{n}\}_{n} $ be a family of quantum states such that, for any $n\geq 1$, $\rho_{AB}^{n} \in \HH_A \otimes \HH_B$ is a bipartite state on $n_A(n) + n_B(n)$. The function $m : \NN \rightarrow \NN$ is a lower bound on the computational distillable entanglement of the family $\{\rho_{AB}^{n}\}_{n} $, i.e. $\hat{E}_D^{\varepsilon} (\{\rho_{AB}^{n}\}_{n}) \geq m(n)$, if there exists an efficient LOCC map family $\{\Gamma_{n}\}_{n}$ such that, for each $n\geq 1$, $\Gamma_{n}$ outputs a $2m(n)-$qubit state, and 
\begin{equation*}
   1 - \text{F} (  \Gamma_{n}(\rho^{n}_{AB}) ,\Phi^{\otimes m (n)}) \leq \varepsilon ({n})   \,.
\end{equation*}
   
\end{definition}

Let us first evaluate $\CompE_R$ on maximally entangled states before comparing it with the previously defined computational entanglement measures.

\begin{lem} Assume that the Bell projector ${\Phi_n}=\ket{\Phi_n}\bra{\Phi_n}={\Phi_1}^{\otimes n}\in \Eeff_{2n}$ is efficiently generated.\footnote{For example, any set of efficient effects $\Eeff_n$ which is generated by the gateset $\mathcal{G}\supset\{H,CNOT\}$ and the super-additive polynomial $p(n)=n^2$.}
Then it follows that 
\begin{align*}
    \CompE_R ( \Phi^{\otimes n}) = n\, ,
\end{align*}
where $\Phi^{\otimes n}\coloneqq\ket{\Phi_n}\bra{\Phi_n} \in \BB\left((\CC_A^2)^{\otimes n} \otimes(\CC_B^2)^{\otimes n}\right) $.
\end{lem}
Note that the condition ${\Phi}^{\otimes n}\in \Eeff_{2n}$ is implied by the much simpler one of $\Phi\in \Eeff_{2}$ whenever the polynomial in the construction of the family $\{\Eeff_n\}_{n\in\NN}$ is (super)-additive.

\begin{proof}
We first show that $\CompE_R ( \Phi^{\otimes n}) \geq n$. Since $\Phi$ can be efficiently generated, the measurement $M_{{\Phi_n}}=({\Phi_n}, \mathbb{I} - {\Phi_n})\in \mathbf{M}^{\mathrm{eff}}_{2n}$. 
Then, since $\langle\Phi_n | \sigma | \Phi_n\rangle \leq \frac{1}{2^n}$ for any $\sigma \in \Sep_n$, we have $D(\mathcal{M}_{M_{{\Phi_n}}}(\Phi_n) \,\|\, \mathcal{M}_{M_{{\Phi_n}}}(\sigma)) \geq n$ for any $\sigma \in \Sep_n$.
Let us prove this in detail for completeness. 
Due to convexity of the map $\sigma\mapsto\langle\Phi_n|\sigma|\Phi_n\rangle$ it suffices to consider product states $\sigma=\rho_A\otimes\rho_B$. Now by a standard Choi argument we have
\begin{align}
    \langle\Phi_n|\rho_A\otimes\rho_B|\Phi_n\rangle&=2^{-n}\Tr\left[\Big(\sum_{i,j=0}^{2^n-1}|i\rangle\langle j|_A\otimes|i\rangle\langle j|_B\Big)(\rho_A\otimes\rho_B)\right] \\ &= 2^{-n}\Tr\left[\rho_A\rho_B^T\right]\leq 2^{-n}\|\rho_A\|_2\|\rho_B^T\|_2\leq 2^{-n},
\end{align} where in the last line we used Cauchy-Schwartz inequality and normalization of the states.
\begin{align*}
    \CompE_R ( \Phi_n) &= \min_{\sigma\in \Sep_n} \sup_{M \in \overlineMeff_n}\CompDiv (\MM_M(\Phi_n)\|\MM_M(\sigma)) \\
    &\geq\min_{\sigma\in \Sep_n}D(\mathcal{M}_{M_{{\Phi_n}}}(\Phi_n) \,\|\, \mathcal{M}_{M_{{\Phi_n}}}(\sigma)) \\
    & =\min_{\sigma\in \Sep_n}(- \log \langle\Phi_n|\sigma|\Phi_n\rangle) \geq n\,.
\end{align*}

We now show $\CompE_R ( \Phi_n) \leq n$. By definition,
\begin{align*}
    \CompE_R ( \Phi_n) &= \min_{\sigma\in \Sep_n} \sup_{M \in \overlineMeff_n}\CompDiv (\MM_M(\Phi_n)\|\MM_M(\sigma)) \\
    &\leq \sup_{M \in \overlineMeff_n}\CompDiv (\MM_M(\Phi_n)\|\MM_M(\sigma^*))\, ,
\end{align*}
where $\sigma_n^* = 2^{-n} \sum_{x\in \{0,1\}^n} \ketbra{x}{x}_A \otimes \ketbra{x}{x}_B$ which is separable. Lastly, by data processing,
\begin{align*}
     \sup_{M \in \overlineMeff_n}\CompDiv (\MM_M(\Phi_n)\|\MM_M(\sigma^*))\leq D(\Phi_n \| \sigma^*) = -\langle \Phi_n |\log \sigma^* | \Phi_n\rangle = n \,.
\end{align*}
\end{proof}
Therefore, in this simple case of effectively not assuming computational restrictions, our measure attains the expected maximum. Let us now also relate $\CompE_R$ to the computational entanglement cost and computational distillable entanglement proposed by~\cite{ABV23}, assuming that the Bell projector $\Phi_n \in \Eeff_{2n}$ is efficiently preparable.

\begin{lem}\label{lem: comp.ent.example}
Assume as before that the effect operator $\Phi^{\otimes n}\in\Eeff_{2n}$ is efficiently generated. Given a family of states $\{\rho_{AB}^{n} \}_{n} \subset \BB\!\left((\mathbb{C}_A^2)^{\otimes n}\otimes(\mathbb{C}_B^2)^{\otimes n}\right)$, it follows that
    \begin{align*}
        \CompE_D(\rho_{AB}^{n}) &\leq \CompE_R(\rho_{AB}^{n} ) + g\!\left(\sqrt{\varepsilon(n)},  n\right),\\
        \CompE_C(\rho_{AB}^{n}) &\geq \CompE_R(\rho_{AB}^{n} ) - g\!\left(\sqrt{\varepsilon(n)},  n\right),
    \end{align*}
where  $\CompE_D(\rho_{AB}^{n} )$ and $\CompE_C(\rho_{AB}^{n} )$ are, respectively, the computational distillable entanglement (\Cref{def:ent.dist}) and the computational entanglement cost (\Cref{def:ent.cost}), and
\[
g\!\left(\sqrt{\varepsilon(n)}, n\right)\coloneqq(1+\sqrt{\varepsilon(n)})\, h\!\left(\frac{\sqrt{\varepsilon(n)}}{1+\sqrt{\varepsilon(n)}}\right) + \sqrt{\varepsilon(n)}\left(n + \log\frac{2}{\sqrt{\varepsilon(n)}}\right).
\]
\end{lem}
\begin{proof}
We first show the bound for the computational distillable entanglement. By \Cref{def:ent.dist}, $\CompE_D(\rho^n_{AB}) \ge m(n)$ if there exists a family of efficient LOCC maps $\{{\Gamma}_{n}\}_{n}$ such that, for each $n\in\NN$,
\begin{align*}
    1-F\left(\eta_{AB}, \Phi^{\otimes m(n)}\right) \le \varepsilon(n),\qquad \eta_{AB}\coloneqq{\Gamma}_{n}(\rho^n_{AB}).
\end{align*}
Using the (information-theoretic) Fuchs–van de Graaf and the data-processing inequality, this also implies that $\CompTrDis(\eta_{AB}, \Phi^{\otimes m(n)}) \le \sqrt{\varepsilon(n)}$. Then, by \Cref{th:comptfannes}, it follows that
\begin{align*}
    \bigl|\CompE_R(\Phi^{\otimes m(n)})-\CompE_R(\eta_{AB})\bigr|\le g\!\left(\sqrt{\varepsilon(n)}, n\right).
\end{align*}
Since, by \Cref{lem: comp.ent.example} $\CompE_R(\Phi^{\otimes m(n)}) = m(n)$, and since $\CompE_R(\eta_{AB}) \le \CompE_R(\rho^n_{AB})$ by the monotonicity property under efficient maps, we obtain $m(n)\le \CompE_R(\rho^n_{AB}) + g(\sqrt{\varepsilon(n)},n)$. Taking the supremum over all achievable $m(n)$ yields the stated bound on $\CompE_D$. 

The computational entanglement cost bound follows analogously. We have that $\CompE_C(\rho^n_{AB}) \le k(n)$, if there exists a family of efficient LOCC maps $\{{\Gamma}_{n}\}_{n}$ such that, for each $n\in\NN$,
\begin{align*}
    1-F\left(\rho^n_{AB}, \eta_{AB}\right) \le \varepsilon(n),\qquad \eta_{AB}\coloneqq{\Gamma}_{n}(\Phi^{\otimes k(n)}).
\end{align*}
Similarly, by Fuchs–van de Graaf and DPI, $\CompTrDis\left(\rho^n_{AB}, \eta_{AB}\right) \le \sqrt{\varepsilon(n)}$. Then, by \Cref{th:comptfannes} the following bound holds,
\begin{align*}
    \bigl|\CompE_R(\eta_{AB})-\CompE_R(\rho^n_{AB})\bigr|\le g\!\left(\sqrt{\varepsilon(n)}, n\right).
\end{align*}
By monotonicity under efficient maps, $\CompE_R(\eta_{AB}) \le \CompE_R(\Phi^{\otimes k(n)})$. Moreover, by \Cref{lem: comp.ent.example} it holds that $\CompE_R(\Phi^{\otimes k(n)}) = k(n)$. Therefore, it follows that $k(n) \geq \CompE_R(\rho^n_{AB}) - g\!\left(\sqrt{\varepsilon(n)}, n\right)$. Then, taking the infimum over all feasible $k(n)$ yields to the desired inequality.
\end{proof}

Consequently, $\CompE_R$ is sandwiched between the computational distillable entanglement and the computational entanglement cost, up to the continuity term $g(\sqrt{\varepsilon(n)},n)$. This parallels the information-theoretic hierarchy $E_D \le E_R \le E_C $ and confirms that $\CompE_R$ behaves as an entanglement measure in the computational setting.

\begin{rem}
For the computational max–entropy of entanglement defined in \Cref{def:max.entropy.resource}, a direct comparison with the one–shot measures of~\cite{ABV23} is not straightforward. In the unconstrained regime, the max–relative entropy of entanglement (equivalently, the log-robustness) attains a one–shot operational meaning after smoothing with respect to non-entangling operations, a class that strictly contains LOCC~\cite{Da09b,BD10}. Analogously, smoothing the computational max–entropy we define and taking the efficient free operations to be (efficient) non-entangling maps may clarify its relation to the measures in~\cite{ABV23}.
\end{rem}

\section{Discussion and Open Questions}

We examine central information-theoretic quantities and naturally incorporate computational constraints, building a cohesive framework. Our guiding principle is to retain the clear operational meaning of the original quantities while enforcing that the measurements are efficient. 

\subsection{Discussion}

On the one hand, we take a geometric approach to incorporate computational constraints into the max-divergence. We construct families of proper cones of polynomially generated effect operators, which we define to be generated by efficiently approximate effect operators using non-uniform polynomial circuits. The computational max-divergence then naturally arises~\cite{RKW11, GC24}, induced by the partial order given by those cones. On the other hand, we define the set of binary POVMs composed by the elements of our polynomially generated effect operators. This set allows us to define a computational notion of measured Rényi divergences, based on constructions from~\cite{RSB24,MH23}. 
We show that these two, a priori, quite different approaches to computational divergences are consistent in the sense that in the limit $\alpha\to\infty$, the computational measured Rényi divergence coincides with the computational max-divergence (see \Cref{thm:measured=conic}). This equivalence is not only an analogue of the information-theoretic case, but also shows the consistency of our defined divergences and highlights the usefulness of conic theory to computational quantum information. Inspired by~\cite{RKW11}, we further consider computational notions of the  fidelity and  trace distance. These quantities are related in meaningful --- yet straightforward --- ways to our computational divergences, such as via a computational Pinsker inequality (\Cref{lem:pinsker}). Moreover, they are related to each other via a computational Fuchs–van de Graaf inequality (\Cref{lem:comp.fuchs.van.Graaf}). These relations reflect the robustness and flexibility of our framework, and allow us to provide an explicit example for which these quantities differ from their information-theoretic counterparts in \Cref{sec:example}.

Lastly, we consider two relevant applications. Firstly, in the context of asymmetric binary hypothesis testing, we show that the regularized version of our computational measured relative entropy  has a sharp operational meaning as an upper bound on the Stein-exponent, \Cref{thm:qpt-stein}. In other words, it upper-bounds the error exponent of asymmetric binary hypothesis testing under our computationally restricted measurements. This result is a direct analog to the information–theoretic setting.

Secondly, in the context of resource theories, these computational divergences allow us to define different resource measures. The main result of that section is an asymptotic continuity bound for the computational measured relative entropy, \Cref{th:comptfannes}. 
This bound is the computational analogue of~\cite{Win16,SW23} and we believe it to be of independent interest as it is, as far as we are aware, the first asymptotic continuity bound relating a measurement-restricted divergence with its corresponding restricted trace distance. This bound allows us to state the following analytical characterization: computationally indistinguishable families of states will also have computationally indistinguishable resource contents. This phenomenon had been previously stated in a slightly more heuristic manner, under cryptographic reductions. Moreover, focusing on entanglement, \Cref{th:comptfannes} lets us compare the computational measured relative entropy of entanglement $\CompE_R$ induced by our divergences with the computational distillable entanglement and entanglement cost of~\cite{ABV23}: under efficient LOCC (up to a continuity term), one has $\CompE_D(\rho^n_{AB}) \le \CompE_R(\rho^n_{AB}) \le \CompE_C(\rho^n_{AB})$, thereby reproducing the information-theoretic hierarchy.

Apart from these applications we consider, we expect this formalism to be relevant across various areas of computational quantum information theory. In part, this is because our quantities directly incorporate a central practical constraint --- scalability --- on which manipulations and measurements remain implementable as the systems grow. Complexity and practical polynomial constraints are already relevant in many-body physics, where tools dealing with the exponentially growing Hilbert-space dimension as a function of the particle number have proven very valuable~\cite{Has07,PVWC07,CPSV21}. Only transformations with at most polynomial cost are asymptotically viable, as made precise by Hamiltonian-complexity results~\cite{AL11,GHLS15}. On the measurement side, compressed-acquisition schemes such as classical shadows estimate observables with polynomial resources, reinforcing the value of complexity-aware information measures~\cite{HKP20,KGK25}. Related ideas surface in quantum thermodynamics, where resource costs are analyzed under explicit complexity restrictions~\cite{MKN+25}, and in high-energy theory, where “complexity = action/volume” links boundary complexity to bulk geometry in AdS/CFT~\cite{BRS+16,BS18}. Lastly, our approach interfaces naturally with cryptography, where indistinguishability is explicitly computational; see e.g. pseudorandom states~\cite{JLS18} and pseudorandom unitaries~\cite{MPMY24}.

\subsection{Open questions}

A variety of interesting and promising avenues for future research follow from this work. It is a mathematically natural question to extend our formalism beyond the binary setting to multi-outcome POVMs, keeping implementability and classical post-processing explicit. This would broaden the scope of our work and connect it to operational tasks such as multi–hypothesis testing and multi-party decoupling. This direction is particularly appealing since it opens the possibility of defining a computational min–entropy within our framework and allows one to study its connection to previously studied computational entropies, such as the HILL entropies from~\cite{CCL+17} or the computational min–/max–entropies defined in~\cite{avidan2025quantum,avidan2025fully}. 

On the technical side, a key open problem is to establish nontrivial lower bounds on the computational Stein exponent. From an operational viewpoint, it would be valuable to clarify the role of our computational divergences in one-shot resource tasks under efficient free operations. Pursuing this direction—including smoothed variants of the computational max divergence and (efficient) non-entangling maps—may illuminate further connections between the measures introduced here and prior work in computational entanglement theory~\cite{ABV23,LREJ25}.

\section*{Acknowledgments}
We thank Johannes Jakob Mayer for coordinating simultaneous submission to the arXiv.

We would like to thank Rotem Arnon, Omar Fawzi, and Alex B. Grilo for helpful feedback on early versions of this manuscript. We would additionally like to thank Joseph M. Renes and Noam Avidan for helpful discussions. 
ÁY is supported by the European Union's Horizon Europe Framework Programme under the Marie Sklodowska Curie Grant No. 101072637, Project Quantum-Safe Internet (QSI). TAH is supported by the Koshland Research Fund and by the Air Force Office of Scientific Research under award number FA9550-22-1-0391.
\bibliography{References}

\newcommand{\etalchar}[1]{$^{#1}$}
\begin{thebibliography}{dCBdV{\etalchar{+}}24}

\bibitem[AA25]{avidan2025quantum}
Noam Avidan and Rotem Arnon.
\newblock Quantum computational unpredictability entropy and quantum leakage resilience, 2025.

\bibitem[Aar16]{Aar16}
Scott Aaronson.
\newblock The complexity of quantum states and transformations: From quantum money to black holes, 2016.

\bibitem[ABF{\etalchar{+}}23]{ABF+23}
Scott Aaronson, Adam Bouland, Bill Fefferman, Soumik Ghosh, Umesh Vazirani, Chenyi Zhang, and Zixin Zhou.
\newblock Quantum pseudoentanglement, 2023.

\bibitem[ABV23]{ABV23}
Rotem {Arnon-Friedman}, Zvika Brakerski, and Thomas Vidick.
\newblock {Computational Entanglement Theory}, 2023.

\bibitem[ACMT{\etalchar{+}}07]{ACM+06}
K.~M.~R. Audenaert, J.~Calsamiglia, R.~Muñoz-Tapia, E.~Bagan, Ll. Masanes, A.~Acin, and F.~Verstraete.
\newblock Discriminating states: The quantum chernoff bound.
\newblock {\em Physical Review Letters}, 98(16), April 2007.

\bibitem[AD13]{AD13}
Scott Aaronson and Andrew Drucker.
\newblock A full characterization of quantum advice, 2013.

\bibitem[AE11]{AL11}
Dorit Aharonov and Lior Eldar.
\newblock On the complexity of commuting local hamiltonians, and tight conditions for topological order in such systems, 2011.

\bibitem[AHRA25]{avidan2025fully}
Noam Avidan, Thomas~A. Hahn, Joseph~M. Renes, and Rotem Arnon.
\newblock Fully quantum computational entropies, 2025.

\bibitem[BaG15]{BG15}
Fernando G. S.~L. Brand\~ao and Gilad Gour.
\newblock Reversible framework for quantum resource theories.
\newblock {\em Phys. Rev. Lett.}, 115:070503, Aug 2015.

\bibitem[BaHLP14]{BHLP20}
Fernando~G.S.L. Brand\~{a}o, Aram~W. Harrow, James~R. Lee, and Yuval Peres.
\newblock Adversarial hypothesis testing and a quantum stein's lemma for restricted measurements.
\newblock In {\em Proceedings of the 5th Conference on Innovations in Theoretical Computer Science}, ITCS '14, page 183–194, New York, NY, USA, 2014. Association for Computing Machinery.

\bibitem[Bar15]{Ba15}
Siddharth Barman.
\newblock Approximating nash equilibria and dense subgraphs via an approximate version of carath\'{e}odory's theorem, 2015.

\bibitem[BCQ23]{BCQ23}
Zvika Brakerski, Ran Canetti, and Luowen Qian.
\newblock On the computational hardness needed for quantum cryptography.
\newblock In Yael~Tauman Kalai, editor, {\em 14th Innovations in Theoretical Computer Science Conference (ITCS 2023)}, volume 251 of {\em Leibniz International Proceedings in Informatics (LIPIcs)}, pages 24:1--24:21, Dagstuhl, Germany, 2023. Schloss Dagstuhl--Leibniz-Zentrum für Informatik.

\bibitem[BD11]{BD10}
Fernando G. S.~L. Brandao and Nilanjana Datta.
\newblock One-shot rates for entanglement manipulation under non-entangling maps, 2011.

\bibitem[BEM{\etalchar{+}}23]{bostanci2023unitary}
John Bostanci, Yuval Efron, Tony Metger, Alexander Poremba, Luowen Qian, and Henry Yuen.
\newblock {Unitary Complexity and the Uhlmann Transformation Problem}, 2023.

\bibitem[Ber09]{Bernstein2009}
Daniel~J. Bernstein.
\newblock {\em Introduction to post-quantum cryptography}, pages 1--14.
\newblock Springer Berlin Heidelberg, Berlin, Heidelberg, 2009.

\bibitem[BHN{\etalchar{+}}15]{BHH+15}
Fernando Brandão, Michał Horodecki, Nelly Ng, Jonathan Oppenheim, and Stephanie Wehner.
\newblock The second laws of quantum thermodynamics.
\newblock {\em Proceedings of the National Academy of Sciences}, 112(11):3275–3279, February 2015.

\bibitem[BMB{\etalchar{+}}24]{BMB+24}
Nikhil Bansal, Wai-Keong Mok, Kishor Bharti, Dax~Enshan Koh, and Tobias Haug.
\newblock Pseudorandom density matrices, 2024.

\bibitem[BMY25]{bostanci2025local}
John Bostanci, Tony Metger, and Henry Yuen.
\newblock {Local transformations of bipartite entanglement are rigid}, 2025.

\bibitem[BRS{\etalchar{+}}16]{BRS+16}
Adam~R. Brown, Daniel~A. Roberts, Leonard Susskind, Brian Swingle, and Ying Zhao.
\newblock Holographic complexity equals bulk action?
\newblock {\em Physical Review Letters}, 116(19), May 2016.

\bibitem[BS18]{BS18}
Adam~R. Brown and Leonard Susskind.
\newblock Second law of quantum complexity.
\newblock {\em Phys. Rev. D}, 97:086015, Apr 2018.

\bibitem[Bus73a]{Bus73a}
P.~J. Bushell.
\newblock Hilbert's metric and positive contraction mappings in a banach space.
\newblock {\em Archive for Rational Mechanics and Analysis}, 52:330--338, 1973.

\bibitem[Bus73b]{Bus73b}
P.~J. Bushell.
\newblock On the projective contraction ratio for positive linear mappings.
\newblock {\em Journal of the London Mathematical Society}, s2-6(2):256--258, February 1973.

\bibitem[BV04]{BV04}
Stephen Boyd and Lieven Vandenberghe.
\newblock {\em Convex Optimization}.
\newblock Cambridge University Press, Cambridge, 2004.
\newblock First edition.

\bibitem[CCL{\etalchar{+}}17]{CCL+17}
Yi-Hsiu Chen, Kai-Min Chung, Ching-Yi Lai, Salil~P. Vadhan, and Xiaodi Wu.
\newblock Computational notions of quantum min-entropy, 2017.

\bibitem[CG19]{CG19}
Eric Chitambar and Gilad Gour.
\newblock Quantum resource theories.
\newblock {\em Reviews of Modern Physics}, 91(2), April 2019.

\bibitem[CPSV21]{CPSV21}
J.~Ignacio Cirac, David {P\'erez-Garc\'{\i}a}, Norbert Schuch, and Frank Verstraete.
\newblock Matrix product states and projected entangled pair states: Concepts, symmetries, theorems.
\newblock {\em Rev. Mod. Phys.}, 93:045003, Dec 2021.

\bibitem[CRF21]{CRF21}
{\'A}ngela Capel, Cambyse Rouz{\'e}, and Daniel~Stilck Fran\c{c}a.
\newblock The modified logarithmic sobolev inequality for quantum spin systems: classical and commuting nearest neighbour interactions, 2021.

\bibitem[Dat09a]{Da09b}
Nilanjana Datta.
\newblock Max- relative entropy of entanglement, alias log robustness, 2009.

\bibitem[Dat09b]{D09}
Nilanjana Datta.
\newblock Min- and max-relative entropies and a new entanglement monotone.
\newblock {\em IEEE Transactions on Information Theory}, 55(6):2816--2826, 2009.

\bibitem[dCBdV{\etalchar{+}}24]{CBV+24}
José~A. de~Carvalho, Carlos~A. Batista, Tiago M.~L. de~Veras, Israel~F. Araujo, and Adenilton~J. da~Silva.
\newblock Quantum multiplexer simplification for state preparation, 2024.

\bibitem[DKLP02]{Dennis_2002}
Eric Dennis, Alexei Kitaev, Andrew Landahl, and John Preskill.
\newblock Topological quantum memory.
\newblock {\em Journal of Mathematical Physics}, 43(9):4452–4505, September 2002.

\bibitem[DW05]{Devetak_2005}
Igor Devetak and Andreas Winter.
\newblock {Distillation of secret key and entanglement from quantum states}.
\newblock {\em Proceedings of the Royal Society A: Mathematical, Physical and Engineering Sciences}, 461(2053):207–235, January 2005.

\bibitem[ECP10]{ECP10}
J.~Eisert, M.~Cramer, and M.~B. Plenio.
\newblock Colloquium: Area laws for the entanglement entropy.
\newblock {\em Rev. Mod. Phys.}, 82:277--306, Feb 2010.

\bibitem[Eve95]{Eve95}
S.~P. Eveson.
\newblock Hilbert’s projective metric and the spectral properties of positive linear operators.
\newblock {\em Proceedings of the London Mathematical Society}, 70:411--440, 1995.

\bibitem[FFF25]{fang2025variational}
Kun Fang, Hamza Fawzi, and Omar Fawzi.
\newblock Variational expressions and uhlmann theorem for measured divergences, 2025.

\bibitem[Fuc96]{Fuc96}
Christopher~A. Fuchs.
\newblock Distinguishability and accessible information in quantum theory, 1996.

\bibitem[FvdG99]{FvG99}
C.A. Fuchs and J.~van~de Graaf.
\newblock Cryptographic distinguishability measures for quantum-mechanical states.
\newblock {\em IEEE Transactions on Information Theory}, 45(4):1216--1227, 1999.

\bibitem[GC24]{GC24}
Ian George and Eric Chitambar.
\newblock Cone-restricted information theory.
\newblock {\em Journal of Physics A: Mathematical and Theoretical}, 57(26):265302, June 2024.

\bibitem[GE24]{GE24}
Manuel Goulão and David Elkouss.
\newblock Pseudo-entanglement is necessary for efi pairs, 2024.

\bibitem[GHLS15]{GHLS15}
Sevag Gharibian, Yichen Huang, Zeph Landau, and Seung~Woo Shin.
\newblock Quantum hamiltonian complexity.
\newblock {\em Foundations and $Trends^\circledR$ in Theoretical Computer Science}, 10(3):159–282, 2015.

\bibitem[Gil10]{Gi10}
Gustavo~L. Gilardoni.
\newblock {On Pinsker's and Vajda's Type Inequalities for Csiszár's $f$ -Divergences}.
\newblock {\em IEEE Transactions on Information Theory}, 56(11):5377--5386, 2010.

\bibitem[Gol08]{Gold08}
Oded Goldreich.
\newblock {\em Computational Complexity: A Conceptual Perspective}.
\newblock Cambridge University Press, Cambridge, UK, 2008.

\bibitem[GY25]{GY25}
Alex~B. Grilo and {\'A}lvaro Y{\'a}ng{\"u}ez.
\newblock Quantum pseudoresources imply cryptography, 2025.

\bibitem[Has07]{Has07}
M~B Hastings.
\newblock An area law for one-dimensional quantum systems.
\newblock {\em Journal of Statistical Mechanics: Theory and Experiment}, 2007(08):P08024–P08024, August 2007.

\bibitem[HBK24]{HBE24}
Tobias Haug, Kishor Bharti, and Dax~Enshan Koh.
\newblock Pseudorandom unitaries are neither real nor sparse nor noise-robust, 2024.

\bibitem[HKP20]{HKP20}
Hsin-Yuan Huang, Richard Kueng, and John Preskill.
\newblock Predicting many properties of a quantum system from very few measurements.
\newblock {\em Nature Physics}, 16(10):1050–1057, June 2020.

\bibitem[HMS23]{haitner23}
Iftach Haitner, Noam Mazor, and Jad Silbak.
\newblock {Incompressiblity and Next-Block Pseudoentropy}.
\newblock In Yael Tauman~Kalai, editor, {\em 14th Innovations in Theoretical Computer Science Conference (ITCS 2023)}, volume 251 of {\em Leibniz International Proceedings in Informatics (LIPIcs)}, pages 66:1--66:18, Dagstuhl, Germany, 2023. Schloss Dagstuhl -- Leibniz-Zentrum f{\"u}r Informatik.

\bibitem[Hol11]{H11}
Alexander Holevo.
\newblock {\em Probabilistic and Statistical Aspects of Quantum Theory}.
\newblock Publications of the Scuola Normale Superiore. Edizioni della Normale Pisa, 2011.

\bibitem[HP91]{HP91}
F.~Hiai and D.~Petz.
\newblock The proper formula for relative entropy and its asymptotics in quantum probability.
\newblock {\em Communications in Mathematical Physics}, 143:99--114, 1991.

\bibitem[IP13]{iyer2013hardness}
Pavithran Iyer and David Poulin.
\newblock {Hardness of decoding quantum stabilizer codes}, 2013.

\bibitem[JLS18]{JLS18}
Zhengfeng Ji, Yi-Kai Liu, and Fang Song.
\newblock Pseudorandom quantum states.
\newblock In Hovav Shacham and Alexandra Boldyreva, editors, {\em Advances in Cryptology -- CRYPTO 2018}, pages 126--152, Cham, 2018. Springer International Publishing.

\bibitem[JW23]{JW23}
Yifan Jia and Michael~M. Wolf.
\newblock Hay from the haystack: Explicit examples of exponential quantum circuit complexity.
\newblock {\em Communications in Mathematical Physics}, 402(1):141–156, May 2023.

\bibitem[KGKB25]{KGK25}
Robbie King, David Gosset, Robin Kothari, and Ryan Babbush.
\newblock Triply efficient shadow tomography.
\newblock {\em PRX Quantum}, 6(1), February 2025.

\bibitem[KKR05]{KKR06}
Julia Kempe, Alexei Kitaev, and Oded Regev.
\newblock The complexity of the local hamiltonian problem, 2005.

\bibitem[Kni95]{Kni95}
E.~Knill.
\newblock Approximation by quantum circuits, 1995.

\bibitem[KQST23]{KQST23}
William Kretschmer, Luowen Qian, Makrand Sinha, and Avishay Tal.
\newblock Quantum cryptography in algorithmica, 2023.

\bibitem[Kre21]{Kre21}
William Kretschmer.
\newblock Quantum pseudorandomness and classical complexity.
\newblock Schloss Dagstuhl – Leibniz-Zentrum für Informatik, 2021.

\bibitem[KW20]{khatri2020principles}
Sumeet Khatri and Mark~M. Wilde.
\newblock Principles of quantum communication theory: A modern approach, 2020.

\bibitem[LBT19]{LBT19}
Zi-Wen Liu, Kaifeng Bu, and Ryuji Takagi.
\newblock One-shot operational quantum resource theory.
\newblock {\em Physical Review Letters}, 123(2), July 2019.

\bibitem[LREJ25]{LREJ25}
Lorenzo Leone, Jacopo Rizzo, Jens Eisert, and Sofiene Jerbi.
\newblock Entanglement theory with limited computational resources, 2025.

\bibitem[McE78]{McEliece1978APK}
Robert~J. McEliece.
\newblock {A public key cryptosystem based on algebraic coding theory}.
\newblock 1978.

\bibitem[MH23]{MH23}
Milan Mosonyi and Fumio Hiai.
\newblock Test-measured rényi divergences.
\newblock {\em IEEE Transactions on Information Theory}, 69(2):1074–1092, February 2023.

\bibitem[MJC{\etalchar{+}}14]{MJC+14}
Lester Mackey, Michael~I. Jordan, Richard~Y. Chen, Brendan Farrell, and Joel~A. Tropp.
\newblock Matrix concentration inequalities via the method of exchangeable pairs.
\newblock {\em The Annals of Probability}, 42(3), May 2014.

\bibitem[MKH{\etalchar{+}}25]{MKN+25}
Anthony Munson, Naga Bhavya~Teja Kothakonda, Jonas Haferkamp, Nicole Yunger~Halpern, Jens Eisert, and Philippe Faist.
\newblock Complexity-constrained quantum thermodynamics.
\newblock {\em PRX Quantum}, 6(1), March 2025.

\bibitem[MLDS{\etalchar{+}}13]{MLDSFT13}
Martin Müller-Lennert, Frédéric Dupuis, Oleg Szehr, Serge Fehr, and Marco Tomamichel.
\newblock On quantum rényi entropies: A new generalization and some properties.
\newblock {\em Journal of Mathematical Physics}, 54(12):122203, 12 2013.

\bibitem[MO14]{Mosonyi_2014}
Milán Mosonyi and Tomohiro Ogawa.
\newblock {Quantum Hypothesis Testing and the Operational Interpretation of the Quantum Rényi Relative Entropies}.
\newblock {\em Communications in Mathematical Physics}, 334(3):1617–1648, December 2014.

\bibitem[MPSY24]{MPMY24}
Tony Metger, Alexander Poremba, Makrand Sinha, and Henry Yuen.
\newblock Simple constructions of linear-depth t-designs and pseudorandom unitaries, 2024.

\bibitem[MRR{\etalchar{+}}25]{MRRLJE25}
Johannes~Jakob Meyer, Asad Raza, Jacopo Rizzo, Lorenzo Leone, Sofiene Jerbi, and Jens Eisert.
\newblock {Computational Relative Entropy}, 2025.
\newblock In preparation.

\bibitem[MWW09]{MWW09}
William Matthews, Stephanie Wehner, and Andreas Winter.
\newblock Distinguishability of quantum states under restricted families of measurements with an application to quantum data hiding.
\newblock {\em Communications in Mathematical Physics}, 291(3):813–843, August 2009.

\bibitem[MY23]{metger2023stateqip}
Tony Metger and Henry Yuen.
\newblock {stateQIP = statePSPACE}, 2023.

\bibitem[NC00]{NC00}
Michael~A. Nielsen and Isaac~L. Chuang.
\newblock {\em Quantum Computation and Quantum Information}.
\newblock Cambridge University Press, Cambridge, UK, 2000.

\bibitem[ON00]{ON00}
T.~Ogawa and H.~Nagaoka.
\newblock Strong converse and stein's lemma in quantum hypothesis testing.
\newblock {\em IEEE Transactions on Information Theory}, 46(7):2428--2433, 2000.

\bibitem[PVWC07]{PVWC07}
D.~{Perez-Garcia}, F.~Verstraete, M.~M. Wolf, and J.~I. Cirac.
\newblock Matrix product state representations, 2007.

\bibitem[Ren05]{renner2005security}
Renato Renner.
\newblock Security of quantum key distribution, 2005.

\bibitem[RFWG19]{RFW19}
Bartosz Regula, Kun Fang, Xin Wang, and Mile Gu.
\newblock One-shot entanglement distillation beyond local operations and classical communication.
\newblock {\em New Journal of Physics}, 21(10):103017, October 2019.

\bibitem[RKW11]{RKW11}
David Reeb, Michael~J. Kastoryano, and Michael~M. Wolf.
\newblock Hilbert’s projective metric in quantum information theory.
\newblock {\em Journal of Mathematical Physics}, 52(8), August 2011.

\bibitem[RLW24]{RLW24}
Bartosz Regula, Ludovico Lami, and Mark~M. Wilde.
\newblock Postselected quantum hypothesis testing.
\newblock {\em IEEE Transactions on Information Theory}, 70(5):3453–3469, May 2024.

\bibitem[RSB24]{RSB24}
Tobias Rippchen, Sreejith Sreekumar, and Mario Berta.
\newblock Locally-measured r\'enyi divergences, 2024.

\bibitem[RW20]{RW20}
Soorya Rethinasamy and Mark~M. Wilde.
\newblock Relative entropy and catalytic relative majorization.
\newblock {\em Physical Review Research}, 2(3), September 2020.

\bibitem[SBM06]{SBM06}
V.V. Shende, S.S. Bullock, and I.L. Markov.
\newblock Synthesis of quantum-logic circuits.
\newblock {\em IEEE Transactions on Computer-Aided Design of Integrated Circuits and Systems}, 25(6):1000–1010, June 2006.

\bibitem[SW23]{SW23}
Joseph Schindler and Andreas Winter.
\newblock Continuity bounds on observational entropy and measured relative entropies.
\newblock {\em Journal of Mathematical Physics}, 64(9), September 2023.

\bibitem[Tom16]{T16}
Marco Tomamichel.
\newblock {\em Quantum Information Processing with Finite Resources}.
\newblock Springer International Publishing, 2016.

\bibitem[Uhl76]{Uhl76}
A.~Uhlmann.
\newblock The “transition probability” in the state space of a $*$-algebra.
\newblock {\em Reports on Mathematical Physics}, 9(2):273--279, 1976.

\bibitem[VAW{\etalchar{+}}19]{PhysRevA.99.032344}
Christophe Vuillot, Hamed Asasi, Yang Wang, Leonid~P. Pryadko, and Barbara~M. Terhal.
\newblock {Quantum error correction with the toric Gottesman-Kitaev-Preskill code}.
\newblock {\em Phys. Rev. A}, 99:032344, Mar 2019.

\bibitem[VW23]{VidickWehner2023}
Thomas Vidick and Stephanie Wehner.
\newblock {\em Introduction to Quantum Cryptography}.
\newblock Cambridge University Press, 1st edition, 2023.

\bibitem[Wat18]{W18}
John Watrous.
\newblock {\em The Theory of Quantum Information}.
\newblock Cambridge University Press, 2018.

\bibitem[Win16]{Win16}
Andreas Winter.
\newblock Tight uniform continuity bounds for quantum entropies: Conditional entropy, relative entropy distance and energy constraints.
\newblock {\em Communications in Mathematical Physics}, 347(1):291–313, March 2016.

\bibitem[Wol21]{Wolf2021QKD}
Ramona Wolf.
\newblock {\em {Quantum Key Distribution: An Introduction with Exercises}}, volume 988 of {\em Lecture Notes in Physics}.
\newblock Springer Cham, 2021.

\bibitem[WWY14]{WWY14}
Mark~M. Wilde, Andreas Winter, and Dong Yang.
\newblock Strong converse for the classical capacity of entanglement-breaking and hadamard channels via a sandwiched r{\'e}nyi relative entropy.
\newblock {\em Communications in Mathematical Physics}, 331(2):593--622, 2014.

\bibitem[Yao82]{4568378}
Andrew~C. Yao.
\newblock {Theory and application of trapdoor functions}.
\newblock In {\em 23rd Annual Symposium on Foundations of Computer Science (sfcs 1982)}, pages 80--91, 1982.

\bibitem[YHKH{\etalchar{+}}22]{YKH+22}
Nicole Yunger~Halpern, Naga B.~T. Kothakonda, Jonas Haferkamp, Anthony Munson, Jens Eisert, and Philippe Faist.
\newblock Resource theory of quantum uncomplexity.
\newblock {\em Physical Review A}, 106(6), December 2022.

\end{thebibliography}
\appendix

\section{Cone of Efficient Binary Measurements}
\subsection{Proof of \Cref{cor:diamondnormbound}}\label{app:dimaondnormbound}
Recall that we want to effectively place a bound on the diamond norm between measurement maps corresponding to binary POVMs, given the operator norm difference between their effect operators. 
\begin{proof}[Proof of \Cref{cor:diamondnormbound}]
\label{app:diamondnormbound}
By \Cref{prop:approxcat} there exists a binary measurement $M^\prime$ such that \\ ${M=(E_1,E_2)}$ and 
${M^\prime=(F_1,F_2)=(\frac{1} {k}\sum_{i=1}^k\,E_{j_i},\id-\frac{1} {k}\sum_{i=1}^k\,E_{j_i})}$ 
satisfy
\begin{align}\label{Appeq: infnormbound}
    \|E_1-F_1\|_\infty=\|E_2-F_2\|_\infty\leq \varepsilon \; . 
\end{align}
Recall that the measurement maps they induce are given by
\begin{align*}
\MM(\rho)=\sum_{i=0}^1\Tr[E_{i+1}\rho] \ketbra{i}{i}, \quad  \MM^\prime(\rho)=\sum_{i=0}^1\Tr[F_{i+1}\rho] \ketbra{i}{i} \; ,
\end{align*}
for any state $\rho \in \Ss(\HH^n)$ (and some choice of $n$). Alternatively, for any state $\rho \in \Ss(\HH^n \otimes \HH^\prime)$ on the extended Hilbert space $\HH^\prime$, $\MM \otimes \Id$ and $\MM^\prime \otimes \Id$ are mappings of the form

\begin{align}
\MM\otimes \Id(\rho)=\sum_{i=0}^1 \ketbra{i}{i} \otimes \tr_{\HH^{\otimes n}}[E_{i+1}\rho]\, :\, \Ss(\HH^{\otimes n} \otimes \HH^\prime) \to \Ss(\CC^2 \otimes \HH^\prime)\; , 
\end{align} where $\MM^\prime \otimes \Id$ is analogous and $\tr$ stands for the partial trace. 
Hence, for these measurements, it must hold that
\begin{align*}
    \|\MM-\MM^\prime\|_\diamond &= \sup_{\rho\in\Ss(\HH^{\otimes n} \otimes \HH^\prime)}\|((\MM-\MM^\prime)\otimes\Id)(\rho)\|_1 \\ &=\sup_{\rho}\Big\|\sum_{i=0}^1\ketbra{i}{i}\otimes\tr_{\HH^{\otimes n}}[(E_{i+1}-F_{i+1})\rho]\Big\|_1 
    \\ &=\sup_{\rho}\sum_{i=0}^1\|\tr_{\HH^{\otimes n}}[(E_{i+1}-F_{i+1})\rho]\|_1 \\ 
    &\leq \sup_{\rho}\sum_{i=0}^1\|(E_{i+1}-F_{i+1})\rho\|_1 \\
    &\leq \sup_{\rho} \sum_{i=0}^1\|E_{i+1}-F_{i+1}\|_\infty \|\rho\|_1 \\
    &\leq 2\varepsilon \; .
\end{align*}
 The first inequality follows from the fact that the partial trace is trace-norm contractive, and the second is an application of Hölder's inequality. The last inequality follows then from Eq.~\eqref{Appeq: infnormbound}.
\end{proof}    

\subsection{Proof of \Cref{lem: comptrnormEeffred}}\label{app:comptrnormEeffred}
Recall that we want to prove the following three equivalences

 \begin{align}
         \|\rho - \sigma\|_{\overlineMeff_n} =  \begin{cases}
2\max_{E\in\mathbf{E}^{\mathrm{eff}}_n}\Tr \left[ E \left(\rho -\sigma\right)\right]\;,\qquad \qquad \qquad \qquad (i) \\
    \max_{M = (E_1,E_2) \in \Meff_n} \sum_{i=1}^2 \left|\Tr[E_i( \rho - \sigma)] \right| \; , \; \quad \;(ii)\\
     \max_{M \in \Meff_n} \|\MM_M(\rho) - \MM_M(\sigma)\|_{1} \,. \qquad \; \qquad \;(iii)
\end{cases}
    \end{align}
Let us prove them separately.
\begin{proof}[Proof of first equality (i)]
By definition of 
\begin{align*}
        \|\rho - \sigma\|_{\overlineMeff_n} =  2\max_{E\in\overlineEeff_n}\Tr \left[ E \left(\rho -\sigma\right)\right] \geq 2\max_{E\in\Eeff_n}\Tr \left[ E \left(\rho -\sigma\right)\right] \; ,
\end{align*}
we get the first inequality. The other inequality follows since the functional $E\mapsto\Tr[E(\rho-\sigma)]$ is linear and thus  convex and so 
\begin{align}
   2\max_{E\in\overlineEeff_n}\Tr \left[ E \left(\rho -\sigma\right)\right] =2\max_{E\in\extr(\overlineEeff_n)}\Tr \left[ E \left(\rho -\sigma\right)\right] \leq 2\max_{E\in\mathbf{E}^{\mathrm{eff}}_n}\Tr \left[ E \left(\rho -\sigma\right)\right].
\end{align}

Combining both inequalities proves the desired result. 
\end{proof}
\begin{proof}[Proof of second equality (ii)]
By the above proven first equivalence, it holds that
    \begin{align*}
        2\max_{E\in\mathbf{E}^{\mathrm{eff}}_n}\Tr \left[ E \left(\rho -\sigma\right)\right]
&=\max_{E\in\mathbf{E}^{\mathrm{eff}}_n}\left|\Tr \left[ E \left(\rho -\sigma\right)\right]\right| + \left|\Tr \left[ -E \left(\rho -\sigma\right)\right] \right| \\
&=\max_{E\in\mathbf{E}^{\mathrm{eff}}_n}\left|\Tr \left[ E \left(\rho -\sigma\right)\right]\right| + \left|\Tr \left[ \left(\id-E\right) \left(\rho -\sigma\right)\right] \right| \\
&=\max_{M = (E_1,E_2) \in \Meff_n}\left|\Tr \left[ E_{1} \left(\rho -\sigma\right)\right]\right| + \left|\Tr \left[ E_{2} \left(\rho -\sigma\right)\right] \right| \\
&= \max_{M = (E_1,E_2) \in \Meff_n} \sum_{i=1}^2 \left|\Tr[E_i( \rho - \sigma)] \right| \;.
\end{align*}
\end{proof}
\begin{proof}[Proof of third equality (iii)]
 For any binary measurement $M = (E_1,E_2) \in \Meff_n$, it further holds that 
    \begin{align*}
        \sum_{i=1}^2 \left|\Tr[E_i( \rho - \sigma)] \right| &= \Big\| \sum_{i=0}^{1}\Tr[E_{i+1}( \rho - \sigma)] \ketbra{i}{i}\Big\|_{1} \\
         &= \Big\|\sum_{i=0}^{1}\Tr[E_{i+1} \rho ] \ketbra{i}{i} -\sum_{i=0}^{1}\Tr[E_{i+1} \sigma ] \ketbra{i}{i}\Big\|_{1} \\
         &= \|\MM_M(\rho) -\MM_M(\sigma)  \|_{1} \; , 
    \end{align*}
Optimizing over $M = (E_1,E_2) \in \Meff_n$ then yields $(iii)$. 
\end{proof}

\subsection{Proof of \Cref{lem:comp.dist.sub}}
Recall that we want to prove that $\|\cdot\|_{\overlineMeff_{n\cdot m}}$ is sub-additive, i.e.
\begin{align*}
        \| \rho^{\otimes m} - \sigma^{\otimes m}\|_{\overlineMeff_{n\cdot m}} \leq m \max_{i \in \{1,\dots, m \}}\| \rho_{i} - \sigma_{i}\|_{\overlineMeff_n} \; ,
    \end{align*}
    where $\rho_{i} = \rho^{\otimes i-1}\otimes\rho \otimes\sigma^{\otimes n-i}$, $\sigma_{i} = \rho^{\otimes i-1}\otimes\sigma \otimes\sigma^{\otimes n-i}$.
\begin{proof}[Proof of \Cref{lem:comp.dist.sub}]
\label{app:comp.dist.sub}
    Note that $\rho_{i-1}=\sigma_{i}$ and $\rho_{m}= \rho^{\otimes m}$, $\sigma_{1}=\sigma^{\otimes m}$. As such
\begin{align*}
    \rho^{\otimes m} - \sigma^{\otimes m} &= \rho_{m} - \sigma_{1} = \sum_{j=1}^{m} \rho_{j} - \sum_{j=1}^{m-1} \rho_{j} - \sigma_{1} =\sum_{j=1}^{m} (\rho_{j} - \sigma_{j})  \; .
\end{align*}
Using this and the  triangle inequality repeatedly now yields the claim
   \begin{align*}
       \| \rho^{\otimes m} - \sigma^{\otimes m}\|_{\overlineMeff_{n\cdot m}} &\leq \sum_{j=1}^{m}  \| \rho_{j} - \sigma_{j}\|_{\overlineMeff_{n}} \leq \sum_{j=1}^{m}  \max_{i \in \{1,\dots, m \}}\| \rho_{i} - \sigma_{i}\|_{\overlineMeff_{n}} \leq m  \max_{i \in \{1,\dots, m \}}\| \rho_{i} - \sigma_{i}\|_{\overlineMeff_{n}} \; .
   \end{align*}
\end{proof}

\section{Computational Quantum Divergences}

\subsection{Proof of \Cref{lem:comp.max.norm}}
Recall, we want to prove that
\begin{align*}
\CompMaxDiv(\rho \| \sigma) \geq \log\left(1 + \CompTrDis \left(\rho,\sigma\right)\right) \; .
\end{align*}
\begin{proof}[Proof of \Cref{lem:comp.max.norm}]
\label{app:comp.max.norm}
If $\CompMaxDiv(\rho \| \sigma) = \infty$, then this trivially holds. Assume that $\CompMaxDiv(\rho \| \sigma) < \infty$.
Hence there exists some $\lambda$ such that $\rho \leq_{\Co^{\mathcal{S_{\mathrm{eff}}}}_n} \lambda \sigma$.  Then, for any such $\lambda$, we have
\begin{align*}
\frac{1}{2} \|\rho - \sigma\|_{\overlineMeff_n} &= \sup_{E \in \overline{\mathbf{E}}^{\mathrm{eff}}_n} \Tr[E(\rho - \sigma)] \\
 &\leq \sup_{E \in \overline{\mathbf{E}}^{\mathrm{eff}}_n} \Tr[E(\lambda\sigma - \sigma)] \\
 &  = (\lambda - 1) \sup_{E \in \overline{\mathbf{E}}^{\mathrm{eff}}_n} \Tr[E \sigma]\\
 &\leq \lambda - 1 \; ,
 \end{align*} where for the inequality we are using that the partial order $\leq_{\Co^{\mathcal{S_{\mathrm{eff}}}}_n}$ is the one induced by the dual cone to $\mathcal{C}^{\Eeff}_n$ to which $\overlineEeff_n$ contains a base.
Rearranging gives, $\lambda \geq 1 + \frac{1}{2} \|\rho - \sigma\|_{\Meff_n}$ and taking the infimum over all such suitable $\lambda$ yields the statement of the lemma, 
\begin{align*}
    \CompMaxDiv(\rho \| \sigma) &\coloneqq \log \inf\{\lambda \in \RR | \rho \leq_{\Co^{\mathcal{S_{\mathrm{eff}}}}_n} \lambda \sigma\} \\
    &\geq \log \left(1 + \frac{1}{2} \|\rho - \sigma\|_{\overlineMeff_n} \right)\\
    &= \log\left(1 + \CompTrDis \left(\rho,\sigma\right)\right) \; .
\end{align*}
\end{proof}
\subsection{Proof of \Cref{lem: MeasRenDivProp}} 
\begin{proof}[Proof of \Cref{lem: MeasRenDivProp}]
\label{app:MeasRenDivProp}
The properties of the computational measured Rényi divergences follows directly from ~\cite[Lemmas 2 \& 5]{RSB24} and~\cite[Chapter 4]{T16}. We will use $\mathbb{D}_{\alpha}(\mu\|\nu)$ for denoting a classical Rényi divergence over probabilities distributions $\mu,\nu$. Recall that, when $\mu_{\rho}^M(i) = \Tr[\rho E_i]$ and $\nu_{\sigma}^M(i) = \Tr[\sigma E_i]$, $\CompDiv_{\alpha}^{\Meff_n}(\rho\|\sigma) = \sup_{M \in \Meff_n}\mathbb{D}_{\alpha}(\mu_{\rho}^M\|\nu_{\sigma}^M)$.\\
(1) Follows from the positivity of the classical Rényi divergence.\\
(2) For every pair of probabilities distributions $\mu,\nu$, we have that for all $\alpha, \beta$ such that $0 < \alpha \leq \beta$,
\begin{align*}
    \mathbb{D}_{\alpha}(\mu\|\nu) \leq \mathbb{D}_{\beta}(\mu\|\nu) \;.
\end{align*}
Therefore, by taking the supremum, it directly follows that $\CompDiv_{\alpha}^{\Meff_{n}}(\rho\|\sigma) \leq  \CompDiv_{\beta}^{\Meff_{n}}(\rho\|\sigma)$.\\
(3) By linearity of the trace, $\mu^M_{\lambda\rho_1 + (1-\lambda) \rho_2} = \lambda \mu^M_{\rho_1} + (1-\lambda)\mu^M_{\rho_2}$. Let us first consider the case $\alpha \in (0,1)$. Therefore,
\begin{align*}
    \CompQ^{\Meff_n}_{\alpha}(\lambda \rho_1 + (1-\lambda)\rho_2\|\lambda \sigma_1 + (1-\lambda)\sigma_2) &= \inf_{M \in \Meff_n}\mathbb{Q}_{\alpha}(\lambda \mu_{\rho_1}^M + (1-\lambda)\mu_{\rho_2}^M\|\lambda \nu_{\sigma_1}^M + (1-\lambda)\nu_{\sigma_2}^M)\\
    & \geq \inf_{M \in \Meff_n}\lambda\mathbb{Q}_{\alpha}(\mu_{\rho_1}^M\|\mu_{\rho_2}^M)+(1-\lambda)\mathbb{Q}_{\alpha}(\nu_{\sigma_1}^M\|\nu_{\sigma_2}^M)\\
    & \geq \lambda \CompQ^{\Meff_n}_{\alpha}(\rho_1\|\rho_2)+(1-\lambda) \CompQ^{\Meff_n}_{\alpha}(\sigma_1\|\sigma_2) \; ,
\end{align*}
where $\CompQ^{\Meff_n}_{\alpha}(\cdot\|\cdot)$ and $\mathbb{Q}_{\alpha}(\cdot\|\cdot)$ are the computational and classical counterparts of Eq.~\eqref{eq:Qrenyi}. The first inequality follows from joint concavity of the classical measure and the last one from the superadditivity of the infimum, proving the joint concavity of the function. Let us also set $f(t)\coloneqq\frac{1}{\alpha-1}\log(t)$ (see \Cref{def:sandwiched.renyi}). In the case of $\alpha \in (0,1)$, $f(t)$ is a non-increasing convex function. Therefore, it follows that $\CompDiv^{\Meff_n}_{\alpha}$ is jointly convex in $(\rho,\sigma)$ for $\alpha \in (0,1)$. $\CompDiv^{\Meff_n}_{1}$ is jointly convex as the supremum of a family of jointly convex functionals (see~\cite[Prop. 3, Lemma 2]{RSB24}. For $\alpha \in (1,\infty)$, $f(t)$ is non-decreasing and quasi-convex. Moreover, by using the fact of $\mathbb{Q}_{\alpha}$ being joint convex and the supremum being superadditive, $\CompQ^{\Meff_n}_{\alpha}$ is joint convex. Therefore, $\CompDiv^{\Meff_n}_{\alpha}$ is jointly quasi-convex. Lastly, the joint quasi convexity of $\CompDiv^{\Meff_n}_{\infty}$ follows by taking the supremum of a family of jointly quasi-convex functions (\cite[Prop. 3, Lemma 2]{RSB24}).\\

(4) By the duality of the quantum channels, $\mu^M_{\EE(\rho)}=\mu^{\EE^* \left( M\right)}_{\rho}$. Since for every admissible $\EE$, it follows that $\EE^* \left( M\right) \in \overlineMeff_{n}$, then,
\begin{align*}
    \CompDiv_{\alpha}^{\Meff_{n}}(\EE(\rho)\|\EE(\sigma)) = \sup_{M \in \Meff_n}\mathbb{D}_{\alpha}(\mu_{\EE(\rho)}^M\|\nu_{\EE(\sigma)}^M) =\sup_{M \in \Meff_n} \mathbb{D}_{\alpha}(\mu_{\rho}^{\EE^\dagger \left( M\right)}\|\nu_{\sigma}^{\EE^\dagger \left( M\right)}) \leq\CompDiv_{\alpha}^{\Meff_{n}}(\rho\|\sigma) \;.
\end{align*}
\end{proof}

\subsection{Proof of \Cref{lem:fid.comp.Renyi}}
Recall that we want to show that
\begin{align*}
    \CompDiv_{1/2}^{\Meff_n}(\rho \|\sigma) = -\log \CompFid(\rho, \sigma) \;.
\end{align*}
\begin{proof}[Proof of \Cref{lem:fid.comp.Renyi}]
\label{app:fid.comp.Renyi}
By definition of the computational measured Rényi divergence we have that for $\alpha = 1/2$,
    \begin{align}
        \CompDiv_{1/2}^{\Meff_n}(\rho \|\sigma) = -2\log \CompQ_{1/2}^{\Meff_n}(\rho \|\sigma)\;.  \label{eq: appdivQrel}
    \end{align}
Then, by Eq.~\eqref{eq:comp.Q.Renyi},
\begin{align*}
     \CompQ^{\Meff_n}_{1/2}(\rho\|\sigma) &=\inf_{M \in \Meff_n} Q_{1/2}(\MM_M(\rho)\|\MM_M(\sigma))\\
     &=\inf_{M \in \overlineMeff_n} Q_{1/2}(\MM_M(\rho)\|\MM_M(\sigma))\\
     &=\inf_{M \in \overlineMeff_n}\|\MM_M(\sigma)^{1/2}\MM_M(\rho)\MM_M(\sigma)^{1/2}\|_{1/2}^{1/2} \\
     &=\inf_{(E_1,E_2)\in \overlineMeff_n}\sum_{i = 1}^{2} \sqrt{\Tr[E_i\rho]\Tr[E_i\sigma]} \\
     &= \sqrt{\CompFid(\rho, \sigma)}\;,
\end{align*}
where the second equality follows from \Cref{lem: equivalanceMandConvM}. Together with Eq.~\eqref{eq: appdivQrel}, this yields the claim.
\end{proof}

\subsection{Proof of \Cref{lem:comp.fuchs.van.Graaf}}
Recall that we want to prove here a computational Fuch-van-Graaf inequality 
\begin{align*}
     1-\sqrt{\CompFid(\rho, \sigma)} \leq \CompTrDis(\rho, \sigma) \leq \sqrt{ 1-\CompFid(\rho, \sigma)}\; .
 \end{align*}
\begin{proof}[Proof of \Cref{lem:comp.fuchs.van.Graaf}]
\label{app:comp.fuchs.van.Graaf}
We note that it suffices to prove that 
    \begin{align}
     \begin{aligned}
        1-\sqrt{\min_{M  \in \Meff_n}F(\MM_M(\rho), \MM_M(\sigma))} \leq \max_{M  \in \Meff_n}\Delta(\MM_M(\rho), \MM_M(\sigma)) \leq \sqrt{ 1-\min_{M  \in \Meff_n}F(\MM_M(\rho), \MM_M(\sigma))} \; , \label{eq: appSuffBoundCond}
        \end{aligned}
    \end{align}
    since it directly implies our bounds.
    Starting with the usual information-theoretic Fuch-van Graaf inequality~\cite{FvG99} applied to the states $\mathcal{M}_M(\rho), \mathcal{M}_M(\sigma)$ for any $M \in \Meff_n$, we have
    \begin{align*}
        1-\sqrt{F(\MM_M(\rho), \MM_M(\sigma))} \leq \Delta(\MM_M(\rho), \MM_M(\sigma)) \leq \sqrt{ 1-F(\MM_M(\rho), \MM_M(\sigma))}\; .
    \end{align*}
Now optimizing over all $M \in \Meff_n$ yields Eq.~\eqref{eq: appSuffBoundCond}.
\end{proof}

\section{Applications}
\subsection{Proof of \Cref{thm:qpt-stein}}
\label{app:Steinproof}
\begin{proof}
The upper bound on the Stein's exponent can be proven in the following way. Let us consider for an $\varepsilon >0$ and $m(k) \in \poly(k) > 0$, a binary measurement $M : = (E_{k \cdot m(k)} , \id -E_{k \cdot m(k)} ) \in \Meff_{k \cdot m(k)}$ that satisfies $\Tr[ E_{k \cdot m(k)}\rho_k^{\otimes m(k)}] \geq 1-\varepsilon$.
Recall that the infimum over all such measurements is the one that defines $\beta^\varepsilon_{m(k)}(\Meff_{k\cdot m(k)})$ in Eq.~\eqref{equ:def:optimaltypeIIerror}.
Then we have
\begin{align*}
    \sup_{M^\prime \in \Meff_{k \cdot m(k)}} & D(\MM_{M^\prime} (\rho_k^{\otimes m(k)}) \| \MM_{M^\prime} (\sigma_k^{\otimes m(k)})) \geq D(\MM_M (\rho_k^{\otimes m(k)}) \| \MM_M (\sigma_k^{\otimes m(k)})) \\
    & = \Tr[ E_{k \cdot m(k)} \rho_k^{\otimes m(k)}]\log \Tr[E_{k \cdot m(k)} \rho_k^{\otimes m(k)}] - \Tr[ E_{k \cdot m(k)} \rho_k^{\otimes m(k)}]\log \Tr[E_{k \cdot m(k)} \sigma_k^{\otimes m(k)}] \\
    & + \Tr[(\id - E_{k \cdot m(k)})\rho_k^{\otimes m(k)}]\log \Tr[(\id -E_{k \cdot m(k)}) \rho_k^{\otimes m(k)}] \\
    &-\Tr[(\id - E_{k \cdot m(k)}) \rho_k^{\otimes m(k)}]\log \Tr[(\id - E_{k \cdot m(k)}) \sigma_k^{\otimes m(k)}] \\
    & \geq -h_2\left(\Tr[ E_{k \cdot m(k)} \rho_k^{\otimes m(k)}]\right) - \Tr[ E_{k \cdot m(k)} \rho_k^{\otimes m(k)}] \log \Tr[ E_{k \cdot m(k)}\sigma_k^{\otimes m(k)}] \\
    & \geq- 1 - (1-\varepsilon)\log \Tr[ E_{k \cdot m(k)} \sigma_k^{\otimes m(k)}]   \;,  
\end{align*}
where $h_2(p) = - p\log(p) - (1-p)\log (1-p)$ is the binary entropy. For the second-to-last line, we use the fact that
\begin{align}
    \Tr[(\id - E_{k \cdot m(k)}) \rho_k^{\otimes m(k)}]\log \Tr[(\id - E_{k \cdot m(k)}) \sigma_k^{\otimes m(k)}] \leq 0 \; .
\end{align}
The last line then follows from the fact that the binary entropy is upper bounded by $1$ and the constraint that $\Tr[ E_{k \cdot m(k)}\rho_k^{\otimes m(k)}] \geq 1-\epsilon$. Rearranging this inequality and subsequently taking the supremum over all such efficient measurements $M$ yields 
  \begin{align*}
    -\log \beta^{\varepsilon}_{m(k)}(\Meff_{k\cdot m(k)})  \leq \frac{1}{1-\varepsilon}\Big(1 + \sup_{M' \in \Meff_{k \cdot m(k)}}D\big(\MM_{M^\prime} (\rho_k^{\otimes m(k)}) \| \MM_{M^\prime} (\sigma_k^{\otimes m(k)})\big)\Big) \; .
\end{align*}  
Now dividing by $m(k)$, and taking the $\limsup$ in $k$, gives
\begin{align*}
    \limsup_{k \rightarrow \infty} - \frac{1}{m(k)}\log \beta^{\varepsilon}_{m(k)}(\Meff_{k\cdot m(k)}) &\leq \frac{1}{1 - \varepsilon} \limsup_{ k \rightarrow \infty} \frac{1}{m(k)}\left(1+ \CompDiv^{\Meff_{k \cdot m(k)}}(\rho_k^{\otimes m(k)} \| \sigma_k^{\otimes m(k)})\right) \\ &=\frac{1}{1 - \varepsilon} \limsup_{ k \rightarrow \infty} \frac{1}{m(k)}\left(\CompDiv^{\Meff_{k \cdot m(k)}}(\rho_k^{\otimes m(k)} \| \sigma_k^{\otimes m(k)})\right) \\
    &= \frac{1}{1-\varepsilon}\CompDiv^{\Meff, \infty}(\{\rho_k\}_k \|\{\sigma_k\}_k)\; .
\end{align*}
Finally taking the limit $\varepsilon \rightarrow 0$, yields the claimed result. 
\end{proof}

\end{document}